\RequirePackage{fix-cm}
\documentclass[twocolumn]{svjour3}          
\smartqed  
\usepackage{graphicx}
\usepackage[linesnumbered,ruled,vlined,noend]{algorithm2e}
\let\oldnl\nl
\newcommand{\nonl}{\renewcommand{\nl}{\let\nl\oldnl}}
\SetNlSty{bfseries}{\color{black}}{}

\usepackage{balance}  
\usepackage{xcolor}
\usepackage{caption}

\usepackage{amssymb,amsmath,amsthm}
\usepackage{setspace}
\usepackage{multirow}

\usepackage{graphics}
\usepackage{mdwlist}
\usepackage[normalem]{ulem}

\usepackage{booktabs}
\usepackage{soul}
\usepackage[center]{subfigure}
\usepackage[export]{adjustbox}
\usepackage{bm}
\usepackage{bbm}
\usepackage{url}
\usepackage{bbold}

\begin{document}

\title{Hypergraph motifs and their extensions beyond binary\thanks{This work was supported by National Research Foundation of Korea (NRF) grant
funded by the Korea government (MSIT) (No. NRF-2020R1C1C1008296) and Institute of Information \&
Communications Technology Planning \& Evaluation (IITP) grant funded by the Korea government (MSIT)
(No. 2019-0-00075, Artificial Intelligence Graduate School Program (KAIST)).}
}


\author{Geon Lee$^{*}$       \and
        Seokbum Yoon$^{*}$ \and Jihoon Ko \and Hyunju Kim \and 
        Kijung Shin$^{\dagger}$ \thanks{$^{*}$ Equal Contribution. $^{\dagger}$ Corresponding Author.}
}


\institute{
    G. Lee \and  S. Yoon \and J. Ko \and H. Kim \at Kim Jaechul Graduate School of AI, KAIST, Seoul, South Korea, 02455 \\\email{\{geonlee0325,jing9044,jihoonko,hyunju.kim\}@kaist.ac.kr}
    \and
    K. Shin \at Kim Jaechul Graduate School of AI and School of Electrical Engineering, KAIST, Seoul, South Korea, 02455 \\\email{kijungs@kaist.ac.kr}
}

\date{Received: date / Accepted: date}


\newcommand\kijung[1]{\textcolor{black}{#1}}
\newcommand\geon[1]{\textcolor{blue}{[Geon:#1]}}
\newcommand\red[1]{\textcolor{black}{#1}}
\newcommand\black[1]{\textcolor{black}{#1}}
\newcommand\blue[1]{\textcolor{black}{#1}}
\newcommand\change[1]{\textcolor{black}{#1}}
\newcommand\nochange[1]{\textcolor{black}{#1}}
\newcommand\Motif{H-motif\xspace}
\newcommand\Motifs{H-motifs\xspace}
\newcommand\motif{h-motif\xspace}
\newcommand\motifs{h-motifs\xspace}

\newcommand\Tmotif{3H-motif\xspace}
\newcommand\Tmotifs{3H-motifs\xspace}
\newcommand\tmotif{3h-motif\xspace}
\newcommand\tmotifs{3h-motifs\xspace}

\newcommand\Kmotif{$k$H-motif\xspace}
\newcommand\kmotif{$k$h-motif\xspace}
\newcommand\Kmotifs{$k$H-motifs\xspace}
\newcommand\kmotifs{$k$h-motifs\xspace}

\newcommand\method{\textsf{MoCHy}\xspace}
\newcommand\methodE{\textsf{MoCHy-E}\xspace}
\newcommand\methodEN{\textsf{MoCHy-E\textsubscript{ENUM}}\xspace}
\newcommand\methodA{\textsf{MoCHy-A}\xspace}
\newcommand\methodAE{\textsf{MoCHy-A}\xspace}
\newcommand\methodAW{\textsf{MoCHy-A\textsuperscript{+}}\xspace}

\newcommand\methodAWrandom{\textsf{On-the-fly MoCHy-A\textsuperscript{+} (Basic)}\xspace}
\newcommand\methodAWgreedy{\textsf{On-the-fly MoCHy-A\textsuperscript{+} (Adv.)}\xspace}

\newcommand\methodX{MoCHy\xspace}
\newcommand\methodEX{MoCHy-E\xspace}
\newcommand\methodENX{MoCHy-E\textsubscript{ENUM}\xspace}
\newcommand\methodAX{MoCHy-A\xspace}
\newcommand\methodAEX{MoCHy-A\xspace}
\newcommand\methodAWX{MoCHy-A\textsuperscript{+}\xspace}

\newcommand\naive{na\"ive\xspace}
\newcommand\Naive{Na\"ive\xspace}

\newcommand\hwedge{hyperwedge\xspace}
\newcommand\hwedges{hyperwedges\xspace}
\newcommand\Hwedge{Hyperwedge\xspace}
\newcommand\Hwedges{Hyperwedges\xspace}

\newcommand\MT{M[t]\xspace}
\newcommand\MB{\bar{M}\xspace}
\newcommand\MBT{\bar{M}[t]\xspace}
\newcommand\MH{\hat{M}\xspace}
\newcommand\MHT{\hat{M}[t]\xspace}
\newcommand\MD{\tilde{M}\xspace}
\newcommand\MDT{\tilde{M}[t]\xspace}

\newcommand\mB{\bar{m}\xspace}
\newcommand\mBT{\bar{m}[t]\xspace}
\newcommand\mH{\hat{m}\xspace}
\newcommand\mHT{\hat{m}[t]\xspace}

\newcommand\Mi{M_{e_i}\xspace}
\newcommand\Mij{M_{\wedge_{ij}}\xspace}
\newcommand\Mijk{M_{\sqcap_{ijk}}\xspace}
\newcommand\MTi{M_{e_i}[t]\xspace}
\newcommand\MTij{M_{\wedge_{ij}}[t]\xspace}
\newcommand\MTijk{M_{\sqcap_{ijk}}[t]\xspace}

\newcommand\wij{\wedge_{ij}\xspace}
\newcommand\wik{\wedge_{ik}\xspace}
\newcommand\wjk{\wedge_{jk}\xspace}
\newcommand\wki{\wedge_{ki}\xspace}
\newcommand\eijk{\{e_i,e_j,e_k\}\xspace}
\newcommand\hijk{h(\{e_i,e_j,e_k\})\xspace}
\newcommand\thijk{h_{3}(\{e_i,e_j,e_k\})\xspace}
\newcommand{\egeneral}{\{e_{s_1},e_{s_2},...,e_{s_k}\}\xspace}
\newcommand{\hgeneral}{h(\{e_{s_1},e_{s_2},...,e_{s_k}\})\xspace}

\newcommand\GT{\bar{G}}

\newcommand\NT{{N}}
\newcommand\nei{N_{e_{i}}}
\newcommand\nej{N_{e_{j}}}
\newcommand\nek{N_{e_{k}}}
\newcommand\PT{\wedge}

\newcommand\degt[1]{|\NT_{#1}|}
\newcommand{\smallsection}[1]{{\vspace{0.02in} \noindent {\bf{\underline{\smash{#1}}}}}}

\maketitle

\begin{abstract}
 Hypergraphs naturally represent group interactions, which are omnipresent in many domains: collaborations of researchers, co-purchases of items, and joint interactions of proteins, to name a few.
In this work, we propose tools for answering the following questions in a systematic manner:
{\bf (Q1)} what are the structural design principles of real-world hypergraphs?
{\bf (Q2)} how can we compare local structures of hypergraphs of different sizes?
{\bf (Q3)} how can we identify domains from which hypergraphs are?

We first define {\it hypergraph motifs} (\motifs), which describe the \kijung{overlapping} patterns of three connected hyperedges.
Then, we define the significance of each \motif in a hypergraph as its occurrences relative to those in properly randomized hypergraphs.
Lastly, we define the {\it characteristic profile} (CP) as the vector of the normalized significance of every \motif.
Regarding Q1, we find that \motifs' occurrences in $11$ real-world hypergraphs from $5$ domains are clearly distinguished from those of randomized hypergraphs. In addition, we demonstrate that CPs capture local structural patterns unique to each domain, and thus comparing CPs of hypergraphs addresses Q2 and Q3.
\kijung{The concept of CP is naturally extended to represent the connectivity pattern of each node or hyperedge as a vector, which proves useful in node classification and hyperedge prediction.}


Our algorithmic contribution is to propose \method, a family of parallel algorithms for counting \motifs' occurrences in a hypergraph. 
We theoretically analyze their speed and accuracy and show empirically that the advanced approximate version \methodAW is up to $25\times$ more accurate and $32\times$ faster than the basic approximate and exact versions, respectively.


\kijung{
Furthermore, we explore \textit{ternary hypergraph motifs}  that extends h-motifs by taking into account not only the presence but also the cardinality of intersections among hyperedges. This extension proves beneficial for all previously mentioned applications.
}
\end{abstract}

\section{Introduction}
\label{intro}

\begin{figure*}[t]
	\hspace{-2mm}
	\subfigure[Example data: coauthorship relations\label{fig:example:coauthorship}]{
		\includegraphics[width=0.5\columnwidth]{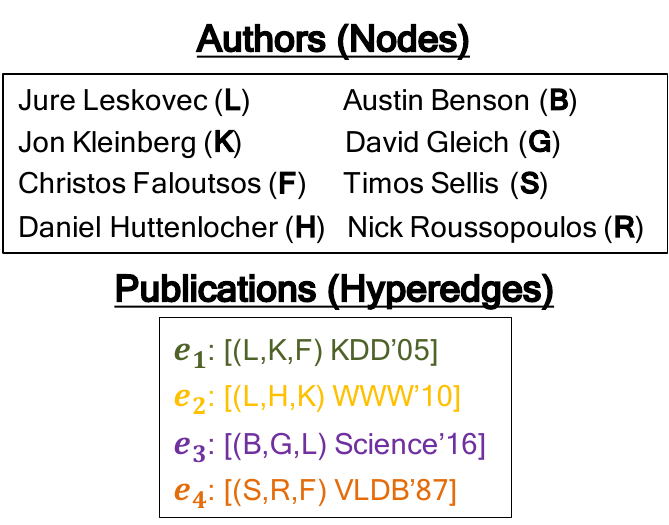}
	}
	\subfigure[Hypergraph representation\label{fig:example:hypergraph}]{
		\includegraphics[width=0.3245\columnwidth]{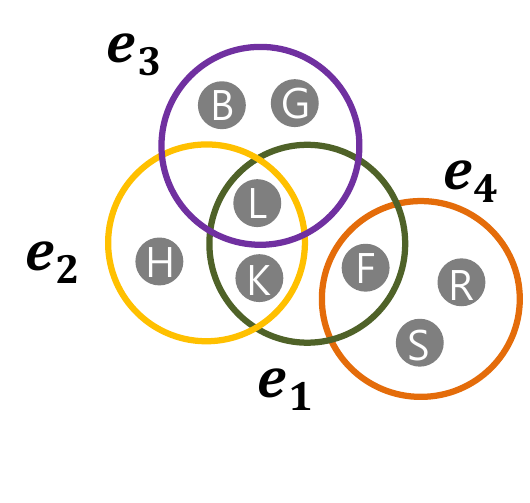}
	}
	\subfigure[\blue{Line-graph representation}\label{fig:example:graph}]{
		\includegraphics[width=0.2135\columnwidth]{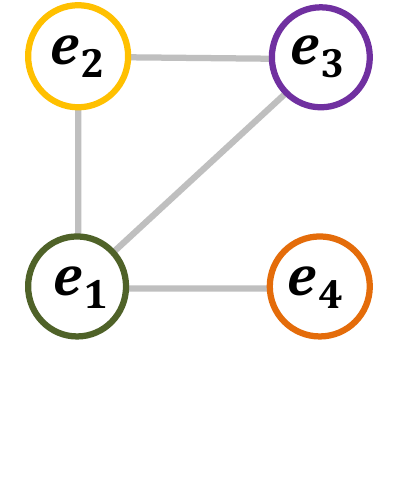} 
	}
	\subfigure[\Motifs and instances \label{fig:example:motif}]{
		\includegraphics[width=0.937\columnwidth]{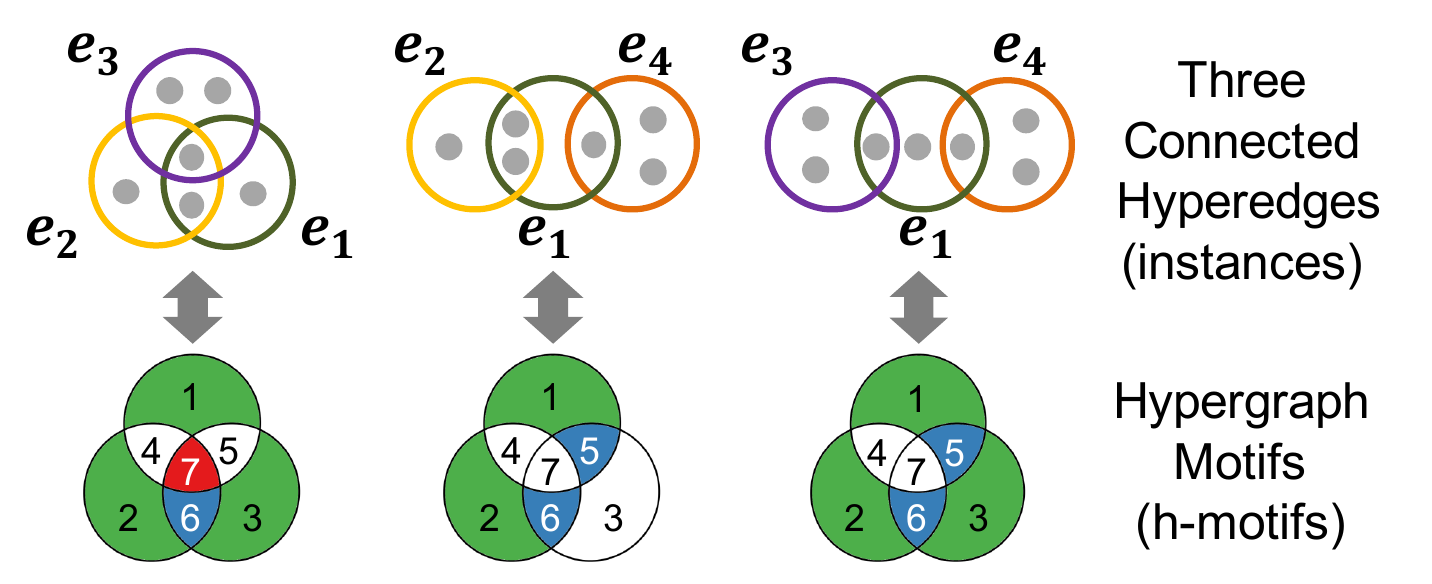}
	}
	\\
	\vspace{-2mm}
	\caption{
		\label{fig:example}
		(a) Example: co-authorship relations.
			(b) Hypergraph: the hypergraph representation of (a).
			(c) Line Graph: the \blue{line-graph} representation of (b).
			(d) Hypergraph Motifs: example \motifs and their instances in (b).}
\end{figure*}

Complex systems consisting of pairwise interactions between individuals or objects are naturally expressed in the form of graphs. Nodes and edges, which compose a graph, represent individuals (or objects) and their pairwise interactions, respectively.
Thanks to their powerful expressiveness, graphs have been used in a wide variety of fields, including social network analysis, web, bioinformatics, and epidemiology. Global structural patterns of real-world graphs, such as power-law degree distribution \cite{barabasi1999emergence,faloutsos1999power} and six degrees of separation \cite{kang2010radius,watts1998collective}, have been extensively investigated.

In addition to global patterns, real-world graphs exhibit patterns in their local structures, which differentiate graphs in the same domain from random graphs or those in other domains. 
Local structures are revealed by counting the occurrences of different network motifs \cite{milo2004superfamilies,milo2002network}, which describe the patterns of pairwise interactions between a fixed number of connected nodes (typically $3$, $4$, or $5$ nodes).   
As a fundamental building block, network motifs have played a key role in many analytical and predictive tasks, including community detection \cite{benson2016higher,li2019edmot,tsourakakis2017scalable,yin2017local}, classification \cite{chen2013identification,lee2019graph,milo2004superfamilies}, and anomaly detection \cite{becchetti2010efficient,shin2020fast}.


Despite the prevalence of graphs, interactions in several complex systems are groupwise rather than pairwise: collaborations of researchers, co-purchases of items, joint interactions of proteins, tags attached to the same web post, to name a few. 
These group interactions cannot be represented by edges in a graph.
Suppose three or more researchers coauthor a publication. This co-authorship cannot be represented as a single edge, and creating edges between all pairs of the researchers cannot be distinguished from multiple papers coauthored by subsets of the researchers.

This inherent limitation of graphs is addressed by hypergraphs, which consist of nodes and hyperedges.
Each hyperedge is a subset of any number of nodes, and it represents a group interaction among the nodes.
For example, the coauthorship relations in Figure~\ref{fig:example:coauthorship} are naturally represented as the hypergraph in Figure~\ref{fig:example:hypergraph}.
In the hypergraph, seminar work \cite{leskovec2005graphs} coauthored by Jure Leskovec (L), Jon Kleinberg (K), and Christos Faloutsos (F) is expressed as the hyperedge $e_{1}=\{L,K,F\}$, and it is distinguished from three papers coauthored by each pair, which, if they exist, can be represented as three hyperedges $\{K,L\}$, $\{F,L\}$, and $\{F,K\}$.



The successful investigation and discovery of local structural patterns in real-world graphs motivate us to explore local structural patterns in real-world hypergraphs.
However, network motifs, which proved to be useful for graphs, are not trivially extended to hypergraphs.
\blue{Due to the flexibility in the size of hyperedges, it is possible to form $2^n$ distinct hyperedges with a given set of $n$ nodes.
As a result, the potential number of hypergraphs is $2^{2^n}$, which is extraordinarily large even for a small number of nodes. 
This implies that there can be numerous possible interactions among hyperedges, highlighting the complexity of hypergraph structures.}

In this work, taking these challenges into consideration, we define $26$ {\it hypergraph motifs} (\motifs) so that they describe overlapping patterns of three connected hyperedges (rather than nodes).
As seen in Figure~\ref{fig:example:motif}, \motifs describe the overlapping pattern of hyperedges $e_{1}$, $e_{2}$, and $e_{3}$ by the emptiness of seven subsets: $e_{1}\setminus e_{2} \setminus e_{3}$, $e_{2}\setminus e_{3} \setminus e_{1}$,  $e_{3}\setminus e_{1} \setminus e_{2}$, $e_{1}\cap e_{2} \setminus e_{3}$, $e_{2}\cap e_{3} \setminus e_{1}$, $e_{3}\cap e_{1} \setminus e_{2}$, and $e_{1}\cap e_{2} \cap e_{3}$.
As a result, every overlapping pattern is described by a unique \motif, independently of the sizes of hyperedges. 
While this work focuses on overlapping patterns of three hyperedges, \motifs are easily extended to four or more hyperedges.


We count the number of each \motif's instances in $11$ real-world hypergraphs from $5$ different domains.
Then, we measure the significance of each \motif in each hypergraph by comparing the count of its instances in the hypergraph against the counts in properly randomized hypergraphs. Lastly, we compute the {\it characteristic profile} (CP) of each hypergraph, defined as the vector of the normalized significance of every \motif.
Comparing the counts and CPs of different hypergraphs leads to the following observations:
\begin{itemize}
	\item Structural design principles of real-world hypergraphs that are captured by frequencies of different \motifs are clearly distinguished from those of randomized hypergraphs.
	\item Hypergraphs from the same domains have similar CPs, while hypergraphs from different domains have distinct CPs (see Figure~\ref{fig:crown}). In other words, CPs successfully capture local structure patterns unique to each domain.
\end{itemize}

\begin{figure}[t]
	\centering 
	\includegraphics[width=\linewidth]{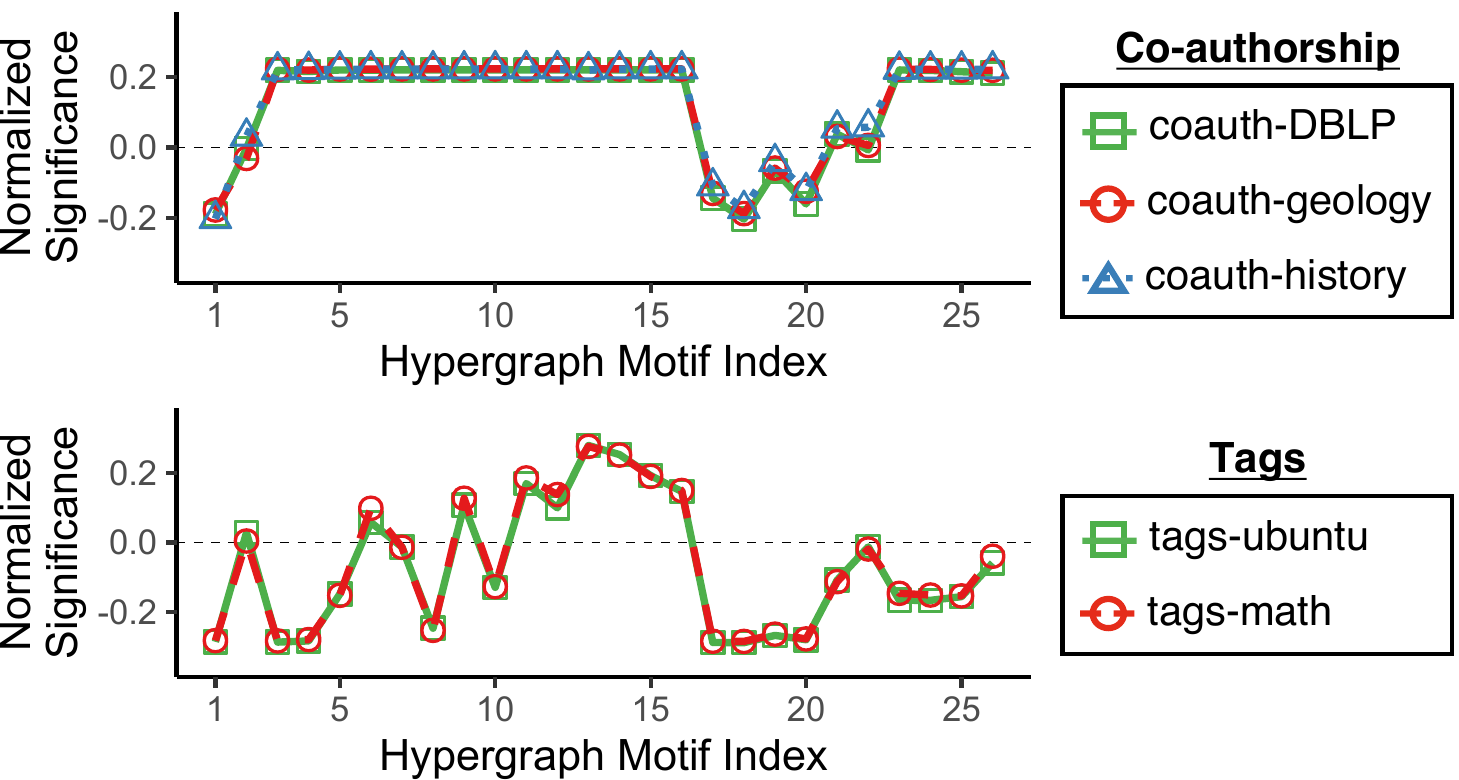}
	\caption{\label{fig:crown} 
	Distributions of \motifs' instances precisely characterize local structural patterns of real-world hypergraphs.
	Note that the hypergraphs from the same domains have similar distributions, while the hypergraphs from different domains do not.
	See Section~\ref{sec:exp:domain} for details.	
}
\end{figure}

\black{
Similarly, \motifs can also be employed to summarize the connectivity pattern of each node or hyperedge.
Specifically, for each node, we can calculate its {\it node profile} (NP), a 26-element vector with each element indicating the frequency of each motif's instances within the node's ego-network.
Likewise, the {\it hyperedge profile} (HP) of a hyperedge is a 26-element vector with each element representing the count of each motif's instances that involve the hyperedge.
We demonstrate empirically that NPs and HPs effectively capture local connectivity patterns, serving as valuable features for node classification and hyperedge prediction tasks.
}


Our algorithmic contribution is to design \method (\textbf{Mo}tif \textbf{C}ounting in \textbf{Hy}pergraphs), a family of parallel algorithms for counting \motifs' instances, which is the computational bottleneck of the aforementioned process.
Note that since \blue{multi-way overlaps} are taken into consideration, counting the instances of \motifs is more challenging than counting the instances of network motifs, which are defined solely based on pairwise interactions.
We provide one exact version, named \methodE, and two approximate versions, named \methodAE and \methodAW.
Empirically, \methodAW is up to $25\times$ more accurate than \methodAE, and it is up to $32\times$ faster than \methodE, with little sacrifice of accuracy.
These empirical results are consistent with our theoretical analyses.

\black{
Additionally, we investigate \textit{ternary hypergraph motifs} (\tmotifs) as a promising extension of \motifs. While \motifs focus only on the emptiness of seven subsets derived from intersections among hyperedges, \tmotifs further differentiate patterns based on the cardinality of these subsets. In particular, \tmotifs consider whether the cardinality of each non-empty subset surpasses a specific threshold or not, resulting in $431$ distinct patterns. We demonstrate that employing \tmotifs instead of \motifs leads to performance improvements in all the previously mentioned applications, i.e., hypergraph (domain) classification, node classification, and hyperedge prediction.
}


In summary, our contributions are summarized as follows:
\begin{itemize}
    \item {\bf Novel Concepts:} We introduce \motifs, which capture the local structures of hypergraphs, independently of the sizes of hyperedges or hypergraphs. 
    \black{We extend this concept to \tmotifs, allowing for a more detailed distinction of local structures.}
    \item {\bf Fast and Provable Algorithms:} We develop \method, a family of parallel algorithms for counting \motifs' instances. We show theoretically and empirically that the advanced version significantly outperforms the basic ones, providing a better trade-off between speed and accuracy.
    \item {\bf Discoveries in $11$ Real-world Hypergraphs:} We show that \motifs and \black{\tmotifs} reveal local structural patterns that are shared by hypergraphs from the same domains but distinguished from those of random hypergraphs and hypergraphs from other domains (see Figure~\ref{fig:crown}). 
    \item {\bf Machine Learning Applications:} \black{We empirically demonstrate that \motifs allow for the extraction of effective features in three machine-learning tasks, and employing \tmotifs enables the extraction of even stronger features.}
\end{itemize}

\noindent{\bf Reproducibility:} The code and datasets used in this work are available at \url{https://github.com/jing9044/MoCHy-with-3h-motif}.

\black{This paper is an extension of our previous work \cite{lee2020hypergraph}, which first introduced the concept of \motifs and related counting algorithms. In this extended version, we investigate various extensions of \motifs, including \tmotifs (Section~\ref{sec:extension} and Appendices~\ref{appendix:hmotif:generalized} and
\ref{appendix:kh-motif:generalized}). Furthermore, \blue{we develop an advanced on-the-fly algorithm for improved space efficiency (Section~\ref{sec:method:par_fly}) and establish accuracy guarantees for the approximate counting algorithms in the form of sample concentration bounds (Theorems~\ref{thm:sampling_ver1:concentration} and \ref{thm:sampling_ver2:concentration}).}
We also evaluate the effectiveness of \motifs for machine learning applications on three tasks using 7 to 11 datasets (Section~\ref{sec:exp:applications} and Appendices~\ref{appendix:prediction} and \ref{appendix:node_classification}). We especially demonstrate the superior performance of \tmotifs over their variants and \motifs in these tasks (Sections~\ref{sec:exp:characterization_power} and \ref{sec:exp:applications}, and Appendix~\ref{appendix:variants}).
Finally, we measure and compare the importance of different \motifs in characterizing hypergraph structures and their correlation with global structural properties (Section~\ref{sec:exp:domain} and Appendix~\ref{sec:addExp-networkProperties}).
}

In Section~\ref{sec:concept}, we introduce \motifs and \black{related concepts.}
\black{In Section~\ref{sec:characterize}, we describe how we use these concepts to characterize hypergraphs, hyperedges, and nodes.}
In Section~\ref{sec:method}, we present exact and approximate algorithms for counting instances of \motifs, and we analyze their theoretical properties.
\black{In Section~\ref{sec:extension}, we extend \motifs to \tmotifs.}
In Section~\ref{sec:exp}, we provide experimental results. 
After discussing related work in Section~\ref{sec:related}, we offer conclusions and \blue{future research directions} in Section~\ref{sec:summary}.

\section{Proposed Concepts}
\label{sec:concept}

\begin{table}[t!]
	\begin{center}
		\caption{\label{notations}Frequently-used symbols.}
		\scalebox{0.76}{
			\begin{tabular}{c|l}
				\toprule
				\textbf{Notation} & \textbf{Definition}\\
				\midrule
				$G=(V,E)$ & hypergraph with nodes $V$ and hyperedges $E$\\
				$E=\{e_1,...,e_{|E|}\}$ & set of hyperedges \\
				$E_v$ & set of hyperedges that contains a node $v$\\
				\midrule
				$\wedge$ & set of \hwedges in $G$\\
				$\wij$ & \hwedge consisting of $e_i$ and $e_j$ \\
				\midrule
				$\GT=(E,\wedge, \omega)$ & \blue{line graph} representation  of $G$ \\
				$\omega(\wij)$ & the number of nodes shared between $e_i$ and $e_j$ \\
				$\nei$ & set of neighbors of $e_i$ in $\GT$ \\
				\midrule
				$h(\{e_i,e_j,e_k\})$ & \motif corresponding to an instance $\{e_i,e_j,e_k\}$ \\ 
				$M[t]$ & count of \motif $t$'s instances\\
				\bottomrule 
			\end{tabular}}
	\end{center}
\end{table}

\begin{figure*}[t]
	\centering
	\includegraphics[width=1.01\linewidth]{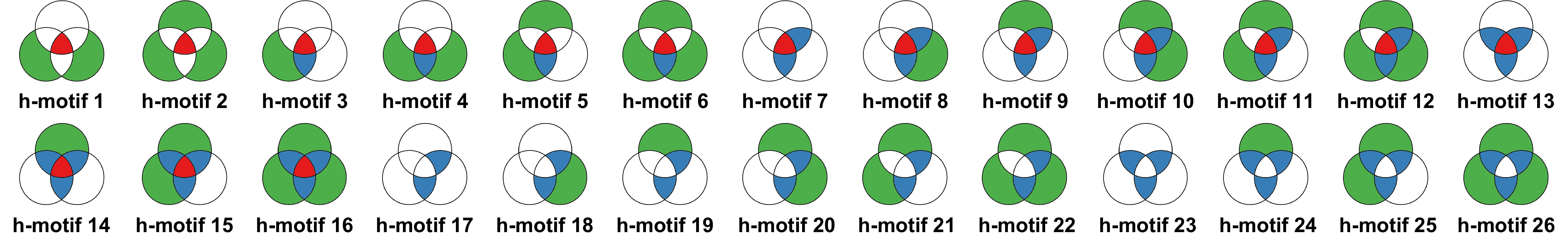}
	\caption{\label{motif_three_hyperedges} The 26 \motifs studied in this work. 
        \black{In each Venn diagram, uncolored regions are empty without containing any nodes, while colored regions include at least one node.}
	\Motifs 17 - 22 are open, while the others are closed.
	}
\end{figure*}

In this section, \black{we introduce preliminary concepts, and based on them, we define the proposed concept, i.e., hypergraph motifs.}  
Refer to Table \ref{notations} for the notations frequently used in the paper.

\subsection{Preliminaries and Notations}
\label{sec:concept:prelim}

We introduce some preliminary concepts and notations.

\smallsection{Hypergraph:} Consider a {\it hypergraph} $G=(V,E)$, where $V$ and $E:=\{e_1,e_2,...,e_{|E|}\}$ are sets of nodes and hyperedges, respectively.\footnote{\blue{Note that, in this work, $E$ is not a multi-set. That is, we assume that every hyperedge is unique.}}
Each hyperedge $e_i\in E$ is a non-empty subset of $V$, and we use $|e_i|$ to denote the number of nodes in it.
For each node $v\in V$, we use $E_v:=\{e_i\in E: v\in e_i\}$ to denote the set of hyperedges that include $v$. 
We say two hyperedges $e_i$ and $e_j$ are {\it adjacent} if they share any member, i.e., if $e_i\cap e_j \neq \varnothing$.
Then, for each hyperedge $e_i$,
we denote the set of hyperedges adjacent to $e_i$ by $\nei:=\{e_j\in E: e_i\cap e_j \neq \varnothing\}$ and the number of such hyperedges by $|\nei|$.
Similarly, we say three hyperedges $e_i$, $e_j$, and $e_k$ are {\it connected} if \blue{there exists at least one hyperedge among them that is adjacent to the other two.} 

\smallsection{\Hwedges:}
We define a \textit{\hwedge} as an unordered pair of adjacent hyperedges. We denote the set of \hwedges in $G$ by $\wedge:=\{\{e_i,e_j\}\in {E \choose 2}: e_i\cap e_j \neq \varnothing\}$.
We use $\wij\in \wedge$ to denote the \hwedge consisting of $e_i$ and $e_j$. 
In the example hypergraph in Figure~\ref{fig:example:hypergraph},
there are four \hwedges: $\wedge_{12}$, $\wedge_{13}$, $\wedge_{23}$, and $\wedge_{14}$.

\begin{figure}[t!]
	\centering
	\includegraphics[width=0.75\columnwidth]{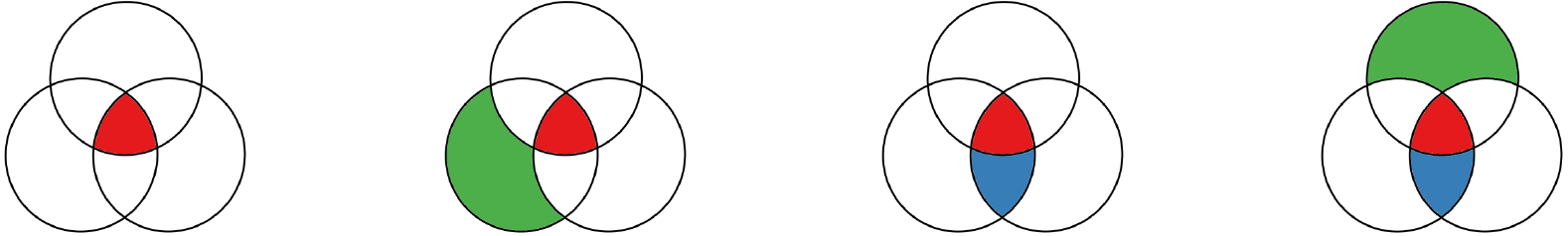} 
	\caption{\label{motif_three_hyperedges_duplicated}
		\blue{The patterns that cannot be obtained from three distinct hyperedges. For example, any three hyperedges corresponding to the leftmost pattern are necessarily identical. However, according to our definition of hypergraphs in Section~\ref{sec:concept:prelim}, every hyperedge is unique. Thus, there should be no instance of the pattern.}
 }
\end{figure}



\smallsection{Line Graph:}
We define the \blue{{\it line graph} (a.k.a., projected graph)} of a hypergraph $G=(V,E)$ as $\GT=(E,\wedge, \omega)$, where $\wedge$ is the set of hyperwedges and $\omega(\wij):=|e_i \cap e_j|$.
That is, in the line graph $\GT$, hyperedges in $G$ act as nodes, and two of them are adjacent if and only if they share any member. \blue{To be more precise, $\GT$ is a weighted variant of a line graph, where each edge is assigned a weight equal to the size of overlap of the corresponding hyperwedge in $G$.}
Note that for each hyperedge $e_i\in E$, $\nei$  is the set of neighbors of $e_i$ in $\GT$, and $|\nei|$ is its degree in $\GT$.
Figure~\ref{fig:example:graph} shows the line graph of the example hypergraph in Figure~\ref{fig:example:hypergraph}.

\smallsection{Incidence Graph:} 
\blue{We define the \textit{incidence graph} (a.k.a., star expansion) of a hypergraph $G=(V,E)$ as $G'=(V',E')$ where $V':=V\cup E$ and $E':=\{(v,e)\in V\times E: v\in e\}$.
That is, in the bipartite graph $G'$, $V$ and $E$ are the two subsets of nodes, and there exists an edge between $v\in V$ and $e \in E$ if and only if $v\in e$.
}



\subsection{Hypergraph Motifs (H-Motifs)}
\label{sec:concept:motif}
We introduce hypergraph motifs, which are basic building blocks of hypergraphs. Then, we discuss their properties and generalization.

\smallsection{Definition and Representation:}
Hypergraph motifs (or \motifs in short) \blue{are designed for describing} the \black{overlapping} patterns of three connected hyperedges.
Specifically, given a set $\eijk$ of three connected hyperedges, \motifs describe its \black{overlapping} pattern by the emptiness of the following seven sets: (1) $e_i\setminus e_j \setminus e_k$, (2) $e_j\setminus e_k \setminus e_i$, (3) $e_k\setminus e_i \setminus e_j$, (4) $e_i\cap e_j \setminus e_k$, (5) $e_j\cap e_k \setminus e_i$, (6) $e_k\cap e_i \setminus e_j$, and (7) $e_i\cap e_j \cap e_k$.
Formally, a \motif is defined as a binary vector of size $7$ whose elements represent the emptiness of the above sets, respectively, and as seen in Figure~\ref{fig:example:motif}, \motifs are naturally represented in the Venn diagram.
\blue{
Equivalently, when we leave at most one node in each of the above subsets, \motifs can be defined based on the isomorphism between sub-hypergraphs consisting of three connected hyperedges.
}
While there can be $2^7$ \motifs, $26$ \motifs remain once we exclude symmetric ones, \blue{those that cannot be obtained from distinct hyperedges} (see Figure~\ref{motif_three_hyperedges_duplicated}), and those that cannot be obtained from connected hyperedges.
The 26 cases, which we call {\it \motif 1} through {\it \motif 26}, are visualized in the Venn diagram in Figure~\ref{motif_three_hyperedges}.


\smallsection{Instances of \Motifs}:
Consider a hypergraph $G=(V,E)$.
A set of three connected hyperedges is an {\it instance} of \motif $t$ if their \black{overlapping} pattern corresponds to \motif $t$.
The count of each \motif's instances is used to characterize the local structure of $G$, as discussed in the following sections.

\smallsection{Open and Closed \Motifs}:
A \motif is {\it closed} if all three hyperedges in its instances are adjacent to (i.e., overlapped with) each other. If its instances contain two non-adjacent (i.e., disjoint) hyperedges, a \motif is {\it open}.
In Figure~\ref{motif_three_hyperedges}, \motifs $17$ - $22$ are open; the others are closed.

\smallsection{Properties of \Motifs:}
From the definition of h-mot-ifs, the following desirable properties are immediate:
\begin{itemize}
\item {\blue{\bf Exhaustivity:}} \motifs capture \black{overlapping} patterns of \textit{all possible} three connected hyperedges.
\item {\blue{\bf Unicity:}} \black{overlapping} pattern of any three connected hyperedges is captured by \textit{at most one} \motif.
\item {\blue{\bf Size Independence:}} \motifs capture \black{overlapping} patterns \textit{independently of the sizes of hyperedges}. Note that there can be infinitely many combinations of sizes of three connected hyperedges.
\end{itemize}
Note that the exhaustiveness and the uniqueness imply that \black{overlapping} pattern of any three connected hyperedges is captured by \textit{exactly one} \motif.

\smallsection{Why \blue{Multi-way} Overlaps?}:
\blue{Multi-way} overlaps (e.g., the emptiness of $e_{1}\cap e_{2} \cap e_{3}$ and $e_{1}\setminus e_{2} \setminus e_{3}$) play a key role in capturing the local structural patterns of real-world hypergraphs.
Taking only the pairwise overlaps (e.g., the emptiness of $e_{1}\cap e_{2}$,  $e_{1}\setminus e_{2}$, and $e_{2} \setminus  e_{1}$) into account limits the number of possible \black{overlapping} patterns of three distinct hyperedges to just eight,\footnote{Note that using the conventional network motifs in s limits this number to two.}  significantly limiting their expressiveness and thus usefulness.
Specifically, $12$ (out of $26$) \motifs have the same pairwise overlaps, while their occurrences and significances vary substantially in real-world hypergraphs. 
For example, in Figure~\ref{fig:example},  $\{e_{1}, e_{2}, e_{4}\}$  and $\{e_{1}, e_{3}, e_{4}\}$ have the same pairwise overlaps, while their \black{overlapping} patterns are distinguished by \motifs.

\section{Characterization using \Motifs}
\label{sec:characterize}
\black{In this section, we outline the process of using \motifs to summarize local structural patterns within a hypergraph, as well as those around individual nodes and hyperedges, for the purpose of characterizing them.
}

\subsection{Hypergraph Characterization}

What are the structural design principles of real-world hypergraphs distinguished from those of random hypergraphs?
Below, we introduce the characteristic profile (CP), which is a tool for answering the above question using \motifs.

\smallsection{Randomized Hypergraphs:}
While one might try to characterize the local structure of a hypergraph by absolute counts of each \motif's instances in it, some \motifs may naturally have many instances.
Thus, for more accurate characterization, we need random hypergraphs to be compared against real-world hypergraphs.
\blue{In the network motif literature, configuration models have been widely employed for this purpose \cite{milo2004superfamilies,milo2002network}. These models generate random graphs while preserving the degree distribution of the original graph. 
Using the configuration model does not introduce an excessive level of randomness, maintaining a meaningful and controlled comparison with the original graph.}

\blue{In line with prior research, we used a configuration model extended to hypergraphs to obtain random hypergraphs.
Specifically, we employ the Chung-Lu model \cite{aksoy2017measuring}, which is a configuration model designed to generate random bipartite graphs while preserving in expectation the degree distributions of the original graph \cite{aksoy2017measuring} (for a precise theoretical description, please refer to Eq.\eqref{eq:degree_preserve} in Appendix~\ref{appendix:random}).
We first apply this model to the incidence graph $G'$ of the input hypergraph $G$ to obtain randomized bipartite graphs, and then we transform them into random hypergraphs.
The empirical distributions of node degrees and hyperedge sizes in the random hypergraphs closely resemble those in
 $G$, as shown in Figure~\ref{deg_fig} in Appendix~\ref{appendix:random}, where
we also provide pseudocode of the process (Algorithm~\ref{randomized_hypergraph_alg}) and its theoretical properties.
}

\smallsection{Significance of \Motifs:}
We measure the significance of each \motif in a hypergraph by comparing the count of its instances against the count of them in randomized hypergraphs.
Specifically, the {\it significance} of a \motif $t$ in a hypergraph $G$ is defined as
	\vspace{-1mm}
\begin{equation}
    \Delta_t := \frac{M[t] - M_{rand}[t]}{M[t] + M_{rand}[t] + \epsilon}, \label{eq:significance}
    	\vspace{-1mm}
\end{equation}
where $M[t]$ is the number of instances of \motif $t$ in $G$, and $M_{rand}[t]$ is the average number of instances of \motif $t$ in randomized hypergraphs. We fixed $\epsilon$ to $1$ throughout this paper.
This way of measuring significance was proposed for network motifs \cite{milo2004superfamilies} as an alternative of normalized Z scores, \blue{which can be dominated by few network motifs with small variances.
Specifically, when the variance of the occurrences of a specific network motif in randomized graphs is very small, the Z-score becomes significantly large, and thus the Z-score of the particular network motif may dominate all others, regardless of its absolute occurrences.}

\smallsection{Characteristic Profile (CP):} By normalizing and concatenating the significances of all \motifs in a hypergraph, we obtain the characteristic profile (CP), which summarizes the local structural pattern of the hypergraph.
Specifically, the {\it characteristic profile} of a hypergraph $G$ is a vector of size $26$, where each $t$-th element is 
	\vspace{-1mm}
\begin{equation}
    CP_t := \frac{\Delta_t}{\sqrt{\sum_{t=1}^{26} \Delta_t^2}}. \label{eq:cp}
    	\vspace{-1mm}
\end{equation}
Note that, for each $t$, $CP_t$ is between $-1$ and $1$.
The CP is used in Section~\ref{sec:exp:domain} to compare the local structural patterns of real-world hypergraphs from diverse domains.


\subsection{\black{Hyperedge Characterization}}

\black{Each individual hyperedge can also  be characterized by the \motif instances that contain it.}

\smallsection{\black{Hyperedge Profile (HP):}}
\black{Specifically, given a hypergraph $G=(V, E)$, the {\it hyperedge profile} (HP) of a hyperedge $e\in E$ is a $26$-element vector, where each $t$-th element is the number of \motif $t$'s instances that include $e$. 
It should be noticed that, for HPs, we use absolute counts of \motif instances rather than their normalized significances.
Normalized significances are introduced for CPs to enable direct comparison of hypergraphs at different scales, specifically with varying numbers of nodes and hyperedges. Since comparisons between individual hyperedges, \blue{such as for the purpose of hyperedge prediction within a hypergraph, may be} free from such issues, we simply use the absolute counts of \motif instances when defining HPs.\footnote{\blue{Recall that the CPs are specifically designed to capture structural similarity between hypergraphs of potentially varying scales, typically using a simple metric such as cosine similarity.
                Regarding HPs and NPs (defined in Section~\ref{sec:characterize:node}), our primary objectives of using them are to distinguish missing hyperedges from other candidates (for HPs)
                and to distinguish nodes from different domains (for NPs).
                For these purposes, the scale information can be useful, and thus, we  employ absolute counts for both HPs and NPs to retain and leverage this scale information.
                It is also important to note that, in our experiments, NPs and HPs are used with classifiers (e.g., hypergraph neural networks) powerful enough to capture (dis)similarity even across differing scales.}}
In Section~\ref{sec:exp:observations:prediction}, we demonstrate the effectiveness of HPs as input features in hyperedge prediction tasks.}



\subsection{\black{Node Characterization}}
\label{sec:characterize:node}

\black{Similarly, we characterize each node by the \motif instances in its ego network.
Below, we introduce three types of ego-networks in hypergraphs, and based on these, we elaborate on the node characterization method.}

\smallsection{\black{Hypergraph Ego-networks:}}
\black{Comrie and Kleinberg \cite{comrie2021hypergraph} defined three distinct types of ego-networks.
For each node $v\in V$ in a hypergraph $G=(V, E)$, 
we denote the neighborhood of $v$ (including $v$ itself) by $V_v:=\bigcup_{e_i\in E_v}e_i$, where $E_v:=\{e_i\in E: v\in e_i\}$.
The \textit{star ego-network} of $v$ \blue{is a subhypergraph of $G$ with $V_v$ as its node set and $E_v$ (i.e., the hyperedges that contain $v$) as its hyperedge set}. 
The \textit{radial ego-network} of $v$ \blue{a subhypergraph of $G$ with $V_v$ as its node set and $R_v:=\{e_i\in E: e_i \subseteq V_v \}$ (i.e., the hyperedges that are subsets of the neighborhood of $v$) as its hyperedge set.}
Lastly, the \textit{contracted ego-network} of $v$ \blue{has $V_v$ as its node set and $C_v:=\bigcup_{e_i\in E} \{e_i \cap V_v\}\setminus\emptyset$ as its hyperedge set,} and mathematically, \blue{the contracted ego-network of $v$} is the subhypergraph of $G$ induced by $V_v$. 
Note that $E_v \subseteq R_v \subseteq  C_v$.
Compared to $E_v$, $R_v$ additionally includes hyperedges that consist only of the neighbors of $v$ but not include $v$.
Compared to $R_v$, $C_v$ additionally includes the non-empty intersection of each hyperedge and the neighborhood of $v$.
}

\smallsection{\black{Node Profile (NP):}}
\black{Given a hypergraph $G=(V, E)$, the {\it node profile} (NP) of a node $v\in V$ is a $26$-element where each $t$-th element is the number of \motif $t$'s instances within an ego-network of $v$. Note that, as for HPs, we use the absolute counts of \motifs, instead of their normalized significances, for NPs.
Depending on the types of ego-networks, we define {\it star node profiles} (SNPs), {\it radial node profiles} (RNPs), and {\it contracted node profiles} (CNPs).
In Appendix~\ref{appendix:node_classification},
we provide an empirical comparison of these three types of NPs in the context of a node classification task. The results show that using RNPs consistently yields better performance than SNPs or CNPs, indicating that additional complete hyperedges (i.e., $R_v\setminus E_v$) are helpful, while partial ones extracted from hyperedges  (i.e., $C_v\setminus R_v$) are not.}


\section{Proposed Algorithms}
\label{sec:method}

Given a hypergraph, how can we count the instances of each \motif? 
Once we count them in the original and randomized hypergraphs, the significance of each \motif and the CP are obtained immediately by Eq.~\eqref{eq:significance} and Eq.~\eqref{eq:cp}.

In this section, we present \method (\textbf{Mo}tif \textbf{C}ounting in \textbf{Hy}pergraphs), which is a family of parallel algorithms for counting the instances of each \motif in the input hypergraph.
We first describe \blue{line-graph construction}, which is a preprocessing step of every version of \method.
Then, we present \methodE, which is for exact counting.
After that, we present two different versions of  \methodA, which are sampling-based algorithms for approximate counting.
Lastly, we discuss parallel and on-the-fly implementations.

Throughout this section, we use $\hijk$ to denote the \motif that describes the connectivity pattern of an \motif instance $\eijk$. We also use $\MT$ to denote the count of instances of \motif $t$.

\begin{algorithm}[t]
	\setstretch{1.25}
	\small
	\caption{\small \blue{Line Graph Construction} \\ (Preprocess)}
	\label{hyperwedge_counting}
	\SetAlgoLined
	\SetKwInOut{Input}{Input}
	\SetKwInOut{Output}{Output}
	\nonl \hspace{-4mm} \Input{input hypergraph: $G=(V,E)$}
	\nonl \hspace{-4mm} \Output{\blue{line graph}: $\GT=(E,\PT, \omega)$}
	$\PT\leftarrow \varnothing$ \\
	$\omega \leftarrow$ map whose default value is $0$ \\
	\For{\textbf{each} hyperedge $e_i \in E$ \blue{\textbf{(in parallel)}} \label{hyperwedge_counting:loop1}}{  
		\For{\textbf{each} node $v \in e_i$ \label{hyperwedge_counting:loop2}}{
			\For{\textbf{each} hyperedge $e_j \in E_v$ where $j>i$ \label{hyperwedge_counting:loop3}}{
				$\PT\leftarrow \PT \cup \{\wij\}$ \\ \label{hyperwedge_counting:body1}
				$\omega(\wij)=\omega(\wij)+ 1$
				\label{hyperwedge_counting:body2}
			}
		}
	}
	return $\GT=(E,\PT, \omega)$ 
\end{algorithm}

\smallsection{Remarks:}
	The problem of counting \motifs' occurrences bears some similarity to the classic problem of counting network motifs' occurrences.
	However, \blue{differently from network motifs,} which are defined solely based on pairwise interactions, \motifs are defined based on triple-wise interactions (e.g., $e_i\cap e_j \cap e_k$). 
	One might hypothesize that our problem can easily be reduced to the problem of counting the occurrences of network motifs, and thus existing solutions (e.g., \cite{bressan2019motivo,pinar2017escape}) are applicable to our problem.	
	In order to examine this possibility, we consider the following two attempts:
	\begin{enumerate}
		\item[(a)] Represent pairwise relations between hyperedges using the \blue{line graph}, where each edge $\{e_i,e_j\}$ indicates $e_i\cap e_j\neq \emptyset$.
		\item[(b)] Represent pairwise relations between hyperedges using the directed \blue{line graph} where each directed edge $e_i \rightarrow e_j$ indicates 
		$e_i\cap e_j\neq \emptyset$ and at the same time $e_i \not\subset e_j$.
	\end{enumerate}
	The number of possible connectivity patterns (i.e., network motifs) among three distinct connected hyperedges is just two (i.e., closed and open triangles) and eight in (a) and (b), respectively.
	In both cases,
	instances of multiple \motifs are not distinguished by network motifs, and 
	the occurrences of \motifs can not be inferred from those of network motifs.

In addition, another computational challenge stems from the fact that triple-wise and even pair-wise relations between hyperedges need to be computed from the input hypergraph, while pairwise relations between edges are given in graphs. This challenge necessitates the precomputation of partial relations, described in the next subsection.
	


\subsection{\blue{Line Graph Construction} (Algorithm~\ref{hyperwedge_counting})}
\label{sec:method:projection}
As a preprocessing step, every version of \method builds the \blue{line graph} $\GT=(E,\PT, \omega)$ (see Section~\ref{sec:concept:prelim}) of the input hypergraph $G=(V,E)$, as described in Algorithm~\ref{hyperwedge_counting}. 
To find the neighbors of each hyperedge $e_i$ (line~\ref{hyperwedge_counting:loop1}), the algorithm visits each hyperedge $e_j$ that contains $v$ and satisfies $j>i$ (line~\ref{hyperwedge_counting:loop3}) for each node $v\in e_i$ (line~\ref{hyperwedge_counting:loop2}).
Then for each such $e_j$, it adds $\wij=\{e_i,e_j\}$ to $\PT$ and increments $\omega(\wij)$ (lines~\ref{hyperwedge_counting:body1} and \ref{hyperwedge_counting:body2}).
The time complexity of this preprocessing step is given in Lemma~\ref{lemma:projection:time}.
\begin{lemma}[Complexity of Line Graph Construction\label{lemma:projection:time}]
	The \blue{expected} time complexity of Algorithm~\ref{hyperwedge_counting} is $O(\sum_{\wij\in\PT}|e_i\cap e_j|)$.
\end{lemma}
\begin{proof}
If all sets and maps are implemented using hash tables, \blue{the expected time complexity of lines~\ref{hyperwedge_counting:body1} and \ref{hyperwedge_counting:body2} is $O(1)$ in expectation with uniform hash functions}, and they are executed $|e_i\cap e_j|$ times for each $\wij\in \PT$.
\end{proof}
\noindent Since $|\PT|<\sum_{e_i\in E}\degt{e_i}$ and $|e_i\cap e_j|\leq|e_i|$, Eq.~\eqref{eq:projection:time} holds.
\begin{equation}
\sum\nolimits_{\wij\in\PT}|e_i\cap e_j| < \sum\nolimits_{e_i\in E}(|e_i|\cdot|\nei|). \label{eq:projection:time} 
\end{equation}


\begin{algorithm}[t]
	\setstretch{1.25}
	\small
	\caption{\label{exact_motif_counting} \small \methodE: Exact \Motif Counting }
	\SetAlgoLined
	\SetKwInOut{Input}{Input}
	\SetKwInOut{Output}{Output}
	\nonl \hspace{-4mm} \Input{ \ \ (1) input hypergraph: $G=(V,E)$ \\ (2) \blue{line graph}: $\GT=(E,\PT, \omega)$}
	\nonl \hspace{-4mm} \Output{exact count of each \motif $t$'s instances $\MT$}
	$M \leftarrow$ map whose default value is $0$ \\
	\For{\textbf{each} hyperedge $e_i \in E$ \blue{\textbf{(in parallel)}}\label{exact_motif:loop1}}{
		\For{\textbf{each} unordered hyperedge pair $\{e_j, e_k\} \in$ $\nei \choose 2$\label{exact_motif:loop2}}{
			\If{$e_j \cap e_k = \varnothing$ or $i<\min(j,k)$\label{exact_motif:condition}}{
				$M[\hijk] \mathrel{+}= 1$ \label{exact_motif:count}
			}
		}
	}
	\textbf{return} $M$ 
\end{algorithm}

\subsection{Exact \Motif Counting (Algorithm~\ref{exact_motif_counting})}
\label{sec:method:exact}

We present \methodE (\method \textbf{E}xact), which counts the instances of each \motif exactly. 
The procedures of \methodE are described in Algorithm~\ref{exact_motif_counting}.
For each hyperedge $e_i\in E$ (line~\ref{exact_motif:loop1}), each unordered pair $\{e_j,e_k\}$ of its neighbors, where $\eijk$ is an \motif instance, is considered (line~\ref{exact_motif:loop2}).
If $e_j \cap e_k = \varnothing$ (i.e., if the corresponding \motif is open), $\eijk$ is considered only once.
However, if $e_j \cap e_k \neq \varnothing$ (i.e., if the corresponding \motif is closed),
$\eijk$ is considered two more times (i.e., when $e_j$ is chosen in line~\ref{exact_motif:loop1} and when $e_k$ is chosen in line~\ref{exact_motif:loop1}).
Based on these observations, given an \motif instance $\eijk$, the corresponding count $M[\hijk]$ is incremented (line~\ref{exact_motif:count}) only if $e_j \cap e_k = \varnothing$ or $i < \min(j,k)$ (line~\ref{exact_motif:condition}). This guarantees that each instance is counted exactly once. 
The time complexity of \methodE is given in Theorem~\ref{thm:exact:time}, which uses Lemma~\ref{lemma:motif:time}.

\begin{lemma}[Complexity of Computing $\hijk$]\label{lemma:motif:time}
	Given the input hypergraph $G=(V,E)$ and its \blue{line graph} $\GT=(E,\wedge,\omega)$, for each \motif instance $\eijk$, \blue{the expected time for computing $\hijk$ is $O(\min\\(|e_i|,|e_j|,|e_k|))$.}
\end{lemma}
\begin{proof}
	Assume $|e_i|=\min(|e_i|,|e_j|,|e_k|)$, without loss of generality, and all sets and maps are implemented using hash tables.
	As defined in Section~\ref{sec:concept:motif}, $\hijk$ is computed in $O(1)$ time from the emptiness of the following sets:
	(1) $e_i\setminus e_j \setminus e_k$, (2) $e_j\setminus e_k \setminus e_i$, (3) $e_k\setminus e_i \setminus e_j$, (4) $e_i\cap e_j \setminus e_k$, (5) $e_j\cap e_k \setminus e_i$, (6) $e_k\cap e_i \setminus e_j$, and (7) $e_i\cap e_j \cap e_k$.
	We check their emptiness from their cardinalities.
	We obtain $e_i$, $e_j$, and $e_k$, which are stored in $G$, and their cardinalities in $O(1)$ time.
	Similarly, we obtain $|e_i \cap e_j|$, $|e_j \cap e_k|$, and $|e_k \cap e_i|$, which are stored in $\GT$, in $O(1)$ time \blue{in expectation with uniform hash functions}.
	Then, we compute $|e_i \cap e_j \cap e_k|$ in $O(|e_i|)$ time \blue{in expectation} by checking for each node in $e_i$ whether it is also in both $e_j$ and $e_k$.
	From these cardinalities, we obtain the cardinalities of the six other sets in $O(1)$ time as follows: 
	\begin{align*}
		& (1) \ |e_i\setminus e_j \setminus e_k|  = |e_i|-|e_i\cap e_j|-|e_k\cap e_i|+|e_i \cap e_j \cap e_k|, \\
		& (2) \ |e_j\setminus e_k \setminus e_i| = |e_j|-|e_i\cap e_j|-|e_j\cap e_k|+|e_i \cap e_j \cap e_k|, \\
		& (3) \ |e_k\setminus e_i \setminus e_j| = |e_k|-|e_k\cap e_i|-|e_j\cap e_k|+|e_i \cap e_j \cap e_k|, \\
		& (4) \ |e_i\cap e_j \setminus e_k| = |e_i\cap e_j| - |e_i \cap e_j \cap e_k|, \\
		& (5) \ |e_j\cap e_k \setminus e_i| = |e_j\cap e_k| - |e_i \cap e_j \cap e_k|, \\
		&(6) \ |e_k\cap e_i \setminus e_j| = |e_k\cap e_i| - |e_i \cap e_j \cap e_k|.
	\end{align*}
	Hence, the \blue{expected} time complexity of computing \\ $\hijk$ is $O(|e_i|)=O(\min(|e_i|,|e_j|,|e_k|))$.
\end{proof}

\begin{theorem}[Complexity of \methodE]\label{thm:exact:time}
	The \blue{expected} time complexity of Algorithm~\ref{exact_motif_counting} is $O(\sum_{e_i \in E}(|\nei|^2 \cdot |e_i|))$. 
\end{theorem}
\begin{proof}
 Assume all sets and maps are implemented using hash tables.
 The total number of triples $\eijk$ considered in line~\ref{exact_motif:loop2} is $O(\sum_{e_i \in E}|\nei|^2)$.
 By Lemma~\ref{lemma:motif:time}, for such a triple $\eijk$, \blue{the expected time for computing $\hijk$ is $O(|e_i|)$.}
 Thus, the total \blue{expected} time complexity of Algorithm~\ref{exact_motif_counting} is $O(\sum_{e_i \in E}(|e_i|\cdot|\nei|^2))$, which dominates that of the preprocessing step (see Lemma~\ref{lemma:projection:time} and Eq.~\eqref{eq:projection:time}). 
\end{proof}	
\smallsection{Extension of \methodE to \Motif Enumeration:} \\ \change{Since \methodE visits all \motif instances to count them, it is extended to the problem of enumerating every \motif instance (with its corresponding \motif), as described in Algorithm~\ref{alg:enum}.
The time complexity remains the same.}

\begin{algorithm}[t]
	\setstretch{1.25}
	\small
	\caption{\label{alg:enum}\small \methodEN for \Motif Enumeration}
	\SetAlgoLined
	\SetKwInOut{Input}{Input}
	\SetKwInOut{Output}{Output}
		\nonl \hspace{-4mm} \Input{ \ \ (1) input hypergraph: $G=(V,E)$ \\ (2) \blue{line graph}: $\GT=(E,\PT, \omega)$}
	\nonl \hspace{-4mm} \Output{\motif instances and their corresponding \motifs}
	\For{\textbf{each} hyperedge $e_i \in E$ \blue{\textbf{(in parallel)}} \label{enum_alg:loop1}}{
		\For{\textbf{each} unordered hyperedge pair $\{e_j, e_k\} \in$ $\nei \choose 2$\label{enum_alg:loop2}}{
			\If{$e_j \cap e_k = \varnothing$ or $i<\min(j,k)$\label{enum_alg:condition}}{
				write($e_i$, $e_j$, $e_k$, $\hijk$)\label{enum_alg:write}
			}
		}
	}
\end{algorithm}

\IncMargin{0.5em}
\begin{algorithm}[t]	
	\setstretch{1.25}
	\small
	\caption{\small \methodAE: Approximate \Motif Counting Based on Hyperedge Sampling}
	\label{sampling_ver1}
	\SetAlgoLined
	\SetKwInOut{Input}{Input}
	\SetKwInOut{Output}{Output}
	\nonl \hspace{-5mm} \Input{ \ \ (1) input hypergraph: $G=(V,E)$ \\ (2) \blue{line graph}: $\GT=(E,\PT, \omega)$ \\  (3) number of samples: $s$}
	\nonl \hspace{-5mm} \Output{estimated count of each \motif $t$'s instances: $\MBT$}
	$\MBT\leftarrow$ map whose default value is $0$\\ 
	\For{$n\leftarrow 1...s$ \blue{\textbf{(in parallel)}}}{
		$e_i\leftarrow$ sample a uniformly random hyperedge \label{sampling_ver1:sample} \\
		\For{\textbf{each} hyperedge $e_j \in N_{e_i}$}{\label{sampling_ver1:loop1:start}
			\For{\textbf{each} hyperedge $e_k \in (N_{e_i} \cup N_{e_j} \setminus \{e_i,e_j\})$}{
				\If{$e_k \not\in N_{e_i}$ or $j<k$\label{sampling_ver1:condition}}{
					$\MB[\hijk] \mathrel{+}= 1$ 
					\label{sampling_ver1:count}
				}
			}
		}
	}
	\label{sampling_ver1:loop1:end}
	\For{\textbf{each} \motif $t$}{\label{sampling_ver1:scale:start}
		$\MBT \leftarrow \MBT \cdot \tfrac{|E|}{3s}$
	}
	\label{sampling_ver1:scale:end}
	\textbf{return} $\MB$ 
\end{algorithm}
\DecMargin{0.5em}

\subsection{Approximate \Motif Counting}
\label{sec:method:approx}

We present two different versions of \methodA (\method \textbf{A}pproximate), which approximately count the instances of each \motif.
Both versions yield unbiased estimates of the counts by exploring the input hypergraph partially through hyperedge and \hwedge sampling, respectively.



\smallsection{\methodAE: Hyperedge Sampling (Algorithm~\ref{sampling_ver1}):}\label{sec:method:approx:ver1} 

\noindent \methodAE (Algorithm~\ref{sampling_ver1}) is based on hyperedge sampling. 
It repeatedly samples $s$ hyperedges from the hyperedge set $E$ uniformly at random with replacement (line~\ref{sampling_ver1:sample}). 
For each sampled hyperedge $e_i$, the algorithm searches for all \motif instances that contain $e_i$ (lines~\ref{sampling_ver1:loop1:start}-\ref{sampling_ver1:loop1:end}), and to this end, the $1$-hop and $2$-hop neighbors of $e_i$ in the \blue{line graph} $\GT$ are explored. After that, for each such instance $\eijk$ of h-motif $t$, the corresponding count $\MBT$ is incremented (line~\ref{sampling_ver1:count}). 
Lastly, each estimate $\MBT$ is rescaled by multiplying it with $\frac{|E|}{3s}$ (lines~\ref{sampling_ver1:scale:start}-\ref{sampling_ver1:scale:end}), which is the reciprocal of the expected number of times that each of the \motif $t$'s instances is counted.\footnote{Each hyperedge is expected to be sampled $\frac{s}{|E|}$ times, and each \motif instance is counted whenever any of its $3$ hyperedges is sampled.}
This rescaling makes each estimate $\MBT$ unbiased, as formalized in Theorem~\ref{thm:sampling_ver1:accuracy}. 

\begin{theorem}[Bias and Variance of \methodAE]
	\label{thm:sampling_ver1:accuracy}
	For every \motif t, Algorithm~\ref{sampling_ver1} provides an unbiased estimate $\MBT$ of the count $\MT$ of its instances, i.e.,	
	\begin{equation}
	\mathbb{E}[\MBT]=\MT. \label{sampling_ver1:bias}
	\end{equation}
	The variance of the estimate is
	\begin{equation}
	\mathbb{V}\mathrm{ar}[\MBT] = \frac{1}{3s}\cdot \MT\cdot (|E|-3)+\frac{1}{9s}\sum_{l=0}^{2}p_l[t] \cdot(l|E|-9),
	\label{sampling_ver1:variance}
	\end{equation}
	where $p_l[t]$ is the number of pairs of \motif $t$'s instances that share $l$ hyperedges. 
\end{theorem}
\begin{proof}
	See Appendix~\ref{sampling_ver1:proof}.
\end{proof}

The time complexity of \methodAE is given in Theorem~\ref{thm:sampling_ver1:time}.
\begin{theorem}[Complexity of \methodAE] \label{thm:sampling_ver1:time}
	The \blue{expected} time complexity of Algorithm~\ref{sampling_ver1} is $O(\frac{s}{|E|}\sum_{e_i\in E}(|e_i|\cdot|\nei|^2))$.
\end{theorem}
\begin{proof}
	Assume all sets and maps are implemented using hash tables.
	For a sample hyperedge $e_i$, computing $\nei \cup \nej$ for every $e_j \in \nei$ takes \blue{$O(\sum_{e_j\in \nei} (|\nei| + |\nej|))$ time in expectation with uniform hash functions if we compute $\nei \cup \nej$ by checking whether each hyperedge $e\in \nej$ is also in $\nei$.} By Lemma~\ref{lemma:motif:time}, computing $\hijk$ for all considered \motif instances takes \blue{$O(\min(|e_i|, |e_j|)\cdot\sum_{e_j\in \nei}$ $|\nei| + |\nej|)$} time \blue{in expectation}.
 Thus, the \blue{expected} time complexity for processing a sample $e_i$ is 
	\begin{multline*}
	O(\min(|e_i|, |e_j|)\cdot \sum\nolimits_{e_j\in \nei} (|\nei|+|\nej|)) \\ =O(|e_i|\cdot|\nei|^2 + \sum\nolimits_{e_j\in \nei}(|e_j|\cdot|\nej|)),
	\end{multline*} 
	which can be written as
	\begin{multline*}
	O(\sum\nolimits_{e_i\in E}(\mathbb{1}(e_i \text{ is sampled})\cdot|e_i|\cdot|\nei|^2) \\ + \sum\nolimits_{e_j\in E}(\mathbb{1}(e_j \text{ is adjacent to the sample})\cdot|e_j|\cdot|\nej|)).
	\end{multline*}
	From this, linearity of expectation, $\mathbb{E}[\mathbb{1}(e_i$ is sampled$)]=\frac{1}{|E|}$, and $\mathbb{E}[\mathbb{1}(e_j$ is adjacent to the sample$)]=\frac{|\nej|}{|E|}$, the \blue{expected} time complexity per sample hyperedge becomes
	$O(\frac{1}{|E|}$ $\sum_{e_i\in E}(|e_i|\cdot|\nei|^2))$.
	Hence, the \blue{expected} total time complexity for processing $s$ samples is \\$O(\frac{s}{|E|}\sum_{e_i\in E}(|e_i|\cdot|\nei|^2))$.\qedhere	
\end{proof}

\blue{We also obtain concentration inequalities of \methodAE (Theorem~\ref{thm:sampling_ver1:concentration}) using Hoeffding's inequality (Lemma~\ref{lem:hoeff}), and the inequalities particularly depend on the number of samples and the number of instances of each \motif.}
\blue{
\begin{lemma}[Hoeffding's Inequality~\cite{hoeffding1994probability}]
    \label{lem:hoeff}
    Let $X_1$, $X_2$, $\dots$, $X_n$ be independent random variables with $a_j\leq X_j\leq b_j$ for every $j\in \{1,2,\dots,n\}$. Consider the sum of random variables $X=X_1+\dots+X_n$, and let \(\mu=\mathbb{E}[X]\). Then for any $\tau>0$, we have 
    \[
    \Pr[|X-\mu|\geq \tau]\leq 2\exp\left(-\frac{2\tau^2}{\sum_{j=1}^n (b_j-a_j)^2}\right).
    \]
\end{lemma}
}

\blue{\begin{theorem}[Concentration Bound of \methodAE] \label{thm:sampling_ver1:concentration}
    Let \(d_{\max}[t]=\max_{e\in E[t]}|N_{e}|\) where $E[t]:=\bigcup_{h(e_i, e_j, e_j)=t}$  $\{e_i, e_j, e_k\}$. For any \(\epsilon, \delta>0\), if \(M[t]>0\) and the number of samples \(s>\frac{1}{18\epsilon^2}\left(\frac{|E|d_{\max}[t]^2}{M[t]}\right)^2 \log(\frac{2}{\delta})\), then \(\Pr(|\bar{M}[t]-M[t]|\geq M[t]\cdot\epsilon)\leq \delta\) holds for each \(t \in \{1,2,\dots,26\}\). 
\end{theorem}
\begin{proof}
    See Appendix~\ref{concentration_1:proof}.
\end{proof}
}

\IncMargin{0.5em}
\begin{algorithm}[t]
	\setstretch{1.25}
	\small
	\caption{\small \methodAW: Approximate \Motif Counting Based on \Hwedge Sampling}
	\label{sampling_ver2}
	\SetAlgoLined
	\SetKwInOut{Input}{Input}
	\SetKwInOut{Output}{Output}
	\nonl \hspace{-5mm} \Input{ \ \ (1) input hypergraph: $G=(V,E)$ \\ (2) \blue{line graph}: $\GT=(E,\PT, \omega)$ \\  (3) number of samples: $r$}
	\nonl \hspace{-5mm} \Output{estimated count of each \motif $t$'s instances: $\MHT$}
	$\MH\leftarrow$ map whose default value is $0$ \\
	\For{$n\leftarrow 1...r$ \blue{\textbf{(in parallel)}}\label{sampling_ver2:loop1}}{
		$\wedge_{ij}\leftarrow$ a uniformly random \hwedge \label{sampling_ver2:sample} \\ 
		\For{\textbf{each} hyperedge $e_k \in (N_{e_i} \cup N_{e_j} \setminus \{e_i,e_j\})$\label{sampling_ver2:loop2}}{
			$\MH[\hijk] \mathrel{+}= 1$
			\label{sampling_ver2:count}
		}
		\label{sampling_ver2:loop2:end}
	}
	\For{\textbf{each} \motif $t$\label{sampling_ver2:rescale:start}}{
		\eIf(\hfill $\triangleright$ \texttt{\color{blue}open \motifs}){17 $\leq$ t $\leq$ 22}{ 
			$\MHT \leftarrow \MHT \cdot \tfrac{|\wedge|}{2r}$ 
		}
		(\hfill $\triangleright$ \texttt{\color{blue}closed \motifs}){
			$\MHT \leftarrow \MHT \cdot \tfrac{|\wedge|}{3r}$
		}
	}\label{sampling_ver2:rescale:end}
	\textbf{return} $\MH$ 
\end{algorithm}
\DecMargin{0.5em}

\smallsection{\methodAW: \Hwedge Sampling (Algorithm~\ref{sampling_ver2}):}

\noindent \methodAW (Algorithm~\ref{sampling_ver2})  provides a better trade-off between speed and accuracy than \methodAE.
\blue{Differently from} \methodAE, which samples hyperedges, \methodAW is based on \hwedge sampling. 
It selects $r$ \hwedges uniformly at random with replacement (line~\ref{sampling_ver2:sample}), and for each sampled \hwedge $\wedge_{ij} \in \wedge$, it searches for all \motif instances that contain $\wedge_{ij}$ (lines~\ref{sampling_ver2:loop2}-\ref{sampling_ver2:loop2:end}). To this end, the hyperedges that are adjacent to $e_i$ or $e_j$ in the \blue{line graph} $\GT$ are considered (line~\ref{sampling_ver2:loop2}). For each such instance $\eijk$ of \motif $t$, the corresponding estimate $\MHT$ is incremented (line~\ref{sampling_ver2:count}). 
Lastly, each estimate $\MHT$ is rescaled so that it unbiasedly estimates $\MT$, as formalized in Theorem~\ref{thm:sampling_ver2:accuracy}.
To this end, it is multiplied by the reciprocal of the expected number of times that each instance of \motif $t$ is counted.\footnote{Note that each instance of open and closed \motifs contains $2$ and $3$ \hwedges, respectively.
	Each instance of closed \motifs is counted if one of the $3$ \hwedges in it is sampled, while that of open \motifs is counted if one of the $2$ \hwedges in it is sampled.
	Thus, \blue{in expectation}, each instance of open and closed \motifs is counted ${3r}/{|\wedge|}$ and ${2r}/{|\wedge|}$ times, respectively.} 

\begin{theorem}[Bias and Variance of \methodAW]
	\label{thm:sampling_ver2:accuracy}
	For every \motif t, Algorithm~\ref{sampling_ver2} provides an unbiased estimate $\MHT$ of the count $\MT$ of its instances, i.e.,	
	\begin{equation}
	\mathbb{E}[\MHT]=\MT. \label{sampling_ver2:bias}
	\end{equation}
	For every closed \motif $t$, the variance of the estimate is
	\begin{equation}
	\mathbb{V}\mathrm{ar}[\MHT] = \frac{1}{3r}\cdot\MT\cdot(|\wedge|-3)+\frac{1}{9r}\sum_{n=0}^{1}q_n[t]\cdot(n|\wedge|-9), \label{sampling_ver2:variance:closed}
	\end{equation}
	where $q_n[t]$ is the number of pairs of \motif $t$'s instances that share $n$ hyperwedges.
	For every open \motif $t$, the variance is
	\begin{equation}
	\mathbb{V}\mathrm{ar}[\MHT] = \frac{1}{2r} \cdot \MT\cdot(|\wedge|-2)+\frac{1}{4r}\sum_{n=0}^{1}q_n[t]\cdot(n|\wedge|-4). \label{sampling_ver2:variance:open} 
	\end{equation}
\end{theorem}
\begin{proof}
	See Appendix~\ref{sampling_ver2:proof}. 
\end{proof}

The time complexity of \methodAW is given in Theorem~\ref{thm:sampling_ver2:time}.
\begin{theorem}[Complexity of \methodAW] \label{thm:sampling_ver2:time}
	The \blue{expected} time complexity of Algorithm~\ref{sampling_ver2} is $O(\frac{r}{|\wedge|}\sum_{e_i\in E}(|e_i|\cdot|\nei|^2))$.
\end{theorem}
\begin{proof}
	Assume all sets and maps are implemented using hash tables.
	For a sample \hwedge $\wij$, \blue{ computing $\nei \cup \nej$ takes $O(|\nei|+|\nej|)$ time in expectation with uniform hash functions if we compute $\nei \cup \nej$ by checking whether each hyperedge $e\in \nej$ is also in $\nei$.} By Lemma~\ref{lemma:motif:time}, computing $\hijk$ for all considered \motif instances takes \blue{$O(\min(|e_i|, |e_j|)\cdot |\nei| + |\nej|)$} time \blue{in expectation}.
	Thus, the \blue{expected} time complexity for processing a sample $\wij$ is 
	$O(\min(|e_i|, |e_j|)\cdot (|\nei|+|\nej|))=O(|e_i|\cdot|\nei| + |e_j|\cdot|\nej|),$
	which can be written as
	\begin{multline*}
	O(\sum\nolimits_{e_i\in E}(\mathbb{1}(e_i \text{ is included in the sample})\cdot|e_i|\cdot|\nei|) \\ + \sum\nolimits_{e_j\in E}(\mathbb{1}(e_j \text{ is included in the sample})\cdot|e_j|\cdot|\nej|)).
	\end{multline*}
	From this, linearity of expectation, $\mathbb{E}[\mathbb{1}(e_i$ is included in the sample$)]=\frac{|\nei|}{|\wedge|}$, and $\mathbb{E}[\mathbb{1}(e_j$ is included in the sample$)]=\frac{|\nej|}{|\wedge|}$, the \blue{expected} time complexity per sample \hwedge is
	$O(\frac{1}{|\wedge|}\sum_{e_i\in E}(|e_i|\cdot|\nei|^2))$.
	Hence, the total time complexity for processing $r$ samples is $O(\frac{r}{|\wedge|}\sum_{e_i\in E}(|e_i|\cdot|\nei|^2))$.\qedhere
\end{proof}

\blue{Additionally, we derive concentration inequalities for \methodAW (Theorem~\ref{thm:sampling_ver2:concentration}), following a similar approach to that of Theorem~\ref{thm:sampling_ver1:concentration}, but with different minimum sample sizes for guaranteeing the same bound.}

\blue{
\begin{theorem}[Concentration Bound of \methodAW]\label{thm:sampling_ver2:concentration}
    Let \(d_{\max}[t]=\max_{e\in E_{[t]}}|N_{e}|\) where $E[t]:=\bigcup_{h(e_i, e_j, e_j)=t}$ $\{e_i, e_j, e_k\}$. For each $t\in \{1,\dots, 26\}$ such that $M[t]>0$ and for any \(\epsilon, \delta>0\), a sufficient condition of being \(\Pr(|\hat{M}[t]-M[t]|\geq M[t]\cdot\epsilon)\leq \delta\) is $r>\frac{1}{18\epsilon^2}\left(\frac{|\Lambda|d_{\max}[t]}{M[t]}\right)^2$ $\log(\frac{2}{\delta})$, if \motif $t$ is closed,
    and $r>\frac{1}{8\epsilon^2}\left(\frac{|\Lambda|d_{\max}[t]}{M[t]}\right)^2$ $\log(\frac{2}{\delta})$, if motif $t$ is open.
\end{theorem}	
\begin{proof}
   See Appendix~\ref{concentration_2:proof}.
\end{proof}
}

\smallsection{Comparison of \methodAE and \methodAW:} \label{compare:sampling}
Empirically, \methodAW provides a better trade-off between speed and accuracy than \methodAE, as presented in Section~\ref{sec:exp:algo}. We provide an analysis that supports this observation. 
Assume that the numbers of samples in both algorithms are set so that $\alpha=\frac{s}{|E|}=\frac{r}{|\wedge|}$.
For each \motif $t$,
since both estimates $\MBT$ of \methodAE and $\MHT$ of \methodAW are unbiased (see Eqs.~\eqref{sampling_ver1:bias} and \eqref{sampling_ver2:bias}), we only need to compare their variances.
By Eq.~\eqref{sampling_ver1:variance}, $\mathbb{V}\mathrm{ar}[\MBT]=O(\frac{\MT+p_1[t]+p_2[t]}{\alpha})$, and by Eq.~\eqref{sampling_ver2:variance:closed} and Eq.~\eqref{sampling_ver2:variance:open}, $\mathbb{V}\mathrm{ar}[\MHT]=O(\frac{\MT+q_1[t]}{\alpha})$.
By definition, $q_1[t]\leq p_2[t]$, and thus $\frac{\MT+q_1[t]}{\alpha} \leq \frac{\MT+p_1[t]+p_2[t]}{\alpha}$. 
Moreover, in real-world hypergraphs, $p_1[t]$ tends to be several orders of magnitude larger than the other terms (i.e., $p_2[t]$, $q_1[t]$, and $\MT$), and thus $\MBT$ of \methodAE tends to have  larger variance (and thus larger estimation error) than $\MHT$ of \methodAW.
Despite this fact, as shown in Theorems~\ref{thm:sampling_ver1:time} and \ref{thm:sampling_ver2:time}, \methodAE and \methodAW have the same time complexity, $O(\alpha\cdot \sum_{e_i\in E}(|e_i|\cdot|\nei|^2))$.
Hence, \methodAW is expected to give a better trade-off between speed and accuracy than \methodAE, as confirmed empirically in Section~\ref{sec:exp:algo}.

\blue{
    Regarding the concentration lower bounds of the number of samples (Theorems~\ref{thm:sampling_ver1:concentration} and \ref{thm:sampling_ver2:concentration}), the ratio of the bound in \methodAE to that \methodAW is $\frac{|E|\cdot d_{\max} [t]}{|\wedge|}$ for each closed \motif $t$,  and $\frac{4|E|\cdot d_{\max} [t]}{9|\wedge|}$ for each open \motif $t$.} 
    \blue{In real-world datasets (refer to Table~\ref{dataset_table} in Section~\ref{sec:exp:settings}), the maximum value (across all \motifs) of $\frac{|E|\cdot d_{\max} [t]}{|\wedge|}$ varies from $5$  (in the \textit{contact-primary}) to $500$ (in the \textit{coauth-history}). That is, \methodAW requires fewer samples than \methodAE for the same bound, thereby supporting the empirical superiority of \methodAW over \methodAE.
    However, it is important to note a limitation in this comparison of bounds. Our concentration bounds may not be optimal since they are based on worst-case scenarios, relying on the term $d_{\max}[t]$.
}

\subsection{Parallel and On-the-fly Implementations}\label{sec:method:par_fly}

We discuss parallelization of \method and then on-the-fly computation of \blue{line graph}s.

\smallsection{Parallelization:}
All versions of \method and \blue{line-graph construction} are easily parallelized as highlighted in Algorithms~\ref{hyperwedge_counting}-\ref{sampling_ver2}.
Specifically, we can parallelize \blue{line-graph construction} and \methodE by letting multiple threads process different hyperedges (in line~\ref{hyperwedge_counting:loop1} of Algorithm~\ref{hyperwedge_counting} and line~\ref{exact_motif:loop1} of Algorithm~\ref{exact_motif_counting}) independently in parallel. 
Similarly, we can parallelize \methodAE and \methodAW by letting multiple threads sample and process different hyperedges (in line~\ref{sampling_ver1:sample} of Algorithm~\ref{sampling_ver1}) and \hwedges (in line~\ref{sampling_ver2:sample} of Algorithm~\ref{sampling_ver2}) independently in parallel. 
The estimated counts of the same \motif obtained by different threads are summed up only once before they are returned as outputs.
We present some empirical results in Section~\ref{sec:exp:algo}.

\smallsection{\Motif Counting without \blue{Line Graph}s:} \label{method:on_the_fly}
If the input hypergraph $G$ is large, computing its \blue{line graph} $\GT$ (Algorithm~\ref{hyperwedge_counting}) is time and space-consuming.
Specifically, building $\GT$ takes $O(\sum_{\wij\in\PT}|e_i\cap e_j|)$ time (see Lemma~\ref{lemma:projection:time}) and requires $O(|E|+|\wedge|)$ space, which often exceeds $O(\sum_{e_i\in E}$ $|e_i|)$ space required for storing $G$.
Thus, instead of precomputing $\GT$ entirely, we can build it incrementally while memoizing partial results within a given memory budget.
\blue{We apply this idea to \methodAW, resulting in the following two versions of the algorithm:
\begin{itemize}
    \item \textbf{\methodAWrandom}: This is a strightfoward application of the memoization idea to \methodAW (Algorithm~\ref{sampling_ver2}). We compute the neighborhood of a hyperedge $e_{i}\in E$ in $\GT$ (i.e., $\{(k,\omega(\wik)):k\in\nei\}$) only if (1) a \hwedge with $e_{i}$ (e.g., $\wij$) is sampled (in line~\ref{sampling_ver2:sample}) and (2) its neighborhood is not memoized.
    The computed neighborhood is memoized with priority based on the degree $|N_{e_{i}}|$ of $e_{i}$ in $\GT$. That is,
    if the memoization budget is exceeded, we evict the memoized neighborhood of hyperedges in decreasing order of their degrees in $\GT$ until the budget is met.
    This is because the neighborhood of high-degree hyperedges is frequently retrieved, despite a higher computational cost.
    According to our preliminary studies, this memoization scheme based on degree demonstrates faster speeds compared to memoizing the neighborhood of random hyperedges or least recently used (LRU) hyperedges.
    \item \textbf{\methodAWgreedy}: This is an improved version that considers the order in which hyperwedges are processed.
    It first collects a list $W$ of sampled hyperwedges and groups the hyperwedges consisting of the same hyperedge. Between the two hyperedges forming a hyperwedge, the one with the larger neighborhood is used to group the hyperwedge.
    The hyperwedges are processed group by group, and thus hyperwedges consisting of the same hyperedges are more likely to be processed consecutively, thereby increasing the chance of utilizing memoized neighborhoods before they are evicted.
    As a result, \methodAWgreedy is empirically faster than \methodAWrandom, as shown in Section~\ref{sec:exp:algo}.
\end{itemize}
For details of \methodAWrandom and \methodAWgreedy, refer to Algorithms~\ref{alg_memdeg} and \ref{alg_memdeg_adv}, respectively, in Appendix~\ref{appendix:on-the-fly}.
}

\section{\kijung{Extensions of \Motifs}}
\label{sec:extension}

\black{In this section, we explore two distinct approaches to extending the concept of \motifs.
We especially define ternary hypergraph motifs, which demonstrate consistent advantages for a variety of real-world applications.}

\subsection{\black{Extensions Beyond Binary}}
\label{sec:concept:tmotif}

\black{As defined in Section~\ref{sec:concept:motif}, \motifs describe overlapping patterns of three hyperedges solely based on the emptiness of the seven subsets derived from their intersections.
That is, for each subset, \motifs classify it into binary states, non-empty or empty, which corresponds to being colored or uncolored in Figure~\ref{motif_three_hyperedges}.
This coarse classification inevitably results in the loss of detailed information within the intersections.
Below, we introduce ternary hypergraph motifs, which mitigate this information loss by assigning ternary states to each subset based on its cardinality.
}


\smallsection{\black{Definition of \Tmotifs:}}
\black{Ternary hypergraph motifs (or \tmotifs in short) are the extension of \motifs, so as \motifs are, they are \blue{designed} for describing the overlapping pattern of three connected hyperedges.
Given an instance (i.e., three connected hyperedges) $\eijk$, \tmotifs describe its overlapping pattern by the cardinality of the following seven sets: (1) $e_i\setminus e_j \setminus e_k$, (2) $e_j\setminus e_k \setminus e_i$, (3) $e_k\setminus e_i \setminus e_j$, (4) $e_i\cap e_j \setminus e_k$, (5) $e_j\cap e_k \setminus e_i$, (6) $e_k\cap e_i \setminus e_j$, and (7) $e_i\cap e_j \cap e_k$.
\blue{Differently from \motifs, which consider two states for each subset (empty or non-empty), \tmotifs takes into account three (denoted by the `3' in \tmotifs.) states for each subset.}
Specifically, for each of these seven sets, we classify it into three states based on its cardinality $c$ as follows: (1) $c=0$, (2) $0< c \leq \theta$, (3) $c>\theta$, where $\theta\geq 1$.
Throughout this paper, we set the value of $\theta$ to $1$, and thus each of these seven sets is classified into one of three categories: empty, singleton, and multiple.
\blue{Equivalently, if we leave $\min(\lceil c/\theta \rceil,2)$ node
in each of the above subsets with cardinality $c$, \tmotifs can be defined based on the isomorphism between sub-hypergraphs formed by three connected hyperedges.}
Refer to Appendix~\ref{appendix:variants} for a discussion on \tmotifs with different values of $\theta$ and additional variants of \tmotifs.
Out of the $3^7$ possible patterns, $431$ \tmotifs remain if we exclude symmetric ones, \blue{those that cannot be obtained from distinct hyperedges}, and those that cannot be derived from connected hyperedges. Visual representations of \tmotifs 1-6, which are the six \tmotifs subdivided from \motif 1, are provided in Figure~\ref{fig:tmotif_example}. For a complete list of all 431 \tmotifs, refer to Appendix~\ref{appendix:variants}. 
}

\begin{figure}[t]
	\centering
	\includegraphics[width=\columnwidth]{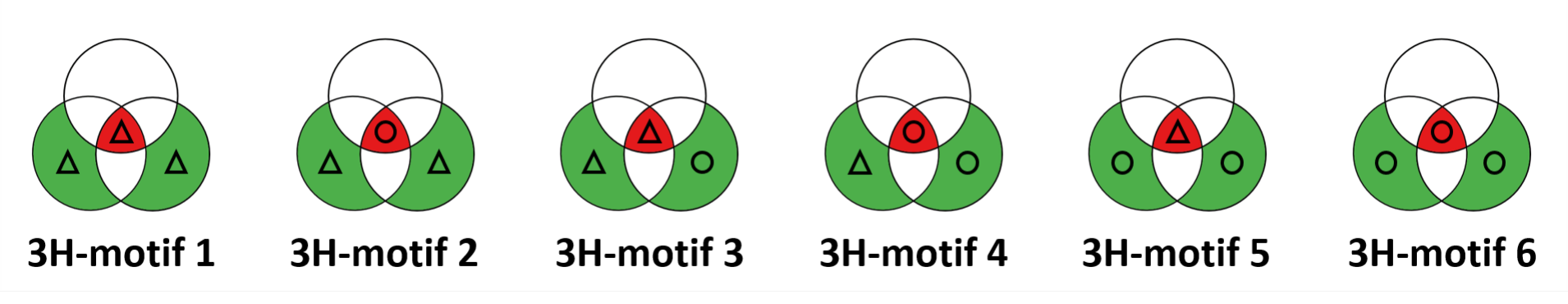} 
	\caption{\label{fig:tmotif_example} \black{The six \tmotifs subdivided from \motif 1. In each Venn diagram, uncolored regions are empty without containing any nodes. Colored regions with a triangle contain more than 0 and at most $\theta$ nodes, while colored regions with a circle contain more than $\theta$ nodes. Throughout this paper, we set $\theta$ to $1$.}
    }
\end{figure}

\smallsection{\black{Characterization using \Tmotifs:}}
\black{\Tmotifs can naturally substitute \motifs for characterizing hypergraphs, hyperedges, and nodes. By using \tmotifs, characteristic profiles (CPs), hyperedge profiles (HPs), and node profiles (NPs) become $431$-element vectors.}

\smallsection{\black{Counting \Tmotifs' Instances:}}
\black{To count instances of \tmotifs using the \method, the only necessary change is to replace $\hijk$ with $\thijk$, which provides the corresponding \tmotif for a given instance $\eijk$.
As formalized in Lemma~\ref{lemma:3H-motif:time}, $\thijk$ can be computed with the same time complexity as $\hijk$, and thus replacing $\hijk$ with $\thijk$ does not change the time complexity of all versions of \method.}


\begin{lemma}[\black{Complexity of Computing  $\thijk$}]\label{lemma:3H-motif:time}
	\black{Given the input hypergraph $G=(V,E)$ and its \blue{line} graph $\GT=(E,\wedge,\omega)$, for each 3h-motif instance $\eijk$, \blue{the expected time complexity for computing $\thijk$ is $O(\min(|e_i|,|e_j|,|e_k|))$.}} 
\end{lemma}
\begin{proof}
\black{Following the proof of Lemma~\ref{lemma:motif:time}, we can show that it takes $O(\min(|e_i|,|e_j|,|e_k|))$ time \blue{in expectation} to obtain the cardinalities of all the following sets: (1) $e_i\setminus e_j \setminus e_k$, (2) $e_j\setminus e_k \setminus e_i$, (3) $e_k\setminus e_i \setminus e_j$, (4) $e_i\cap e_j \setminus e_k$, (5) $e_j\cap e_k \setminus e_i$, (6) $e_k\cap e_i \setminus e_j$, and (7) $e_i\cap e_j \cap e_k$.
Based on the cardinality $c$ of each of the seven sets, it takes $O(1)$ time to classify it into (1) $c=0$, (2) $0< c \leq \theta$, and (3) $c>\theta$.
Classifying all seven sets, which takes $O(1)$ time, determines a specific \tmotif.
Thus, the \blue{expected} time complexity of computing $\thijk$ is $O(\min(|e_i|,|e_j|,|e_k|))$, which is same as that of computing $\hijk$.}
\end{proof}	

\smallsection{\black{Extensions Beyond Ternary:}}
\black{The concept of \tmotifs can be generalized to \kmotifs for any $k>3$ by classifying each of the seven considered sets into $k$ states.
For instance, for $k=4$, each set can be classified into four states based on its cardinality $c$ as follows: (1)  $c=0$, (2) $0< c \leq \theta_1$, (3) $\theta_1<c \leq \theta_2$, (4) $c>\theta_2$, where $\theta_2>\theta_1\geq 1$.
The number of \kmotifs increases rapidly with respect to $k$. Specifically, the number becomes $3,076$ for $k=4$, $14,190$ for $k=5$, and  $49,750$ for $k=6$, as derived in Appendix~\ref{appendix:kh-motif:generalized}.
In this study, we concentrate on \motifs and \tmotifs (i.e., $k=2$ and $k=3$), which are already capable of characterizing local structures in real-world hypergraphs, as evidenced by the empirical results in Section~\ref{sec:exp}.}

\subsection{Extensions Beyond Three Hyperedges}

The concept of \motifs is easily generalized to four or more hyperedges. For example, a \motif for four hyperedges can be defined as a binary vector of size $15$ indicating the emptiness of each region in the Venn diagram for four sets.
After excluding disconnected ones, symmetric ones, and \blue{those that cannot be obtained from distinct hyperedges}, there remain $1,853$ and $18,656,322$ \motifs for four and five hyperedges, respectively, as discussed in detail in Appendix~\ref{appendix:hmotif:generalized}.
This work focuses on the \motifs for three hyperedges, which are already capable of characterizing local structures of real-world hypergraphs, as shown empirically in Section~\ref{sec:exp}.

\section{Experiments}
\label{sec:exp}

\begin{table}[t]
	\begin{center}
		\caption{\label{dataset_table}Statistics of 11 real hypergraphs from 5 domains.} 
		\scalebox{0.64}{
			\begin{tabular}{l|c|c|c|c|c|c}
				\toprule
				\textbf{Dataset} & \textbf{$|V|$} & \textbf{$|E|$} & \textbf{$|\bar{e}|$*} & \textbf{$|\wedge|$} & \textbf{$|\bar{N_{e}}|$**} & \textbf{\# \Motifs}\\
				\midrule
				coauth-DBLP & 1,924,991 & 2,466,792 & \change{25} & 125M & 3,016 & 26.3B $\pm$ 18M\\
				coauth-geology& 1,256,385 & 1,203,895 & \change{25} & 37.6M & 1,935 & 6B $\pm$ 4.8M\\
				coauth-history& 1,014,734 & 895,439 & \change{25} & 1.7M & 855 & 83.2M\\
				\midrule
				contact-primary  & 242 & 12,704 & \change{5} & 2.2M & 916 & 617M\\
				contact-high & 327 & 7,818 & \change{5} & 593K & 439 & 69.7M\\
				\midrule
				email-Enron & 143 & 1,512 & \change{18} & 87.8K & 590 & 9.6M\\
				email-EU & 998 & 25,027 & \change{25} & 8.3M & 6,152 & 7B\\
				\midrule
				tags-ubuntu & 3,029 & 147,222 & \change{5} & 564M & 40,836 & 4.3T $\pm$ 1.5B\\
				tags-math & 1,629 & 170,476 & \change{5} & 913M & 49,559 & 9.2T $\pm$ 3.2B\\
				\midrule
				threads-ubuntu & 125,602 & 166,999 & \change{14} & 21.6M & 5,968 & 11.4B\\
				threads-math & 176,445 & 595,749 & \change{21} & 647M & 39,019 & 2.2T $\pm$ 883M\\
				\bottomrule
				\multicolumn{7}{l}{$*$ The maximum size of a hyperedge. $**$ \blue{The maximum degree in the line graph.}}
		\end{tabular}}
	\end{center}
\end{table}

\begin{table*}
	\centering
	\caption{\label{absolute_table} 
		Real-world and random hypergraphs have distinct distributions of \motif instances. We report the absolute counts of each \motif's instances in a real-world hypergraph from each domain and its corresponding random hypergraph. 
		To compare the counts in both hypergraphs, we measure the relative count (RC) of each \motif. 
		We also rank the counts, and we report each \motif's rank difference (RD) in the real-world and corresponding random hypergraphs.}
	\includegraphics[width=1.0\textwidth]{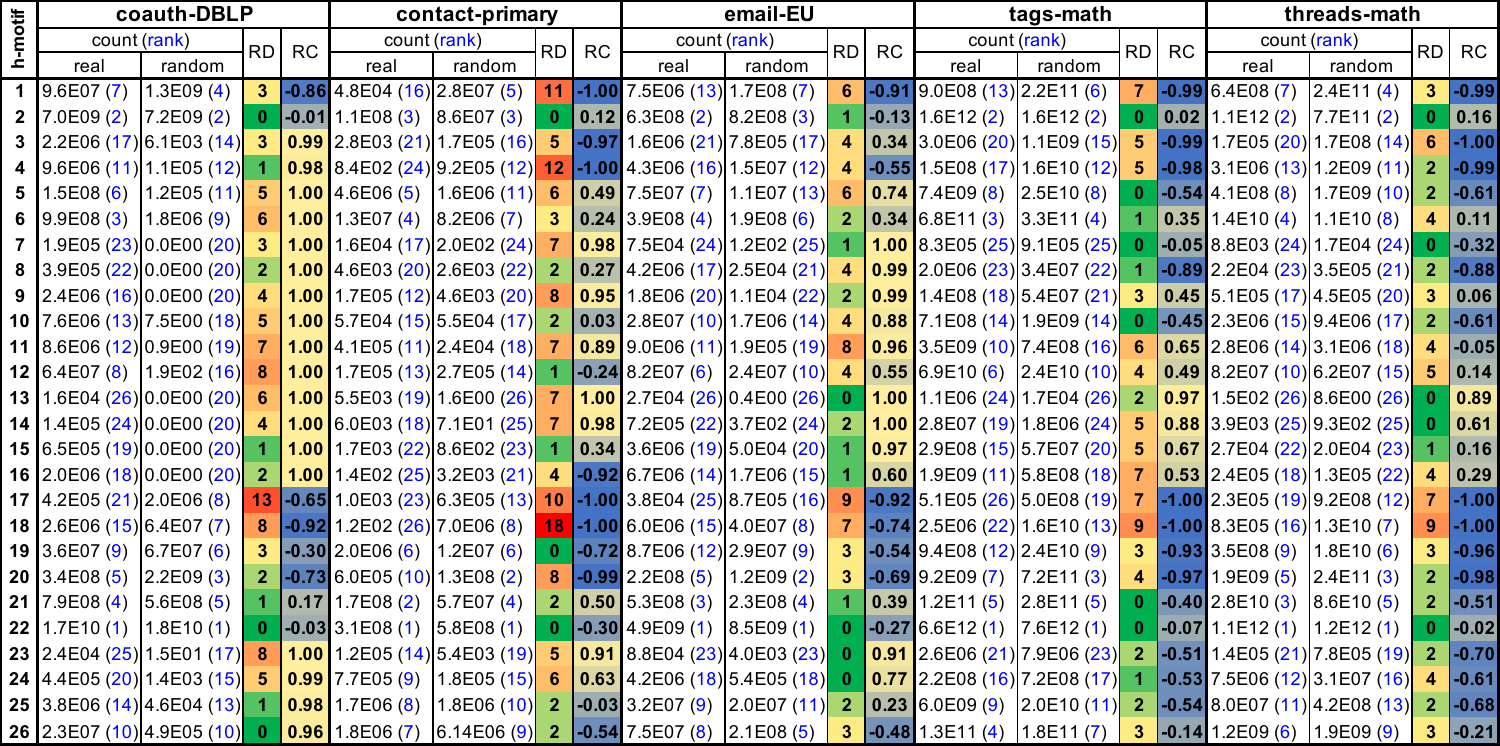}
\end{table*}

%

In this section, we review the experiments that we design for answering the following questions:
\begin{itemize}
	\item {\bf Q1. Comparison with Random:} Does counting instances of different \motifs reveal structural design principles of real-world hypergraphs distinguished from those of random hypergraphs? 
	\item {\bf Q2. Comparison across Domains:} Do characteristic profiles capture local structural patterns of hypergraphs unique to each domain?
 	\item \black{{\bf Q3. Comparison of Characterization Powers:} How well do \motifs, \tmotifs, and network motifs capture the structural properties of real-world hypergraphs?}
        \black{\item {\bf Q4. Machine Learning Applications:} Can \motifs and \tmotifs offer useful input features for machine learning applications?}
        	\item {\bf Q5. Further Discoveries:} What interesting discoveries can be uncovered by employing \motifs in real-world hypergraphs?
	\item {\bf Q6. Performance of Counting Algorithms:} How fast and accurate are the different versions of \method? 
\end{itemize}

\subsection{Experimental Settings}
\label{sec:exp:settings}
\smallsection{Machines:} We conducted all the experiments on a machine with an AMD Ryzen 9 3900X CPU and 128GB RAM. 

\smallsection{Implementations:}
We implemented every version of \method using C++ and OpenMP. \blue{For hash tables, we used the implementation named `unordered\_map' provided by the C++ Standard Template Library.} 

\smallsection{Datasets:} We used the following eleven real-world hypergraphs from five different domains:
\begin{itemize}
	\item \textbf{\texttt{co-authorship}} (coauth-DBLP, coauth-geology~\cite{sinha2015overview}, and coauth-history~\cite{sinha2015overview}): A node represents an author. A hyperedge represents all authors of a publication. 
	\item \textbf{\texttt{contact}} (contact-primary~\cite{stehle2011high} and contact-high~\cite{mastrandrea2015contact}): A node represents a person. A hyperedge represents a group interaction among individuals.
	\item \textbf{\texttt{email}} (email-Enron~\cite{klimt2004enron} and email-EU~\cite{leskovec2005graphs,yin2017local}): A node represents an e-mail account. A hyperedge consists of the sender and all receivers of an email.
	\item \textbf{\texttt{tags}} (tags-ubuntu and tags-math): A node represents a tag. A hyperedge represents all tags attached to a post.
	\item \textbf{\texttt{threads}} (threads-ubuntu and threads-math): A node represents a user. A hyperedge groups all users participating in a thread.
\end{itemize}
These hypergraphs are made public by the authors of \cite{benson2018simplicial}, and in Table~\ref{dataset_table} we provide some statistics of the hypergraphs after removing duplicated hyperedges.
We used $\methodE$ for the \textit{coauth-history} dataset, the \textit{threads-ubuntu} dataset, and all datasets from the \texttt{contact} and \texttt{email} domains. For the other datasets, we used \methodAW with $r=2,000,000$, unless otherwise stated. 
We used a single thread unless otherwise stated.
We computed CPs based on five hypergraphs randomized as described in Section~\ref{sec:concept:motif}.
\black{We computed CPs based on \motifs (instead of \tmotifs), unless otherwise stated.}

\subsection{Q1. Comparison with Random}
\label{sec:exp:random}

We analyze the counts of different \motifs' instances in real and random hypergraphs.
In Table~\ref{absolute_table}, we report the (approximated) count of each \motif $t$'s instances in each real hypergraph with the corresponding count averaged over five random hypergraphs obtained as described in Section~\ref{sec:concept:motif}.
\change{For each \motif $t$, we measure its relative count, which we define as $\frac{\MT- M_{rand}[t]}{\MT+  M_{rand}[t]}.$}
We also rank \motifs by the counts of their instances and examine the difference between the ranks in real and corresponding random hypergraphs.
As seen in the table, the count distributions in real hypergraphs are clearly distinguished from those of random hypergraphs.




\smallsection{\Motifs in Random Hypergraphs:}
We notice  that instances of \motifs $17$ and $18$ appear much more frequently in random hypergraphs than in real hypergraphs from all domains. For example, instances of \motif $17$ appear only about $510$ thousand times in the \textit{tags-math} dataset, while they appear about $500$ million times (about $\mathbf{980 \times}$ more often) in the corresponding randomized hypergraph. In the \textit{threads-math} dataset, instances of \motif $18$ appear about $830$ thousand times, while they appear about $13$ billion times (about $\mathbf{15,660\times}$ more often) in the corresponding randomized hypergraph. 
Instances of \motifs $17$ and $18$ consist of a hyperedge and its two disjoint subsets (see Figure~\ref{motif_three_hyperedges}).



\smallsection{\Motifs in Co-authorship Hypergraphs:}
We observe that instances of \motifs $10$, $11$, and $12$ appear more frequently in all three hypergraphs from the \texttt{co-au\\thorship} domain than in the corresponding random hypergraphs.
Although there are only about $190$ instances of \motif $12$ in the corresponding random hypergraphs, there are about $64$ million such instances (about $\mathbf{337,000\times}$ more instances) in the \textit{coauth-DBLP} dataset.
As seen in Figure~\ref{motif_three_hyperedges}, in instances of \motifs $10$, $11$, and $12$, a hyperedge is overlapped with the two other overlapped hyperedges in three different ways.


\smallsection{\Motifs in Contact Hypergraphs:}
Instances of h-mo\-tifs $9$, $13$, and $14$ are noticeably more common in both \texttt{contact} datasets than in the corresponding random hypergraphs.  
As seen in Figure~\ref{motif_three_hyperedges}, in instances of \motifs $9$, $13$, and $14$, hyperedges are tightly connected and nodes are mainly located in the intersections of all or some hyperedges.

\smallsection{\Motifs in Email Hypergraphs:} 
Both email datasets contain particularly many instances of \motifs $8$ and $10$, compared to the corresponding random hypergraphs.
As seen in Figure~\ref{motif_three_hyperedges}, instances of \motifs $8$ and $10$ consist of three hyperedges one of which contains the most nodes.

\begin{figure}[t!]  
	\centering
	\includegraphics[width=0.48\textwidth]{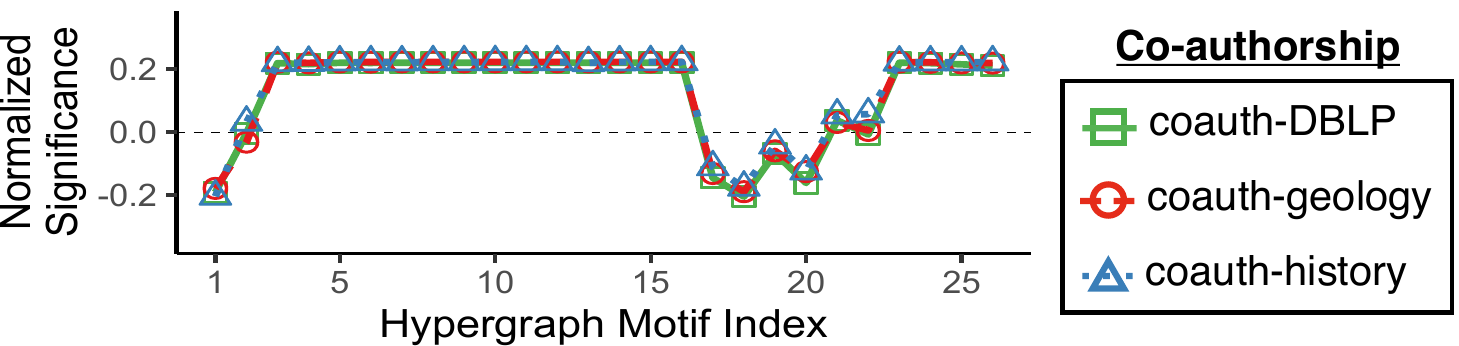}
	\includegraphics[width=0.48\textwidth]{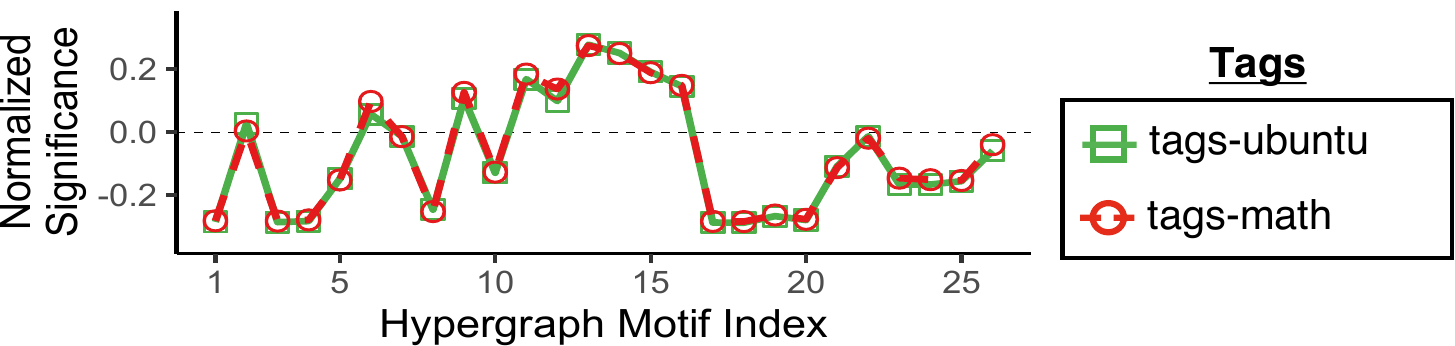} 
	\includegraphics[width=0.48\textwidth]{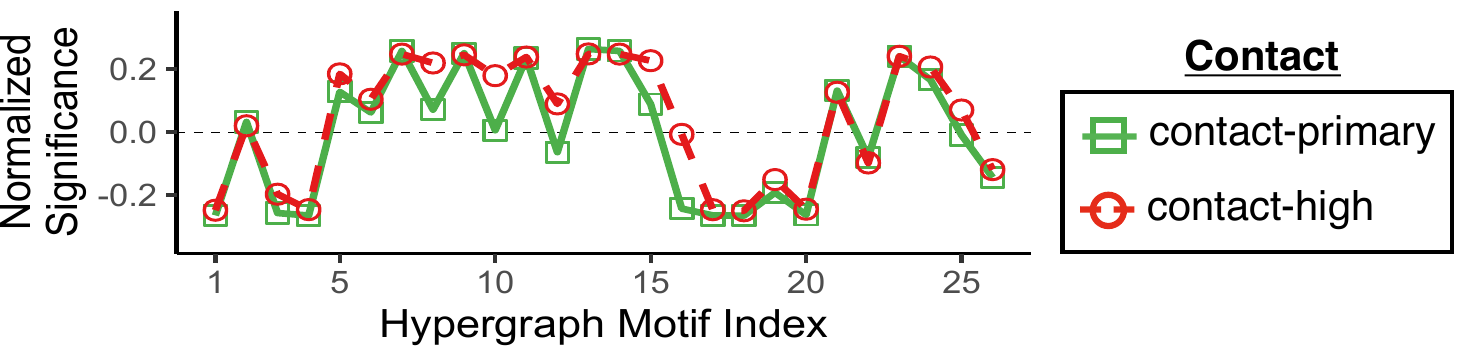}
	\includegraphics[width=0.48\textwidth]{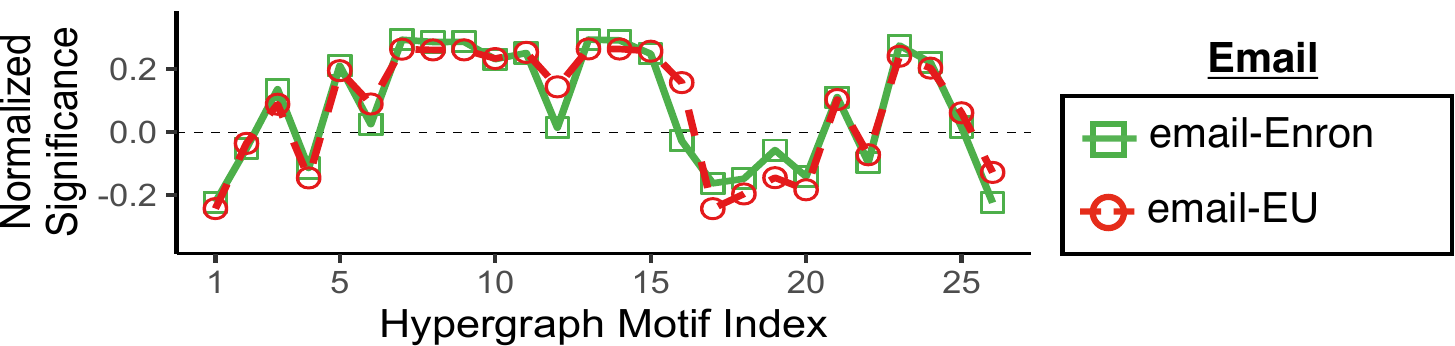}
	\includegraphics[width=0.48\textwidth]{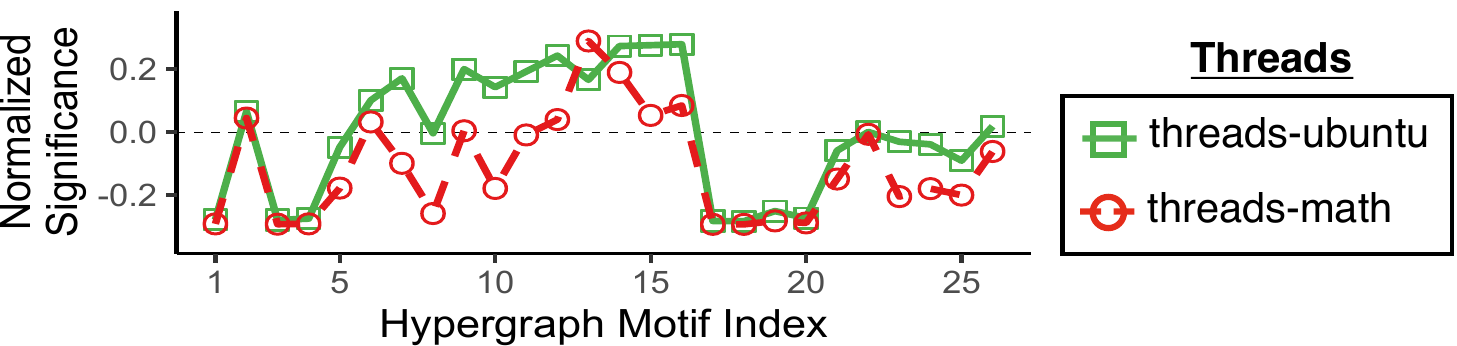}
	\caption{\label{fig:cp}
		Characteristic profiles (CPs) capture local structural patterns of real-world hypergraphs accurately.
		The CPs are similar within domains but different across domains. 
		Note that the significance of \motif 3 distinguishes 
		the contact hypergraphs from the email hypergraphs.
	}
\end{figure}

\smallsection{\Motifs in Tags Hypergraphs:}
In addition to instances of \motif $11$, which are common in most real hypergraphs, 
instances of \motif $16$, where all seven regions are not empty (see Figure~\ref{motif_three_hyperedges}), are particularly frequent in both \texttt{tags} datasets than in corresponding random hypergraphs.

\smallsection{\Motifs in Threads Hypergraphs:}
Lastly, in both data sets from the \texttt{threads} domain, instances of \motifs $12$ and $24$ are noticeably more frequent than expected from the corresponding random hypergraphs.

In Appendix~\ref{appendix:addExp-networkProperties}, we analyze how the significance of each \motif is correlated with the global structural properties of hypergraphs.

\subsection{Q2. Comparison across Domains}
\label{sec:exp:domain}

We compare the characteristic profiles (CPs) of the real-world hypergraphs.
In Figure~\ref{fig:cp}, we present the CPs (i.e., the significances of the $26$ \motifs) of each hypergraph.
As seen in the figure, hypergraphs from the same domains have similar CPs. 
Specifically, all three hypergraphs from the \texttt{co-authorship} domain share extremely similar CPs, even when the absolute counts of \motifs in them are several orders of magnitude different.
Similarly, the CPs of both hypergraphs from the \texttt{tags} domain are extremely similar.
However, the CPs of the three hypergraphs from the \texttt{co-authorship} domain are clearly distinguished by them of the hypergraphs from the  
\texttt{tags} domain.
While the CPs of the hypergraphs from the \texttt{contact} domain and the CPs of those from the \texttt{email} domain are similar for the most part, they are distinguished by the significance of \motif 3.
These observations confirm that CPs accurately capture local structural patterns of real-world hypergraphs.

\begin{figure}[t]  
	\centering
	\includegraphics[width=0.48\textwidth]{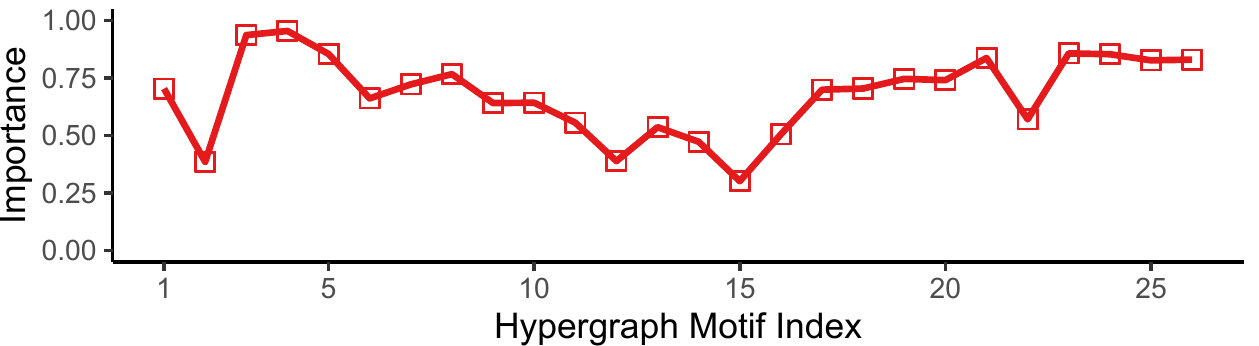}
	\caption{\label{sig_fig}
		\black{Importance of \motifs in differentiating hypergraph domains: All 26 \motifs contribute to distinguishing hypergraph domains, with each \motif having varying levels of importance.}
	}
\end{figure}

\begin{figure*}[t!]
	\centering     
	\hspace{-2mm}
	\subfigure[Similarity matrix based on network motifs
	]{
		\includegraphics[width=0.65\columnwidth]{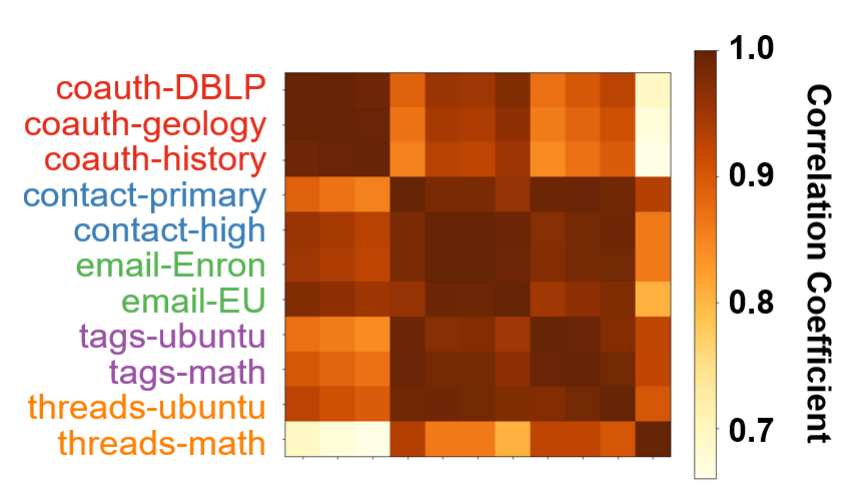}
	}
    \hspace{-2mm}
	\subfigure[Similarity matrix based on hypergraph motifs (\motifs)]{
		\includegraphics[width=0.65\columnwidth]{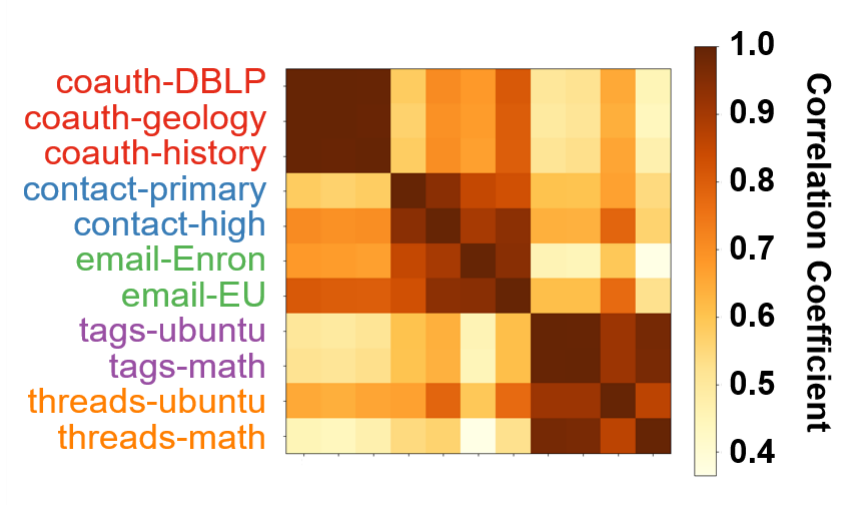}
	}
    \hspace{-2mm}
    \subfigure[\black{Similarity matrix based on ternary hypergraph motifs (\tmotifs)}
	]{
		\includegraphics[width=0.65\columnwidth]{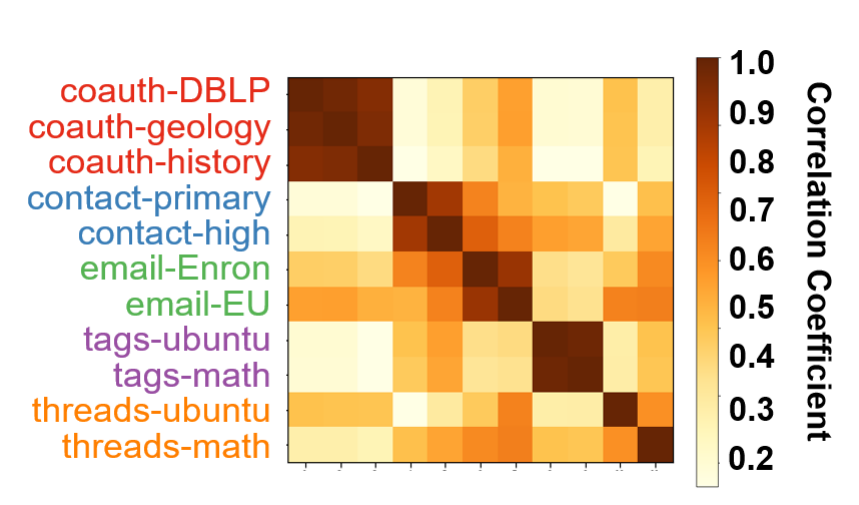}
	}

     \subfigure[\black{Numerical comparisons of the similarity matrices}
	]{
    	\scalebox{0.75}{
      \begin{tabular}{c|ccc}
                    \toprule
                    & \textbf{Network Motifs} & \textbf{H-Motifs} & \textbf{3H-Motifs}\\
                    \midrule
                    \multirow{1}{*}{\textbf{(1) Average Similarity within Domains}} & 0.988 & 0.978 & 0.932\\
                    \midrule
                    \multirow{1}{*}{\textbf{(2) Average Similarity across Domains}} & 0.919 & 0.654 & 0.370\\
                    \midrule
                    \multirow{1}{*}{\textbf{Gap between (1) and (2)}} & 0.069 & 0.324 & 0.562\\
                    
                    \midrule
                    \multirow{1}{*}{\textbf{Clustering Performance (NMI Score)}} & 0.678 & 0.905 & 1.000 \\
                    
                    \bottomrule
        	\end{tabular}
         }
	}
	\caption{\label{heatmap_fig}
		\black{Real-world hypergraphs from the same domain exhibit similar characteristic profiles (CPs), while those from different domains have distinct CPs. Notably, the CPs based on \motifs and \tmotifs capture local structural patterns more accurately than those based on network motifs, as supported numerically in the table.
  }}
\end{figure*}

\smallsection{\black{Importance of \Motifs:}}
\black{Since some \motifs can be more useful than others, we measure the importance of each \motif in distinguishing the domains of hypergraphs. 
We define the \textit{importance} of a \motif as its contribution to differentiating the domains of hypergraphs. 
The importance of each \motif $t$ is defined as: 
\begin{equation*}
importance[t] = 1-\frac{dist_{within}[t]}{dist_{across}[t]},
\end{equation*}
where $dist_{within}[t]$ is the average CP distance between hypergraphs from the same domain, and $dist_{across}[t]$ is the average CP distance between hypergraphs from different domains. As seen in Figure~\ref{sig_fig}, all $26$ \motifs have positive importances, indicating that all \motifs do contribute to distinguishing the domains of hypergraphs. Note that each \motif has different importance: some \motifs are extremely important (e.g., \motifs $3$, $4$, and $23$), while some are less important (e.g., \motifs $2$, $12$, and $15$).
\blue{
It is important to note that these importance scores should be interpreted with caution, as they may be overfitted given the limited number of datasets (specifically, the similarities observed in 7 within-domain pairs and 48 cross-domain pairs).
}
}



\subsection{Q3. Comparison of Characterization Powers}
\label{sec:exp:characterization_power}

\black{We compare the characterization power of \motifs, \tmotifs, and basic network motifs. Through this comparison, we demonstrate the effectiveness of \motifs and \tmotifs in capturing the structural properties of real-world hypergraphs.}

\smallsection{CPs Based on Network Motifs:}
In addition to characteristic profiles (CPs) based on \motifs and \tmotifs, we additionally compute CPs based on network motifs.
\blue{Specifically, we construct the incidence graph $G'$  (defined in Section~\ref{sec:concept:prelim}) of each hypergraph $G=(V,E)$.}
Then, we compute the CPs based on the network motifs consisting of $3$ to $5$ nodes, using \cite{bressan2019motivo}.\footnote{\black{Nine patterns can be obtained from incident graphs, which are bipartite graphs, and thus CPs based on network motifs are 9-element vectors.}} \black{Using each of the three types of CPs, we compute the similarity matrices (specifically, correlation coefficient matrices) of the real-world hypergraphs and provide them in Figure~\ref{heatmap_fig}.}

\smallsection{Comparison of Pearson Correlations:}
As seen in Figures~\ref{heatmap_fig}(a), \ref{heatmap_fig}(b) and \ref{heatmap_fig}(d), the domains of the real-world hypergraphs are distinguished more clearly by the CPs based on \motifs than by the CPs based on network motifs.
Numerically, when the CPs based on \motifs are used, the average correlation coefficient is $0.978$ within domains and $0.654$ across domains, and the gap is $0.324$.
However, when the CPs based on network motifs are used, the average correlation coefficient is $0.988$ within domains and $0.919$ across domains, and the gap is just $0.069$.
\black{As seen in Figures~\ref{heatmap_fig}(c) and \ref{heatmap_fig}(d), the hypergraph domains are distinguished even more distinctly differentiated by the CPs based on \tmotifs.
Using \tmotifs as a basis for the CPs results in significantly lower correlation coefficients between the contact and email domains, as well as between the tag and thread domains, allowing for a better distinction between these domains. Numerically,  when \tmotifs are used, the average correlation coefficient is $0.932$ within domains and $0.370$ across domains, and the gap is $0.562$.
These results support that \motifs and \tmotifs play a key role in capturing local structural patterns of real-world hypergraphs.
}

\smallsection{Comparison of Clustering Performances:}
\black{We further compare the characterization powers by evaluating clustering performance using each similarity matrix as the input for spectral clustering \cite{scikit-learn}. We set the target number of clusters to the number of hypergraph domains.
As summarized in Figure~\ref{heatmap_fig}(d), the NMI scores, where higher scores indicate better clustering performance, are $0.678$, $0.905$, and $1$ when network motifs, \motifs, and \tmotifs, respectively, are used as a basis for the CPs.
Notably, when \tmotifs are used, the hypergraph domains are perfectly classified into distinct clusters. These results confirm again the effectiveness of \motifs and \tmotifs in characterizing real-world hypergraphs.}

\begin{figure*}[t]
	\centering     
	\subfigure[\label{fig:dblp:year} Fraction of the instances of each \motif in the coauth-DBLP dataset over time.]{\includegraphics[width=0.71\textwidth]{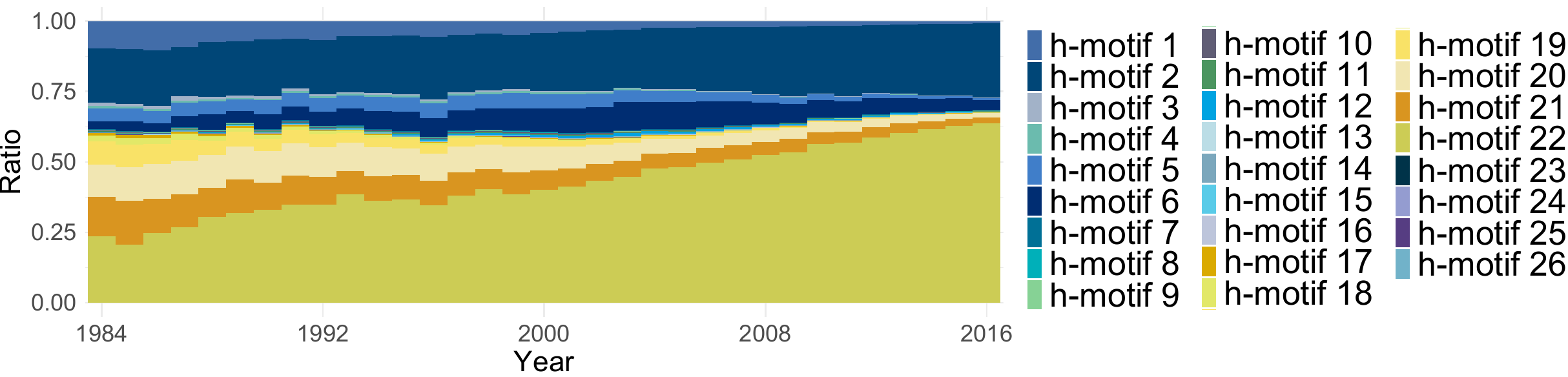}}
	\subfigure[\label{fig:dblp:openclosed} Open and closed \motifs.]{
		\includegraphics[width=0.01\textwidth]{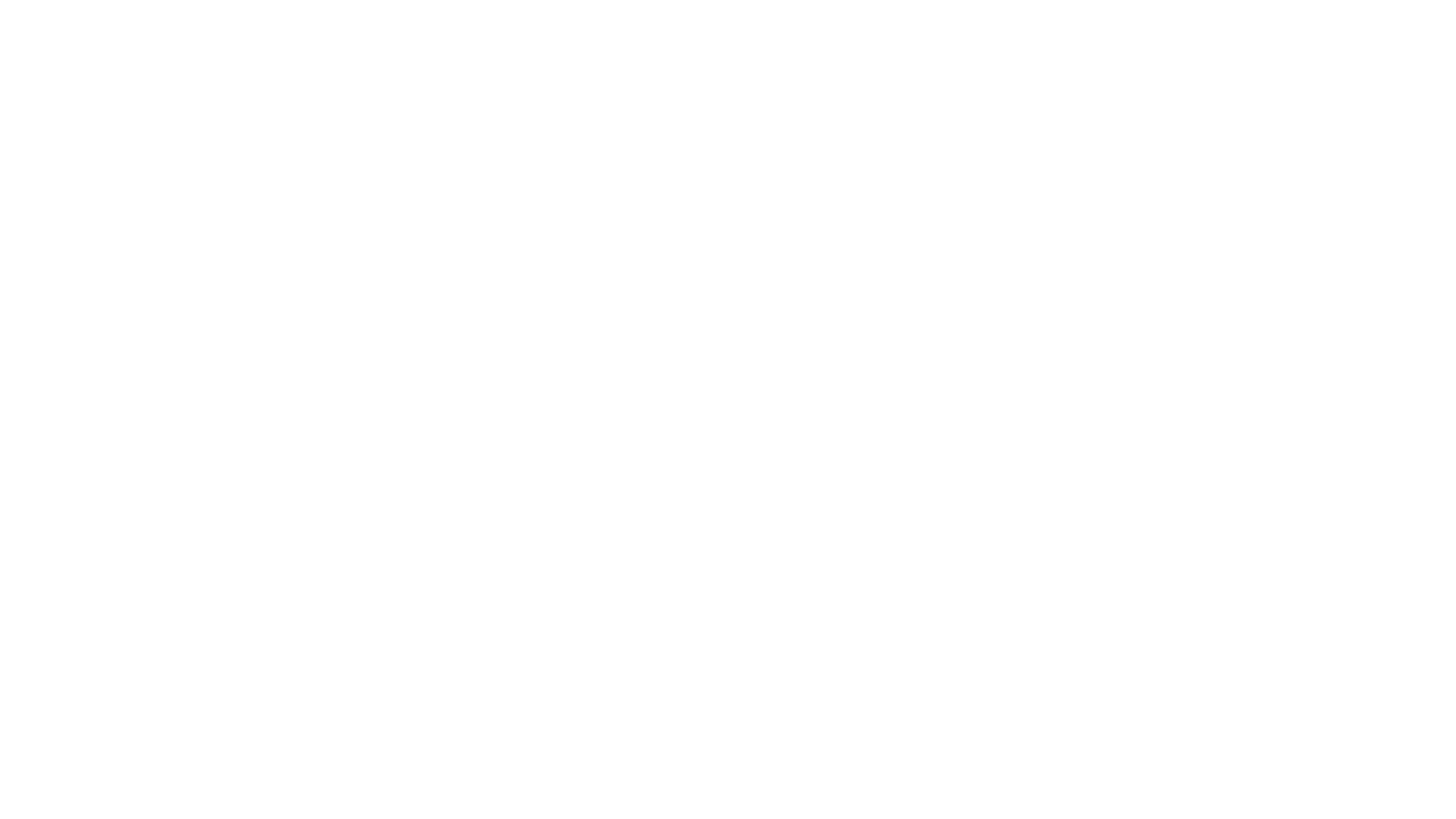}
		\includegraphics[width=0.24\textwidth]{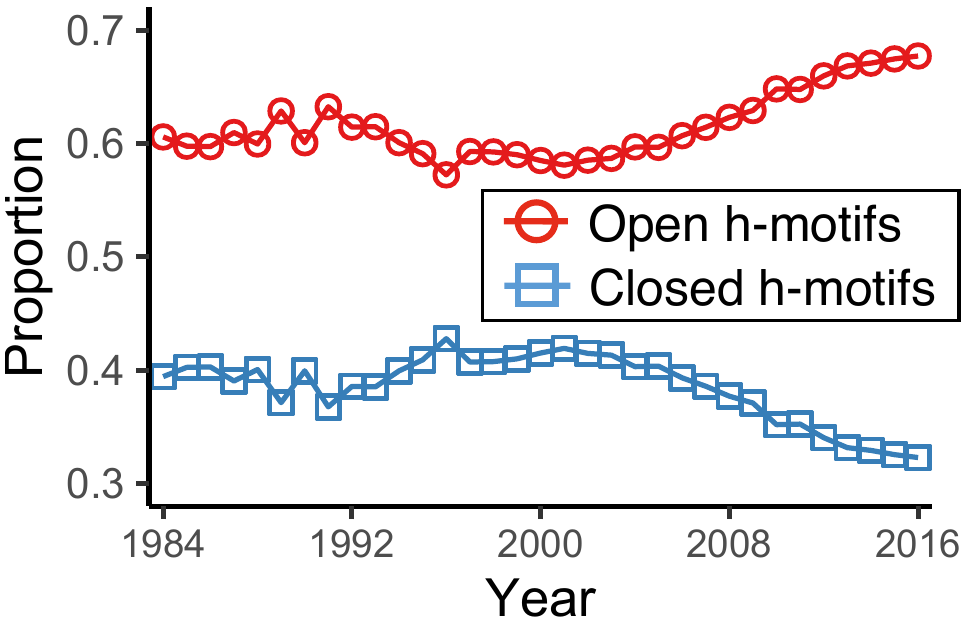}
		\includegraphics[width=0.01\textwidth]{empty.pdf}}
	\caption{\label{fig:dblp_year}\change{Trends in the formation of collaborations are captured by \motifs. (a) The fractions of the instances of \motifs 2 and 22 have increased rapidly.
	(b) The fraction of the instances of open \motifs has increased steadily since 2001. }}
\end{figure*}

\subsection{Q4. Machine Learning Applications}\label{sec:exp:applications}
\label{sec:exp:observations:prediction}

\black{We demonstrate that \motifs and \tmotifs provide useful input features for two machine learning tasks.}

\smallsection{\black{Hyperedge Prediction:}}
\black{We first consider the task of predicting future hyperedges in the seven real-world hypergraphs where \methodE completes within a reasonable duration.}
As in~\cite{yoon2020much}, we formulate this problem as a binary classification problem, aiming to classify real hyperedges and fake ones. To this end, we create fake hyperedges in both training and test sets by replacing some fraction of nodes in each real hyperedge with random nodes. Refer to Appendix~\ref{appendix:prediction} for detailed settings. Then, we train  classifiers using each of the following sets of input hyperedge features:
\begin{itemize}
    \item {{\bf HP26} ($\in\mathbb{R}^{26}$): HP based on \motifs.}
    \item {{\bf HP7} ($\in\mathbb{R}^{7}$): The seven features with the largest variance among those in HP based on \motifs.}
    \item \black{{{\bf THP} ($\in\mathbb{R}^{431}$): HP based on \tmotifs.}}
    \item {{\bf BASELINE} ($\in\mathbb{R}^{7}$): The mean, maximum, and minimum degree\footnote{The degree of a node $v$ is the number of hyperedges that $v$ is in.} and the mean, maximum, and minimum number of neighbors\footnote{The neighbors of a node $v$ is the nodes that appear in at least one hyperedge together with $v$.} of the nodes in each hyperedge and its size.}
\end{itemize}
\black{We employ XGBoost \cite{chen2016xgboost} as the classifier since it outperforms other classifiers, specifically logistic regression, random forest, decision tree, and multi-layer perception, on average, regardless of the feature sets used. Results with other classifiers can be found in Appendix~\ref{appendix:prediction}.} 

\black{We report the accuracy (ACC) and the area under the ROC curve (AUC) in each setting in Table~\ref{prediction_table}. Using HP26 and HP7, which are based on \motifs, yields consistently better predictions than using BASELINE, which is a baseline feature set. In addition, using THP, which is based on \tmotifs, leads to the best performance in almost all settings. These results suggest that \motifs provide informative hyperedge features, and \tmotifs provide even stronger hyperedge features for hyperedge prediction.}

\begin{table}[t]
	\begin{center}
		\caption{\label{prediction_table}
\black{
\Motifs and \tmotifs give informative hyperedge features.
The use of \motifs and \tmotifs for input features in HP26 and THP, respectively, consistently outperforms using the baseline features in BASELINE for predicting hyperedges in all datasets. Even when reducing the dimension of HM26 to that of BASELINE (i.e., using HP7), the accuracy of predictions using \motif-based features remains superior.
For each setting, the best result is in \textbf{bold} and the second best one is \underline{underlined}.
The standard deviations of all the results are smaller than $0.0001$.
}
            }
		\scalebox{0.68}{
			\begin{tabular}{cc|cccc}
				\toprule
				\multicolumn{2}{c|}{} & \textbf{HP26} & \textbf{HP7} & \textbf{THP} & \textbf{BASELINE}\\
				\midrule
				\multirow{2}{*}{ \textbf{coauth-DBLP}} & \textbf{ACC} & \underline{0.801} & 0.744 & \textbf{0.836} & 0.646\\
				& \textbf{AUC} & \underline{0.886} & 0.820 & \textbf{0.909} & 0.707\\
                \midrule
				\multirow{2}{*}{ \textbf{coauth-MAG-Geology}} & \textbf{ACC} & \underline{0.782} & 0.722 & \textbf{0.819} & 0.661\\
				& \textbf{AUC} & \underline{0.865} & 0.798 & \textbf{0.892} & 0.741\\
                \midrule
				\multirow{2}{*}{ \textbf{coauth-MAG-History}} & \textbf{ACC} & \underline{0.696} & 0.683 & \textbf{0.716} & 0.608\\
				& \textbf{AUC} & \underline{0.811} & 0.761 & \textbf{0.820} & 0.732\\
                \midrule
				\multirow{2}{*}{ \textbf{contact-primary-school}} & \textbf{ACC} & \underline{0.772} & 0.769 & \textbf{0.779} & 0.603\\
				& \textbf{AUC} & \underline{0.879} & 0.868 & \textbf{0.886} & 0.647\\
                \midrule
				\multirow{2}{*}{ \textbf{contact-high-school}} & \textbf{ACC} & \textbf{0.907} & 0.860 & \underline{0.904} & 0.585\\
				& \textbf{AUC} & \textbf{0.968} & 0.949 & \underline{0.967} & 0.641\\
                \midrule
				\multirow{2}{*}{ \textbf{email-Enron}} & \textbf{ACC} & \underline{0.815} & 0.725 & \textbf{0.827} & 0.633\\
				& \textbf{AUC} & \textbf{0.922} & 0.816 & \underline{0.921} & 0.701\\
                \midrule
				\multirow{2}{*}{ \textbf{email-Eu}} & \textbf{ACC} & \underline{0.911} & 0.878 & \textbf{0.920} & 0.702\\
				& \textbf{AUC} & \underline{0.972} & 0.954 & \textbf{0.977} & 0.781\\
				\bottomrule
			\end{tabular}
		}
	\end{center}
\end{table}

\begin{table}[t]
	\begin{center}
		\scriptsize
		\caption{\label{node_classification_table}
         \black{\Motifs and \tmotifs provide valuable input features for node classification, with \tmotifs showing particularly strong performance. The use of them for input features in NP26 and TNP  yields better classification results than using the baseline features in BASELINE.
         For each metric, the best result is in \textbf{bold} and the second best one is \underline{underlined}. The standard deviations of all the results are smaller than $0.0001$.}}
                \begin{tabular}{c|cccc} 
        				\toprule
        				& \textbf{NP26} & \textbf{NP7} & \textbf{TNP} & \textbf{BASELINE}\\
        				\midrule
        				\textbf{ACC} & \underline{0.682} & 0.545 & \textbf{0.723} & 0.659\\  
                        \textbf{AVG AUC} & \underline{0.952} & 0.901 & \textbf{0.967} & 0.950\\
                        \bottomrule
        		\end{tabular}
	\end{center}
\end{table}

\begin{figure*}[t]
    \centering     
    \includegraphics[width=0.35\textwidth]{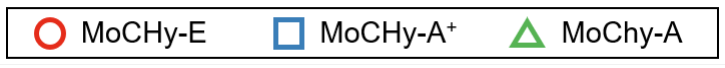}\\
    \subfigure[threads-ubuntu]{\includegraphics[width=0.16\textwidth]{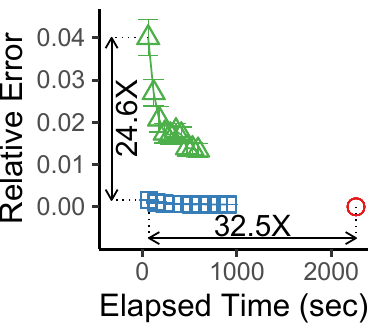}}
    \subfigure[email-Eu]{\includegraphics[width=0.16\textwidth]{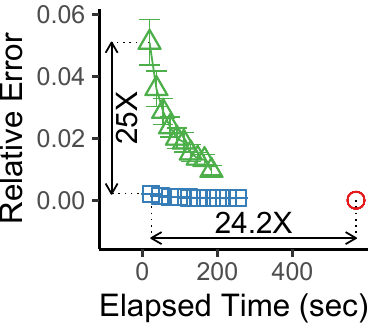}}
    \subfigure[contact-primary]{\includegraphics[width=0.16\textwidth]{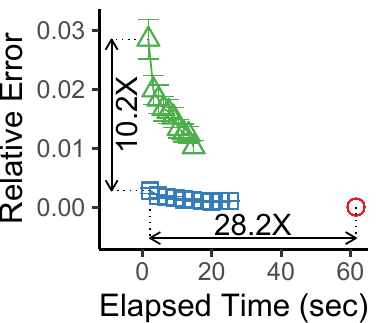}}
    \subfigure[coauth-history]{\includegraphics[width=0.16\textwidth]{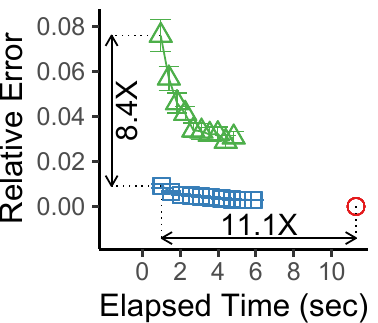}}
    \subfigure[contact-high]{\includegraphics[width=0.16\textwidth]{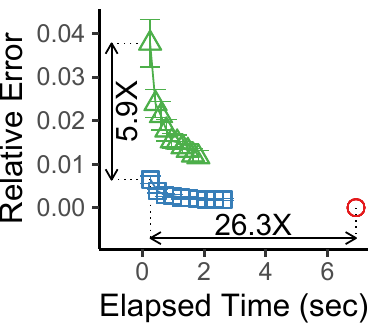}}
    \subfigure[email-Enron]{\includegraphics[width=0.16\textwidth]{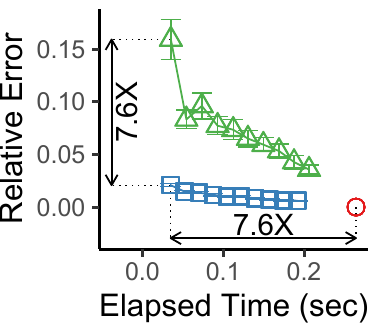}}
    \caption{\label{fig:tradeoff}
        \methodAWX gives the best trade-off between speed and accuracy. It yields up to $25 \times$ more accurate estimation than \methodAEX, and it is up to $32.5 \times$ faster than \methodEX.
		The error bars indicate $\pm$ $1$ standard error over $20$ trials.
	}
\end{figure*}

\begin{figure*}[t]
	\centering
	\includegraphics[width=0.15\textwidth]{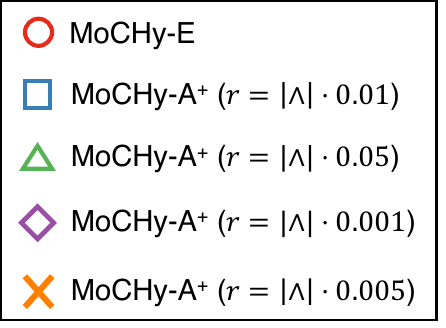}
	\subfigure[email-EU]{\includegraphics[width=0.276\textwidth]{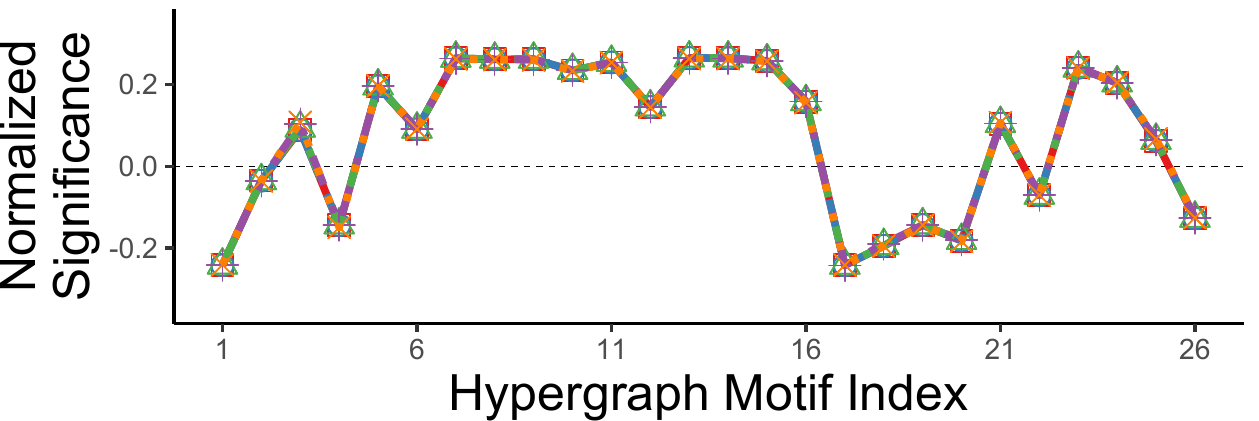}}
	\subfigure[contact-primary]{\includegraphics[width=0.276\textwidth]{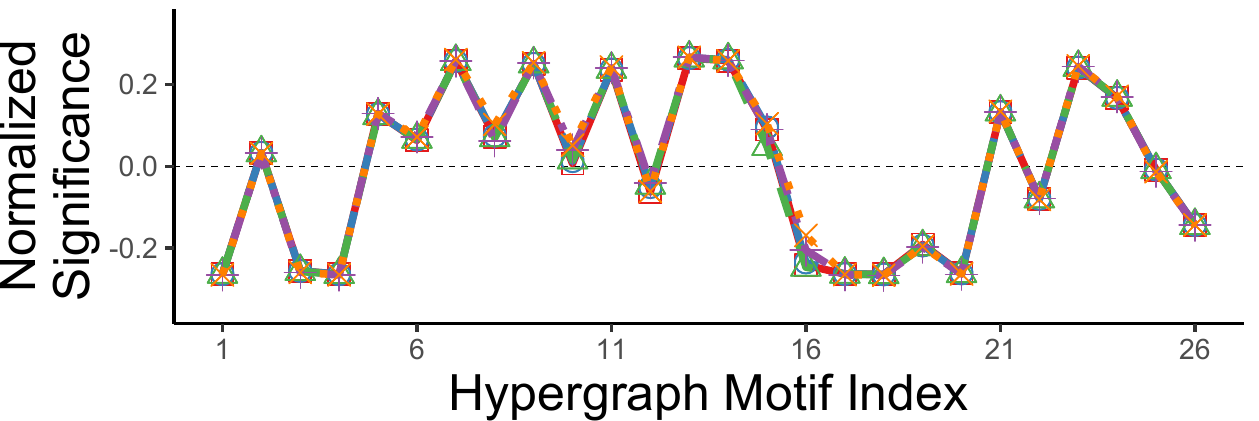}}
	\subfigure[coauth-history]{\label{sensitivity_fig:a}\includegraphics[width=0.276\textwidth]{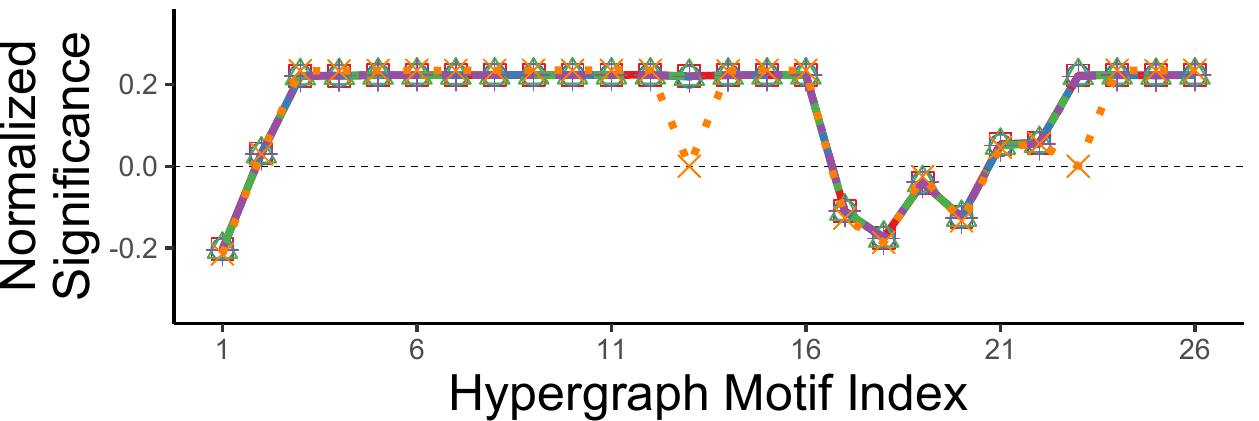}}
	\caption{\label{sensitivity_fig}Using \methodAWX, characteristic profiles (CPs) can be estimated accurately from a small number of samples.}
\end{figure*}

\smallsection{\black{Node Classification:}}
\black{As another machine learning application, we consider the task of node classification, where the label of each node is the hypergraph it belongs to.
Since we utilize all eleven real-world hypergraphs, each node can have one of eleven possible labels. 
We draw 100 nodes uniformly at random from each hypergraph, and we use $80\%$ of them for training and the remaining $20\%$ of them for testing. 
Refer to Appendix~\ref{appendix:node_classification} for detailed experimental settings.
We train four classifiers using each of the following sets of input node features:
\begin{itemize}
    \item {{\bf NP26} ($\in\mathbb{R}^{26}$): NP based on \motifs.}
    \item {{\bf NP7} ($\in\mathbb{R}^{7}$): The seven features with the largest variance among those in NP based on \motifs.}
    \item {{\bf TNP} ($\in\mathbb{R}^{431}$): NP based on 3h-motifs.}
    \item {{\bf BASELINE} ($\in\mathbb{R}^{7}$): The  
    node count, hyperedge count, average hyperedge size, average overlapping size, density \cite{hu2017maintaining}\footnote{The ratio between the hyperedge count and the node count.}, overlapness \cite{lee2021hyperedges}\footnote{The ratio between the sum of hyperedge sizes and the node count.}, and the number of hyperedges that contain the ego-node in each ego-network.}
\end{itemize}
For all feature sets, we use radial ego-networks as the ego-networks and XGBoost \cite{chen2016xgboost} as the classifier. This is because using radial ego-networks and XGBoost gives better classification results than using other types of ego-networks and other classifiers in most cases. Refer to Appendix~\ref{appendix:node_classification} for full experimental results with other types of ego-networks and other classifiers.}


\black{We report the accuracy (ACC) and the average area under the ROC curve (AVG AUC) in each setting in Table~\ref{node_classification_table}. Using TNP, which is based on \tmotifs, yields the best classification result.
Using NP26, which is based on \motifs, outperforms using BASELINE, which is a baseline feature set.
However, reducing the dimension of NP26 to that of BASELINE results in the worst performance.
These results demonstrate that \motifs and particularly \tmotifs provide effective input features for node classification, highlighting the importance of local structural patterns in hypergraphs for this task.}



\subsection{Q5. Further Observations}
\label{sec:exp:observations}
\label{sec:exp:observations:evolution}  
We analyze the evolution of the co-authorship Hypergraphs by employing \motifs.
The dataset contains bibliographic information on computer science publications.	
Using the publications in each year from $1984$ to $2016$, we create $33$ hypergraphs where each node corresponds to an author, and each hyperedge indicates the set of the authors of a publication.
Then, we compute the fraction of the instances of each \motif in each hypergraph to analyze patterns and trends in the formation of collaborations.
 As shown in Figure~\ref{fig:dblp_year}, over the 33 years, the fractions have changed with distinct trends.
 First, as seen in Figure~\ref{fig:dblp:openclosed}, the fraction of the instances of open \motifs has increased steadily since 2001, 
 indicating that collaborations have become less clustered, i.e., the probability that two collaborations intersecting with a collaboration also intersect with each other has decreased. 
 Notably, the fractions of the instances of \motif $2$ (closed) and \motif $22$ (open) have increased rapidly, accounting for most of the instances.

\begin{figure*}[t]
	\centering     
	\includegraphics[width=0.14\textwidth]{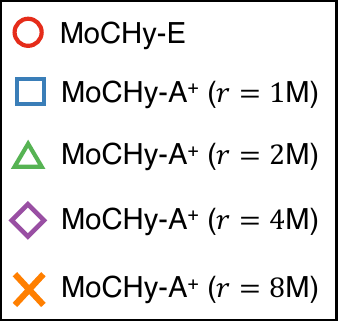} 
        \hspace{15pt}
	\subfigure[\label{fig:par:time} Elapsed time on threads-ubuntu]{\includegraphics[width=0.164\textwidth]{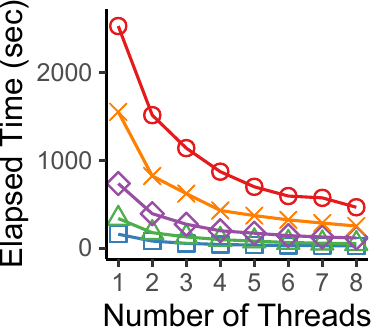}}
        \hspace{10pt}
	\subfigure[\label{fig:par:speedup} Speedup on threads-ubuntu]{\includegraphics[width=0.164\textwidth]{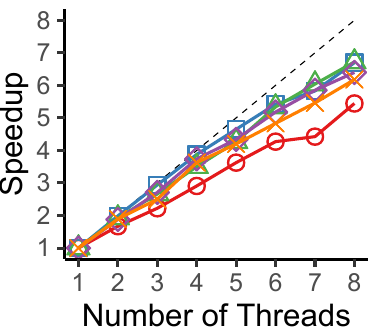}} 
        \hspace{10pt}
        \subfigure[\label{fig:par:time} Elapsed time on coauth-DBLP]{\includegraphics[width=0.164\textwidth]{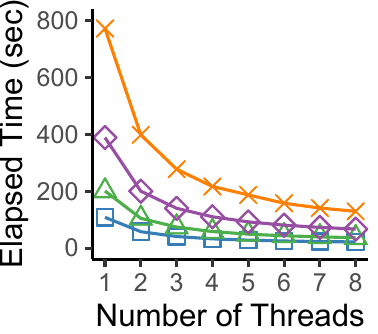}}
        \hspace{10pt}
	\subfigure[\label{fig:par:speedup} Speedup on coauth-DBLP]{\includegraphics[width=0.164\textwidth]{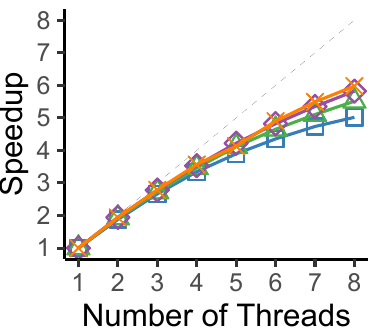}} 
	\caption{\label{fig:par} Both \methodEX and \methodAWX achieve significant speedups with multiple threads.}
\end{figure*}

\begin{figure*}[t]
	\centering     
        \includegraphics[width=0.6\textwidth]{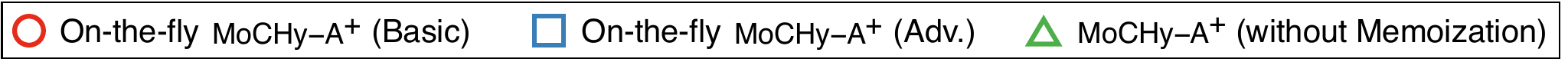}\\
	\subfigure[\label{fig:mem:time2} Speedup $r$=1M]{\includegraphics[width=0.195\textwidth]{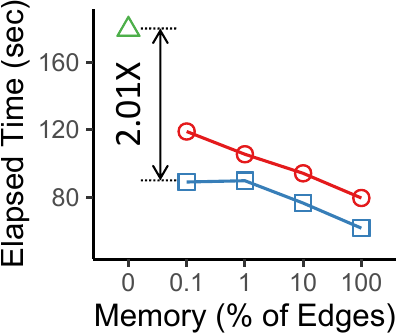}}
        \hspace{10pt}
        \subfigure[\label{fig:mem:time2} Speedup $r$=1M]{\includegraphics[width=0.195\textwidth]{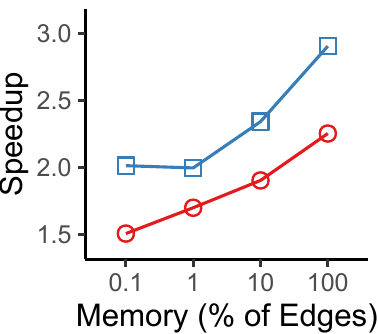}}
        \hspace{10pt}
        \subfigure[\label{fig:mem:time3} Elapsed Time $r$=8M]{\includegraphics[width=0.195\textwidth]{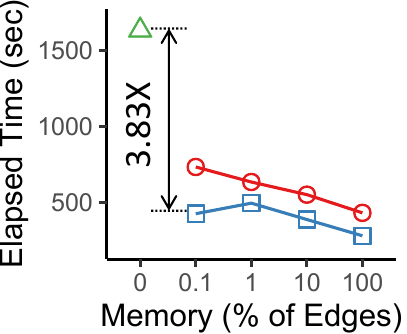}}
        \hspace{10pt}
        \subfigure[\label{fig:mem:time4} Speedup $r$=8M]{\includegraphics[width=0.195\textwidth]{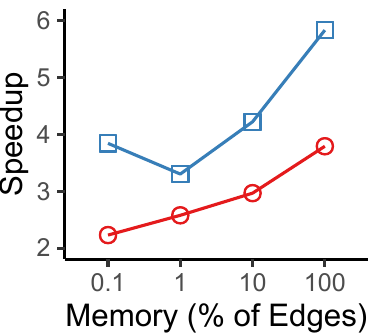}}
	\caption{\label{fig:mem}
		\blue{\methodAWrandom and \methodAWgreedy achieve substantial speed improvements, compared to \methodAW without memoization, even when memoizing a small fraction of line graphs
        Between the two methods,
        \methodAWgreedy is faster up to $1.72\times$ than \methodAWrandom, due to its carefully ordered processing scheme for sampled hyperwedges. 
        }
	}
\end{figure*}

\begin{figure}[t]
	\centering     
	\includegraphics[width=0.5\textwidth]{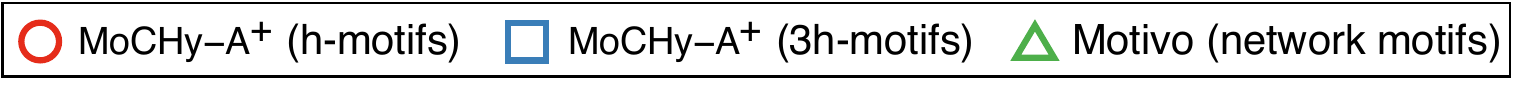} \\
	\subfigure[\label{fig:motivo:time} Elapsed Time]{\includegraphics[width=0.175\textwidth]{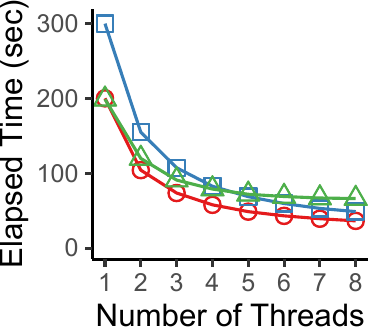}}
        \hspace{10pt}
	\subfigure[\label{fig:motivo:speedup} Speedup]{\includegraphics[width=0.175\textwidth]{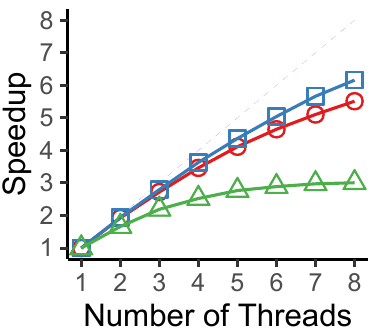}} 
	\caption{\label{fig:motivo} \blue{
 \methodAW for counting the instances of \motifs is consistently faster than Motivo~\cite{bressan2019motivo}, which counts the instances of network motifs of size up to $5$, across different numbers of threads. When counting \tmotifs, \methodAW is faster than Motivo with five or more threads. 
 This is attributed to the fact that \methodAW exhibits better speedup as the number of threads increases compared to Motivo. 
 }}
\end{figure}

\subsection{Q6. Performance of Counting Algorithms}
\label{sec:exp:algo}

We test the speed and accuracy of all versions of \method under various settings.
To this end, we measure elapsed time and relative error defined as
$$\frac{\sum_{t=1}^{26}|\MT-\MBT|}{\sum_{t=1}^{26}\MT} \text{  and  } \frac{\sum_{t=1}^{26}|\MT-\MHT|}{\sum_{t=1}^{26}\MT},$$
for \methodAE and \methodAW, respectively.
\blue{Unless otherwise stated, we use a single thread without the on-the-fly computation scheme.}

\smallsection{Speed and Accuracy:}
In Figure~\ref{fig:tradeoff}, we report the elapsed time and relative error of all versions of \method on the $6$ different datasets where \methodE terminates within a reasonable time.
The numbers of samples in \methodAE and \methodAW are set to $\{2.5\times k: 1\leq k \leq 10\}$ percent of the counts of hyperedges and \hwedges, respectively. \methodAW provides the best trade-off between speed and accuracy. For example, in the \textit{threads-ubuntu} dataset, \methodAW provides $24.6\times$ lower relative error than \methodAE, consistently with our theoretical analysis (see the last paragraph of Section~\ref{sec:method:approx}). Moreover, in the same dataset, \methodAW is $32.5\times$ faster than \methodE with little sacrifice on accuracy.

\smallsection{Effects of the Sample Size on CPs:}
In Figure~\ref{sensitivity_fig}, we report the CPs obtained by \methodAW with different numbers of hyperwedge samples on $3$ datasets.
Even with a smaller number of samples, the CPs are estimated near perfectly.


\smallsection{Parallelization:}
We measure the running times of the proposed method with different numbers of threads on the \textit{threads-ubuntu} \blue{and \textit{coauth-DBLP}} datasets. 
As seen in Figure~\ref{fig:par}, in both datasets, MoCHy achieves significant speedups with multiple threads. 
Specifically, with $8$ threads, \methodE and \methodAW ($r=1M$) achieve speedups of $5.4$ and $6.7$, respectively in \textit{threads-ubuntu} dataset.
\blue{In the \textit{coauth-DBLP} dataset, similar trends can be observed with speedups of $5.0$ and $6.0$ when using \methodAW for $r=1M$ and $r=8M$, respectively.
\methodE cannot be tested on the \textit{coauth-DBLP} dataset since it does not complete within a reasonable duration. 
}

\smallsection{Effects of On-the-fly Computation on Speed:}
We analyze the effects of the on-the-fly computation of \blue{line} graphs (discussed in Section~\ref{sec:method:par_fly}) on the speed of \methodAW under different memory budgets for memoization.
To this end, we use the \textit{coauth-DBLP} dataset, and we set the memory budgets so that up to \{0\%, 0.1\%, 1\%, 10\%, 100\%\} of the edges in the \blue{line} graph can be memoized.
\blue{
When the budget is $0\%$, we compute the neighbors of each hyperedge within the sampled hyperwedge every time, without precomputing or memoizing (a part of) the line graph. 
As shown in Figure~\ref{fig:mem}, both \methodAWrandom and \methodAWgreedy faster than \methodAW without memoization, and their speed tends to improve as the memory budget increases. 
In addition, \methodAWgreedy is consistently faster than \methodAWrandom across different memory budgets.
Specifically, it achieves up to $1.72\times$ reduced runtime, demonstrating the effectiveness of its carefully ordered processing schemes for sampled hyperwedges. 
}

\smallsection{Comparison with Network-motif Counting:}
\blue{
We assess the computational time needed for counting the instances of \motifs, \tmotifs, and network motifs on the \textit{coauth-DBLP} dataset, which is our largest dataset.
We employ \methodAW for both \motifs and \tmotifs, and for network motifs, we utilize Motivo~\cite{bressan2019motivo}, a recently introduced algorithm, to count the instances of network motifs up to size 5.
In all cases, we fix the sample size to 2 million. 
As shown in Figure~\ref{fig:motivo:time}, when counting instances of \motifs, \methodAW is consistently faster than Motivo across different numbers of threads, and the gap increases as the number of threads grows.
When it comes to counting \tmotifs, \methodAW is slower than Motivo with a single thread, but it becomes faster with five or more threads. 
This is attributed to \methodAW achieving significant speedups with more threads, compared to Motivo, as shown in Figure~\ref{fig:motivo:speedup}.
}




\section{Related Work}
\label{sec:related}

We review prior work on network motifs, algorithms for counting them, and hypergraphs.
While the definition of a network motif varies among studies, here we define it as a connected graph composed by a predefined number of nodes. 

\smallsection{Network Motifs}. Network motifs were proposed as a tool for understanding the underlying design principles and capturing the local structural patterns of graphs \cite{holland1977method,shen2002network,milo2002network}. 
The occurrences of motifs in real-world graphs are significantly different from those in random graphs \cite{milo2002network}, and they vary also depending on the domains of graphs \cite{milo2004superfamilies}.
The concept of network motifs has been extended to various types of graphs, including dynamic \cite{paranjape2017motifs} graphs, bipartite graphs \cite{borgatti1997network}, heterogeneous  graphs \cite{rossi2019heterogeneous}, and simplicial complexes \cite{benson2018simplicial,kim2023characterization,preti2022fresco} 
The occurrences of network motifs have been used in a wide range of graph applications: community detection \cite{benson2016higher,tsourakakis2017scalable,yin2017local,li2019edmot}, ranking \cite{zhao2018ranking}, graph embedding \cite{rossi2018higher,yu2019rum}, and graph neural networks \cite{lee2019graph}, to name a few.

\smallsection{Algorithms for Network Motif Counting.}
	We focus on algorithms for counting the occurrences of every network motif whose size is fixed or within a certain range \cite{ahmed2015efficient,ahmed2017graphlet,aslay2018mining,bressan2019motivo,chen2016general,han2016waddling,pinar2017escape}, while
	many are for a specific motif (e.g., the clique of size $3$) \cite{ahmed2017sampling,de2016triest,hu2013massive,hu2014efficient,jha2013space,kim2014opt,ko2018turbograph,pagh2012colorful,sanei2018butterfly,shin2017wrs,shin2020fast,tsourakakis2009doulion,wang2019rept,wang2017approximately}.	
	Given a graph,
	they aim to count rapidly and accurately the instances of motifs with $4$ or more nodes, despite the combinatorial explosion of the instances, using the following techniques:
	\begin{enumerate}
		{\setlength\itemindent{5pt}\item[(1)] {\bf Combinatorics:} For exact counting, combinatorial relations between counts have been employed \cite{ahmed2015efficient,pinar2017escape,paranjape2017motifs}.
		That is, prior studies deduce the counts of the instances of motifs from those of other smaller or equal-size motifs.}
		{\setlength\itemindent{5pt}\item[(2)] {\bf MCMC:} Most approximate algorithms sample motif instances from which they estimate the counts.
		\blue{Based on MCMC sampling, 
		the idea of performing a random walk over instances (i.e, connected subgraphs) until it reaches the stationarity
		 to sample an instance from a fixed probability distribution (e.g., uniform) has been employed \cite{bhuiyan2012guise,chen2016general,han2016waddling,saha2015finding,wang2014efficiently,matsuno2020improved}}.}
		{\setlength\itemindent{5pt}\item[(3)] {\bf Color Coding:} \blue{Instead of MCMC, color coding \cite{alon1995color} can be employed for sampling \cite{bressan2017counting,bressan2019motivo,bressan2021faster}. Specifically, prior studies color each node uniformly at random among $k$ colors, count the number of $k$-trees with $k$ colors rooted at each node, and use them to sample instances from a fixed probability distribution.}}
	\end{enumerate}
	In our problem, which focuses on \motifs with only $3$ hyperedges, sampling instances with fixed probabilities is straightforward without (2) or (3), 
	and the combinatorial relations on graphs in (1) are not applicable.
	In algorithmic aspects, we address the computational challenges discussed at the beginning of Section~\ref{sec:method} by answering
	(a) what to precompute (Section~\ref{sec:method:projection}),
	(b) how to leverage it (Sections~\ref{sec:method:exact} and \ref{sec:method:approx}), and 
	(c) how to prioritize it (Sections~\ref{sec:method:par_fly} and \ref{sec:exp:algo}), with formal analyses (Lemma~\ref{lemma:motif:time}; Theorems~\ref{thm:exact:time},~\ref{thm:sampling_ver1:time}, and \ref{thm:sampling_ver2:time}).


\smallsection{Hypergraph Mining}. 
Hypergraphs naturally represent group interactions  
occurring in a wide range of fields, including computer vision \cite{huang2010image,yu2012adaptive}, bioinformatics \cite{hwang2008learning}, circuit design \cite{karypis1999multilevel,ouyang2002multilevel}, social network analysis \cite{li2013link,yang2019revisiting}, cryptocurrency \cite{kim2022reciprocity}, 
and recommender systems \cite{bu2010music,li2013news}.
\black{There also has been considerable attention on machine learning on hypergraphs, including clustering \cite{agarwal2005beyond,amburg2019hypergraph,karypis2000multilevel,lee2022tri,zhou2007learning}, classification \cite{jiang2019dynamic,lee2022tri,sun2008hypergraph,yu2012adaptive}, hyperedge prediction \cite{benson2018simplicial,yoon2020much,zhang2018beyond,hwang2022ahp}, and anomaly detection \cite{lee2022hashnwalk}.
Recent studies on real-world hypergraphs revealed interesting patterns commonly observed across domains, including (a) global structural properties (e.g., giant connected components and small diameter) \cite{do2020multi,ko2022growth,bu2023hypercore} and their temporal evolution (e.g., shrinking diameter) \cite{ko2022growth}; (b) structural properties of ego-networks (e.g., density and overlapness) \cite{lee2021hyperedges} and their temporal evolution (e.g., decreasing rates of novel nodes) \cite{comrie2021hypergraph}; and (c) temporal patterns regarding arrivals of the same or overlapping hyperedges \cite{benson2018sequences,choo2022persistence,cencetti2021temporal}.}
Notably, Benson et al. \cite{benson2018simplicial} studied how several local features, including edge density, average degree, and probabilities of simplicial closure events for $4$ or less nodes\footnote{The emergence of the first hyperedge that includes a set of nodes each of whose pairs co-appear in previous hyperedges. The configuration of the pairwise co-appearances affects the probability.}, differ across domains.
Our analysis using \motifs is complementary to these approaches in that it (1) captures local patterns systematically without hand-crafted features, (2) captures static patterns without relying on temporal information, and (3) naturally uses hyperedges with any number of nodes without decomposing them into small ones. 

\black{Recently, there has been an extension of hypergraph motifs to temporal hypergraphs, which evolve over time \cite{lee2021thyme+,lee2023temporal}. This extension introduces 96 temporal hypergraph motifs (TH-motifs) that capture not only the overlapping patterns but also the relative order among three connected hyperedges. This extension has been shown to improve the characterization power of \motifs in hypergraph classification and hyperedge prediction tasks. Along with the concept of TH-motifs, a family of algorithms has been proposed for the exact and approximate counting of TH-motifs. The focuses of the algorithms are the dynamic update of the \blue{line graph} over time and the prioritized sampling of time intervals for estimation.
It is important to note that this conceptual and algorithmic extension requires temporal information as input, and is orthogonal to our extension to \tmotifs, which only requires topological information.
}

\section{Conclusions and Future Directions}
\label{sec:summary}

In this section, we present conclusions and future research directions.

\subsection{Conclusions}
In this work, we introduce hypergraph motifs (\motifs), and their extensions, ternary hypergraph motifs (\tmotifs). Using them, we investigate the local structures of $11$ real-world hypergraphs from $5$ different domains. We summarize our contributions as follows:
\begin{itemize}
	\item {\bf Novel Concepts:} We define 26 \motifs, which describe connectivity patterns of three connected hyperedges in a unique and exhaustive way, independently of the sizes of hyperedges (Figure~\ref{motif_three_hyperedges}). \black{We extend this concept to 431 \tmotifs, enabling a more specific differentiation of local structures (Figure~\ref{fig:tmotif_example}).}
	\item {\bf Fast and Provable Algorithms:} We propose $3$ parallel algorithms for (approximately) counting every \motif's instances, and we theoretically and empirically analyze their speed and accuracy. Both approximate algorithms yield unbiased estimates (Theorems~\ref{thm:sampling_ver1:accuracy} and \ref{thm:sampling_ver2:accuracy}), and especially the advanced one is up to $32\times$ faster than the exact algorithm, with little sacrifice on accuracy (Figure~\ref{fig:tradeoff}).
	\item {\bf Discoveries in $11$ Real-world Hypergraphs:} We confirm the efficacy of \motifs and \black{\tmotifs} by showing that local structural patterns captured by them are similar within domains but different across domains (Figures \ref{fig:cp} and \ref{heatmap_fig}).
    \item {\bf Machine Learning Applications:} 
    \black{Our experiments have shown that h-motifs are effective in extracting features for hypergraphs, hyperedges, and nodes in tasks such as hypergraph clustering, hyperedge prediction, and node classification. Furthermore, using 3h-motifs has been demonstrated to improve the feature extraction capabilities, resulting in even better performances on these applications.}
\end{itemize}

\subsection{Future Research Directions}

\blue{Future directions include exploring the practical applications of \motifs and \tmotifs,
motivated by the numerous successful use cases of network motifs in practical applications. 
For example, network motifs have been used in the domain of biology, for identifying crucial interactions between proteins, DNA, and metabolites within biological networks~\cite{yeger2004network,ma2004extended}.
Another compelling example lies within mobile communication networks, where network motifs have been observed to significantly impact the efficiency of information delivery across users~\cite{zhang2022influence}.
In addition, network motifs are proven to be powerful tools for enhancing the performance of other practical applications including anomaly detection~\cite{noble2003graph,yuan2021higher} and recommendation~\cite{gupta2014real,zhao2019motif,sun2022motifs,cui2021motif}.
Furthermore, they are recognized as a useful ingredient when designing graph-related algorithms, such as graph neural networks~\cite{sankar2017motif,yu2022molecular} and graph clustering algorithms~\cite{yin2017local,benson2016higher}.
These examples demonstrate the substantial potential of \motifs in diverse applications, and notably, most of them are less explored in hypergraphs than in graphs.
}

\blue{
In Sections~\ref{sec:exp:characterization_power} and \ref{sec:exp:applications}, we demonstrated the critical role of \motifs and \tmotifs in enhancing performance across hypergraph learning tasks, including node classification and hyperedge prediction
We believe that the considered hypergraph learning tasks can be readily applied to practical applications~\cite{choe2023classification,antelmi2023survey}.
For example, to achieve effective educational management and evaluation, it is important to classify the academic performance (e.g., poor, medium, and excellent) of students (nodes) based on the associations (hyperedges) among them \cite{li2022multi}.
It is also crucial to classify fake news (nodes) based on the patterns of news consumption by users (hyperedges)~\cite{jeong2022nothing}.
Accurately identifying labels for objects (nodes) in images (hyperedges) containing multiple entities is a crucial task in computer vision~\cite{wu2020adahgnn}.
Refer to a survey \cite{antelmi2023survey} for a broader range of applications formulated as node classification on hypergraphs.
In addition, hyperedge prediction can be employed for identifying novel sets (e.g., outfits) of items (e.g., fashion items) to be purchased together~\cite{li2021hyperbolic} (b) suggesting novel combinations of ingredients for recipes~ \cite{zhang2018beyond}, (c) recommending new collaborations among researchers~\cite{liu2018context}, and (d) discovering groups of genes collaborating for specific biological functions \cite{nguyen2022sparse}.
As we demonstrated in Sections~\ref{sec:exp:characterization_power} and \ref{sec:exp:applications}, \motifs and \tmotifs serve as valuable tools for addressing such tasks, indicating their potential applicability in practical scenarios, which we leave for future work.
}

\blue{Other promising research directions include (a) extending \motifs and \tmotifs to complex and rich hypergraphs, such as labeled or heterogeneous hypergraphs, and (b) investigating alternative random hypergraph models for assessing the significance of \motifs and \tmotifs.}

\smallsection{Reproducibility:} The code and datasets used in this work are available at \url{https://github.com/jing9044/MoCHy-with-3h-motif}.


\bibliographystyle{abbrv}
\bibliography{geon}

\appendix

\section{Proof of Theorem~2}
\label{sampling_ver1:proof}
We let $X_{ij}[t]$ be a random variable indicating whether the $i$-th sampled hyperedge (in line~\ref{sampling_ver1:sample} of Algorithm~\ref{sampling_ver1}) is included in the $j$-th instance of \motif $t$ or not. That is, $X_{ij}[t]=1$ if the hyperedge is included in the instance, and $X_{ij}[t]=0$ otherwise. We let $\mBT$ be the number of times that \motif $t$'s instances are counted while processing $s$ sampled hyperedges. That is,
\begin{equation} 
\mBT := \sum_{i=1}^{s} \sum_{j=1}^{\MT}X_{ij}[t]. \label{eq:mbt}
\end{equation}
Then, by lines~\ref{sampling_ver1:scale:start}-\ref{sampling_ver1:scale:end} of Algorithm~\ref{sampling_ver1}, 
\begin{equation}
\MBT=\mBT\cdot \tfrac{|E|}{3s}. \label{eq:mbt:scale}
\end{equation}

\smallsection{Proof of the Bias of $\MBT$ (Eq.~\eqref{sampling_ver1:bias}):}
	Since each \motif instance contains three hyperedges, the probability that each $i$-th sampled hyperedge is contained in each $j$-th instance of \motif $t$ is 
	\begin{equation} 
	P[X_{ij}[t]=1] = \mathbb{E}[X_{ij}[t]]=\tfrac{3}{|E|}. \label{eq:Xij:exp}
	\end{equation}
	From linearity of expectation,
	\begin{equation*} 
	\mathbb{E}[\mBT] = \sum_{i=1}^{s} \sum_{j=1}^{\MT}\mathbb{E}[X_{ij}[t]]= \sum_{i=1}^{s} \sum_{j=1}^{\MT}\frac{3}{|E|} = \frac{3s\cdot \MT}{|E|}. 
	\end{equation*}
	Then, by Eq.~\eqref{eq:mbt:scale}, $\mathbb{E}[\MBT] = \tfrac{|E|}{3s} \cdot \mathbb{E}[\mBT] = \MT$. \hfill $\qed$ \\

\smallsection{Proof of the Variance of $\MBT$ (Eq.~\eqref{sampling_ver1:variance}):} From Eq.~\eqref{eq:Xij:exp} and $X_{ij}[t]=X_{ij}[t]^2$, the variance of $X_{ij}[t]$ is
	\begin{equation}
		\mathbb{V}\mathrm{ar}[X_{ij}[t]] = \mathbb{E}[X_{ij}[t]^2] - \mathbb{E}[X_{ij}[t]]^2 = \frac{3}{|E|} - \frac{9}{|E|^2}. \label{eq:Xij:var}
	\end{equation}
	Consider the covariance between $X_{ij}[t]$ and $X_{i'j'}[t]$. 
	If $i = i'$, then from Eq.~\eqref{eq:Xij:exp},
	\begin{align}
			& \mathbb{C}\mathrm{ov}(X_{ij}[t], X_{i'j'}[t]) = \mathbb{E}[X_{ij}[t]\cdot X_{ij'}[t]]-\mathbb{E}[X_{ij}[t]]\mathbb{E}[X_{ij'}[t]] \nonumber\\
			& = P[X_{ij}[t]=1, X_{ij'}[t]=1] -\mathbb{E}[X_{ij}[t]]\mathbb{E}[X_{ij'}[t]] \nonumber\\
			& = P[X_{ij}[t]=1]\cdot P[X_{ij'}[t]=1|X_{ij}[t]=1] \nonumber\\
			& \hspace{10pt} - \mathbb{E}[X_{ij}[t]]\mathbb{E}[X_{ij'}[t]] \nonumber\\
			& = \frac{3}{|E|}\cdot \frac{l_{jj'}}{3} - \frac{9}{|E|^2} = \frac{l_{jj'}}{|E|} - \frac{9}{|E|^2}, \label{eq:Xij:cov}
			\vspace{-1mm}
	\end{align}
	where $l_{jj'}$ is the number of hyperedges that the $j$-th and $j'$-th instances share.	
	However, since hyperedges are sampled independently (specifically, uniformly at random with replacement), if $i \neq i'$, then $\mathbb{C}\mathrm{ov}(X_{ij}[t], X_{i'j'}[t])=0$. 
	This observation, Eq.~\eqref{eq:mbt}, Eq.~\eqref{eq:Xij:var}, and Eq.~\eqref{eq:Xij:cov} imply 
	\begin{align*}
			&\mathbb{V}\mathrm{ar}[\mBT] = \mathbb{V}\mathrm{ar}[\sum_{i=1}^{s} \sum_{j=1}^{\MT}X_{ij}[t]] \nonumber\\
			& = \sum_{i=1}^{s} \sum_{j=1}^{\MT} \mathbb{V}\mathrm{ar}[X_{ij}[t]] + \sum_{i=1}^{s} \sum_{j \neq j'}\mathbb{C}\mathrm{ov}(X_{ij}[t], X_{ij'}[t]) \nonumber\\
			&=s \cdot \MT \cdot  (\frac{3}{|E|} - \frac{9}{|E|^2}) + s \sum_{l=0}^{2} p_l[t] (\frac{l}{|E|} - \frac{9}{|E|^2}),
	\end{align*}
	where $p_l[t]$ is the number of pairs of \motif $t$'s instances sharing $l$ hyperedges.
	This and Eq.~\eqref{eq:mbt:scale} imply Eq.~\eqref{sampling_ver1:variance}. $\qed$

\section{Proof of Theorem~4}\label{concentration_1:proof}
\blue{Let $\tau:=M[t]\cdot\epsilon$ and $X_{ij}[t]$ be a random variable indicating whether the $i$-th sampled hyperedge (in line~\ref{sampling_ver1:sample} of Algorithm~\ref{sampling_ver1}) is included in the $j$-th instance of \motif $t$ or not. Also, let $\tilde{X}_{i}[t]=\frac{|E|}{3s}\sum_{j=1}^{M[t]}X_{ij}[t]$ (where the sum indicates the number of instances of h-motif t that contains i-th sampled hyperedge) so that \(\bar{M}[t]=\sum_{i=1}^s \tilde{X}_{i}[t]\). Note that \(0\leq \tilde{X}_i[t]\leq \frac{|E|d_{\max}[t]^2}{3s}\) holds for every \(i\).
     Since $\tilde{X}_{1}[t], \tilde{X}_2[t], \dots, \tilde{X}_s[t]$ are independent random variables and \(\mathbb{E}[\bar{M}[t]]= M[t]\) (Theorem~\ref{thm:sampling_ver1:accuracy}), we can apply Hoeffding's inequality (Lemma~\ref{lem:hoeff}):
    \begin{align*}
        \Pr&[|\bar{M}[t]-M[t]|\geq M[t]\cdot \epsilon]
        \\&\leq 2\exp\left(-\frac{2\epsilon^2M[t]^2}{s\cdot (d_{\max}[t]^2|E|/3s)^2}\right)
        \\&\leq 2\exp(-\frac{18s\epsilon^2M[t]^2}{|E|^2d_{\max}[t]^4})\leq \delta, 
    \end{align*}
    which implies the condition for the sample size. \hfill $\qed$.
    }

\section{Proof of Theorem 5}
\label{sampling_ver2:proof}
A random variable $Y_{ij}[t]$ denotes whether the $i$-th sampled \hwedge (in line~\ref{sampling_ver2:sample} of Algorithm~\ref{sampling_ver2}) is included in the $j$-th instance of \motif $t$.
That is, $Y_{ij}[t]=1$ if the sampled \hwedge is included in the instance, and $Y_{ij}[t]=0$ otherwise.
We let $\mHT$ be the number of times that \motif $t$' instances are counted while processing $r$ sampled \hwedges. That is,
\begin{equation}
\mHT := \sum_{i=1}^{r} \sum_{j=1}^{\MT}Y_{ij}[t] \label{eq:mht}
\end{equation}
We use $w[t]$ to denote the number of \hwedges included in each instance of \motif $t$.
That is,  
\begin{equation}\label{eq:wt}
w[t]:=\begin{cases}
2 & \text{if \motif $t$ is open,} \\
3 & \text{if \motif $t$ is closed.}
\end{cases}
\end{equation}

Then, by lines~\ref{sampling_ver2:rescale:start}-\ref{sampling_ver2:rescale:end} of Algorithm~\ref{sampling_ver2}, 
\begin{equation}
\MHT=
\mHT\cdot \tfrac{1}{w[t]}\cdot \tfrac{|\wedge|}{r}. \label{eq:mht:scale}
\end{equation}


\smallsection{Proof of the Bias of $\MHT$ (Eq.~\eqref{sampling_ver2:bias}):} Since each instance of \motif $t$ contains $w[t]$ hyperwedges, the probability that each $i$-th sampled \hwedge is contained in each $j$-th instance of \motif $t$ is 
	\begin{equation}
	P[Y_{ij}[t]=1]=\mathbb{E}[Y_{ij}[t]]=\tfrac{w[t]}{|\wedge|}. \label{eq:Yij:exp}
	\end{equation}
	From linearity of expectation,
	\begin{equation*} 
		\mathbb{E}[\mHT] = \sum_{i=1}^{r} \sum_{j=1}^{\MT}\mathbb{E}[Y_{ij}[t]]= \sum_{i=1}^{r} \sum_{j=1}^{\MT}\frac{w[t]}{|\wedge|} = \frac{w[t]\cdot r\cdot \MT}{|\wedge|}.
	\end{equation*}
	Then, by Eq.~\eqref{eq:mht:scale}, $\mathbb{E}[\MHT] =\mathbb{E}[\mHT]\cdot\tfrac{1}{w[t]}\cdot \tfrac{|\wedge|}{r}=\MT$. $\qed$ \\

\smallsection{Proof of the Variance of $\MHT$ (Eq.~\eqref{sampling_ver2:variance:closed} and Eq.~\eqref{sampling_ver2:variance:open}:} \\
	From Eq.~\eqref{eq:Yij:exp} and $Y_{ij}[t]=Y_{ij}[t]^2$, the variance of each random variable $Y_{ij}[t]$ is
	\begin{equation}
	\mathbb{V}\mathrm{ar}[Y_{ij}[t]] = \mathbb{E}[Y_{ij}[t]^2] - \mathbb{E}[Y_{ij}[t]]^2 = \frac{w[t]}{|\wedge|} - \frac{w[t]^2}{|\wedge|^2}. \label{eq:Yij:var}
	\end{equation}
	Consider the covariance between $Y_{ij}[t]$ and $Y_{i'j'}[t]$.
	If $i = i'$, then from Eq.~\eqref{eq:Yij:exp},
	\begin{align}
	& \mathbb{C}\mathrm{ov}(Y_{ij}[t], Y_{i'j'}[t])  = \mathbb{E}[Y_{ij}[t]\cdot Y_{ij'}[t]]-\mathbb{E}[Y_{ij}[t]]\mathbb{E}[Y_{ij'}[t]] \nonumber\\
	& = P[Y_{ij}[t]=1, Y_{ij'}[t]=1] -\mathbb{E}[Y_{ij}[t]]\mathbb{E}[Y_{ij'}[t]] \nonumber\\
	& = P[Y_{ij}[t]=1]\cdot P[Y_{ij'}[t]=1|Y_{ij}[t]=1] -\mathbb{E}[Y_{ij}[t]]\mathbb{E}[Y_{ij'}[t]] \nonumber\\
	& = \frac{w[t]}{|\wedge|}\cdot \frac{n_{jj'}}{w[t]} - \frac{w[t]^2}{|\wedge|^2} = \frac{n_{jj'}}{|\wedge|} - \frac{w[t]^2}{|\wedge|^2} \label{eq:Yij:cov},
	\end{align}
	where $n_{jj'}$ is the number of \hwedges that the $j$-th and $j'$-th instances share.	
	However, 	
	since \hwedges are sampled independently (specifically, uniformly at random with replacement), if $i\neq i'$, then $\mathbb{C}\mathrm{ov}(Y_{ij}[t], Y_{i'j'}[t])=0$. This observation, Eq.~\eqref{eq:mht}, Eq.~\eqref{eq:Yij:var}, and Eq.~\eqref{eq:Yij:cov} imply 
	\begin{align*}
		&\mathbb{V}\mathrm{ar}[\mHT] = \mathbb{V}\mathrm{ar}[\sum_{i=1}^{r} \sum_{j=1}^{\MT}Y_{ij}[t]] \nonumber\\
		& = \sum_{i=1}^{r} \sum_{j=1}^{\MT} \mathbb{V}\mathrm{ar}[Y_{ij}[t]] + \sum_{i=1}^{r} \sum_{j \neq j'}\mathbb{C}\mathrm{ov}(Y_{ij}[t], Y_{ij'}[t]) \nonumber\\
		& = r \cdot \MT \cdot(\frac{w[t]}{|\wedge|} - \frac{w[t]^2}{|\wedge|^2}) + r \sum_{n=0}^{1}q_n[t] \cdot(\frac{n}{|\wedge|} - \frac{w[t]^2}{|\wedge|^2}), \label{eq:mht:var}
	\end{align*}
	where $q_n[t]$ is the number of pairs of \motif $t$'s instances that share $n$ \hwedges. 
	This and Eq.~\eqref{eq:mht:scale} imply Eq.~\eqref{sampling_ver2:variance:closed} and Eq.~\eqref{sampling_ver2:variance:open}.	\hfill $\qed$

\section{Proof of Theorem 7}~\label{concentration_2:proof}

\blue{
Let $\tau:=M[t]\cdot\epsilon$ and $Y_{ij}[t]$ denotes whether the $i$-th sampled \hwedge (in line~\ref{sampling_ver2:sample} of Algorithm~\ref{sampling_ver2}) is included in the $j$-th instance of \motif $t$. Also, let $\tilde{Y}_{i}[t]=\frac{|\Lambda|}{w[t]s}\sum_{j=1}^{M[t]}Y_{ij}[t]$ with $w[t]$ defined in Eq.~\eqref{eq:wt}, where the sum indicates the number of instances of h-motif $t$ that contains the $i$-th sampled hyperwedge, so that \(\hat{M}[t]=\sum_{i=1}^r \tilde{Y}_{i}[t]\). Then, \(0\leq \tilde{Y}_i[t]\leq \frac{|\Lambda|d_{\max}[t]}{w[t]r}\) holds for every \(i\).
    Since $\tilde{Y}_{1}[t], \tilde{Y}_2[t], \dots, \tilde{Y}_r[t]$ are independent random variables and \(\mathbb{E}[\hat{M}[t]]= M[t]\) (in Theorem~\ref{thm:sampling_ver2:accuracy}), we can apply Hoeffding's inequality:
    \begin{align*}
    \Pr&\left[|\hat{M}[t]-M[t]|\geq M[t]\cdot \epsilon\right]
    \\&
    \leq 2\exp\left(-\frac{2\epsilon^2M[t]^2}{r\cdot (d_{\max}[t]|\Lambda|/w[t]r)^2}\right)
    \\&\leq 2\exp(-\frac{2\cdot w[t]^2r\epsilon^2M[t]^2}{|\Lambda|^2d_{\max}[t]^2})\leq \delta, \qedhere
    \end{align*}
     which implies the condition for the sample size.\hfill $\qed$.
    }

\begin{table*}[t]
	\begin{center}
		\scriptsize
		\caption{\label{global_table} \black{Global structural properties of real-world hypergraphs.}}
		\scalebox{1.1}{
			\begin{tabular}{l|cc|cc|cc|c|c}
				\toprule
				& \multicolumn{2}{c|}{Size} & \multicolumn{2}{c|}{Average Degree} & \multicolumn{2}{c|}{Clustering Coeff.} & \multirow{2}{*}{\shortstack{Effective\\ Diamter}} & \multirow{2}{*}{\shortstack{\# of\\ \Motifs}}\\
				\textbf{Dataset} & Node & Hyperedge & Node & Hyperedge & Node & Hyperedge & &\\
				\midrule
				coauth-DBLP & 1,924,991 & 2,466,792 & 4.013 & 3.132 & 0.296 & 0.225 & 6.590 & 26.3B\\
				coauth-geology& 1,256,385 & 1,203,895 & 3.015 & 3.146 & 0.382 & 0.208 & 6.809 & 6B\\
				coauth-history& 1,014,734 & 895,439 & 1.354 & 1.535 & 0.575 & 0.331 & 12.946 & 83.2M\\
				\midrule
				contact-primary  & 242 & 12,704 & 126.9 & 2.418 & 0.011 & 0.270 & 1.888 & 617M\\
				contact-high & 327 & 7,818 & 55.6 & 2.326 & 0.018 & 0.286 & 2.564 & 69.7M\\
				\midrule
				email-Enron & 143 & 1,512 & 31.8 & 3.009 & 0.064 & 0.249 & 2.583 & 9.6M\\
				email-EU & 998 & 25,027 & 85.9 & 3.425 & 0.030 & 0.207 & 2.836 & 7B\\
				\midrule
				tags-ubuntu & 3,029 & 147,222 & 164.8 & 3.391 & 0.007 & 0.182 & 2.262 & 4.3T\\
				tags-math & 1,629 & 170,476 & 364.1 & 3.479 & 0.005 & 0.180 & 2.189 & 9.2T\\
				\midrule
				threads-ubuntu & 125,602 & 166,999 & 2.538 & 1.908 & 0.160 & 0.301 & 4.657 & 11.4B\\
				threads-math & 176,445 & 595,749 & 8.261 & 2.446 & 0.033 & 0.250 & 3.662 & 2.2T\\
				\bottomrule
			\end{tabular}
		}
	\end{center}
\end{table*}

\section{Correlation between Global Structural Properties and \Motifs} \label{sec:addExp-networkProperties}
\label{appendix:addExp-networkProperties}

\black{
We investigate the correlation between global structural properties and the hypergraph motifs (\motifs). We consider eight global properties, categorized into five as follows.
\begin{itemize}
	\item \textbf{Size}: We consider (1) the number of nodes and (2) the number of hyperedges.
	\item \textbf{Average degree}: The degree of a node is defined as the number of hyperedges that contain the node. The degree of a hyperedge is defined as the number of nodes that the hyperedge contains. 
	We consider (3) the average degree of all nodes and (4) the average degree of all hyperedges.
	\item \textbf{Clustering coefficients}: 
	In \cite{gallagher2013clustering,latapy2008basic}, the clustering coefficient of two nodes $u\neq v$ is defined as:
	\begin{equation*}
	CC(u,v)=\frac{|E_u \cap E_v|}{|E_u \cup E_v|},
	\end{equation*}
	where $E_u$ is the set of hyperedges that contain the node $u$. Then, the average clustering coefficient of each node is defined as the mean of the clustering coefficients of it and each of its neighbors, i.e.,
	\begin{equation*}
	CC(u)=\frac{\sum_{v \in N_u}CC(u,v)}{|N_u|},
	\end{equation*}
	where $N_u$ is the set of neighbors of node $u$ (i.e., $\{v\in V: E_u \cap E_v \neq \emptyset \}$).
	Similarly, we can define the clustering coefficients of hyperedge pairs and those of hyperedges.
	We consider (5) the mean of the clustering coefficients of all nodes and (6) the mean of the clustering coefficients of all hyperedges.
	\item \textbf{Effective diameter}: The diameter of a graph is defined as the maximum length of the shortest paths in the graph. 
	Similarly, we define (7) the diameter of a hypergraph as the maximum length of the shortest hyperpaths in the hypergraph. 
	In hypergraphs, a hyperpath between two nodes $v_1$ and $v_k$ is a sequence of distinct nodes and hyperedges $v_1$, $e_1$, $v_2$, $e_2$, ...,$e_{k-1}$, $v_k$ where there $v_i \in e_i$ and $v_{i+1} \in e_i$ for $1 \leq i < k$~\cite{zhou2007learning}. The length of a hyperpath is the number of intermediate hyperedges contained in the hyperpath.
	In hypergraphs that have disconnected components, the diameter is not computable. Thus, we compute effective diameter~\cite{palmer2002anf}, which is the 90th percentile of the distribution of shortest-path lengths. We use \texttt{GetAnfEffDiam} function provided by~\cite{leskovec2016snap}.
	\item \textbf{Total number of \motifs}: We consider (8) the total number of hypergraph motifs computed using \method.
\end{itemize}
The statistics of the global properties of hypergraphs are provided in Table~\ref{global_table}. 
We compute the Pearson correlation coefficients between each \motif's CP value and eight different global properties. 
As seen in Table~\ref{cor_cp_table}, different \motifs show different correlations with each global property. 
For example, \motif 13 shows the lowest correlation with the node size ($-0.558$), the node-based clustering coefficient ($-0.650$), the hyperedge-based clustering coefficient ($-0.499$), and the effective diameter ($-0.620$), but the highest correlation with the average node degree ($0.523$) and the average hyperedge degree ($0.522$). 
In addition, some global properties are strongly correlated with \motifs, while some are weakly correlated. 
For example, there are many \motifs highly correlated with the node size, the node-based clustering coefficient, and the effective diameter, while most of the \motifs have near-zero correlations with the average hyperedge degree and the hyperedge-based clustering coefficient.}


\begin{table*}[t]
	\begin{center}
		\scriptsize
		\caption{\label{cor_cp_table} \black{Correlation between global structural properties and characteristic profiles (CPs). Strong positive or negative correlations ($\geq 0.5$) are colored as \textcolor{red}{red} (positive) or \textcolor{blue}{blue} (negative).}}
		\scalebox{1.1}{
			\begin{tabular}{c|cc|cc|cc|c|c}
				\toprule
				& \multicolumn{2}{c|}{Size} & \multicolumn{2}{c|}{Average Degree} & \multicolumn{2}{c|}{Clustering Coeff.} & \multirow{2}{*}{\shortstack{Effective\\ Diamter}} & \multirow{2}{*}{\shortstack{\# of\\ \Motifs}}\\
				\textbf{\motif} & Node & Hyperedge & Node & Hyperedge & Node & Hyperedge & &\\
				\midrule
				1 & \textcolor{red}{$+0.787$} & \textcolor{red}{$+0.648$} & $-0.497$ & $+0.003$ & \textcolor{red}{$+0.751$} & $+0.115$ & \textcolor{red}{$+0.626$} & \textcolor{blue}{$-0.530$} \\
				2 & $-0.169$ & $-0.110$ & $-0.023$ & \textcolor{blue}{$-0.645$} & $-0.027$ & $+0.479$ & $+0.099$ & $+0.101$ \\
				3 & \textcolor{red}{$+0.710$} & \textcolor{red}{$+0.587$} & $-0.492$ & $+0.065$ & \textcolor{red}{$+0.718$} & $+0.067$ & \textcolor{red}{$+0.633$} & \textcolor{blue}{$-0.501$} \\
				4 & \textcolor{red}{$+0.901$} & \textcolor{red}{$+0.783$} & $-0.480$ & $-0.071$ & \textcolor{red}{$+0.891$} & $+0.122$ & \textcolor{red}{$+0.806$} & $-0.410$ \\
				5 & $+0.469$ & $+0.315$ & \textcolor{blue}{$-0.506$} & $-0.126$ & $+0.498$ & $+0.316$ & $+0.406$ & \textcolor{blue}{$-0.730$} \\
				6 & \textcolor{red}{$+0.877$} & \textcolor{red}{$+0.762$} & $-0.311$ & $-0.161$ & \textcolor{red}{$+0.872$} & $+0.151$ & \textcolor{red}{$+0.801$} & $-0.249$ \\
				7 & $+0.229$ & $+0.071$ & $-0.436$ & $-0.188$ & $+0.313$ & $+0.394$ & $+0.215$ & \textcolor{blue}{$-0.744$} \\
				8 & $+0.388$ & $+0.237$ & \textcolor{blue}{$-0.561$} & $-0.154$ & $+0.447$ & $+0.359$ & $+0.361$ & \textcolor{blue}{$-0.772$} \\
				9 & $+0.138$ & $-0.023$ & $-0.266$ & $-0.081$ & $+0.229$ & $+0.293$ & $+0.124$ & \textcolor{blue}{$-0.606$} \\
				10 & $+0.444$ & $+0.288$ & \textcolor{blue}{$-0.584$} & $-0.178$ & \textcolor{red}{$+0.537$} & $+0.363$ & $+0.447$ & \textcolor{blue}{$-0.733$} \\
				11 & $+0.123$ & $-0.041$ & $+0.023$ & $+0.098$ & $+0.200$ & $+0.077$ & $+0.087$ & $-0.329$ \\
				12 & \textcolor{red}{$+0.611$} & \textcolor{red}{$+0.540$} & $-0.246$ & $-0.096$ & \textcolor{red}{$+0.676$} & $+0.050$ & \textcolor{red}{$+0.643$} & $-0.061$ \\
				13 & \textcolor{blue}{$-0.558$} & $-0.485$ & \textcolor{red}{$+0.523$} & \textcolor{red}{$+0.552$} & \textcolor{blue}{$-0.650$} & $-0.499$ & \textcolor{blue}{$-0.620$} & $+0.413$ \\
				14 & $-0.490$ & \textcolor{blue}{$-0.567$} & $+0.232$ & $+0.156$ & $-0.351$ & $-0.010$ & $-0.417$ & $-0.064$ \\
				15 & $+0.166$ & $+0.042$ & $-0.148$ & $+0.080$ & $+0.299$ & $+0.045$ & $+0.224$ & $-0.224$ \\
				16 & $+0.481$ & $+0.464$ & $-0.184$ & $+0.028$ & \textcolor{red}{$+0.531$} & $-0.110$ & \textcolor{red}{$+0.532$} & $+0.106$ \\
				17 & \textcolor{red}{$+0.754$} & \textcolor{red}{$+0.608$} & \textcolor{blue}{$-0.501$} & $-0.127$ & \textcolor{red}{$+0.828$} & $+0.231$ & \textcolor{red}{$+0.741$} & $-0.481$ \\
				18 & $+0.442$ & $+0.298$ & $-0.444$ & $+0.031$ & \textcolor{red}{$+0.551$} & $+0.149$ & $+0.473$ & \textcolor{blue}{$-0.532$} \\
				19 & \textcolor{red}{$+0.623$} & $+0.471$ & \textcolor{blue}{$-0.507$} & $-0.089$ & \textcolor{red}{$+0.677$} & $+0.251$ & \textcolor{red}{$+0.584$} & \textcolor{blue}{$-0.596$} \\
				20 & \textcolor{red}{$+0.628$} & $+0.473$ & $-0.483$ & $-0.001$ & \textcolor{red}{$+0.724$} & $+0.136$ & \textcolor{red}{$+0.633$} & \textcolor{blue}{$-0.514$} \\
				21 & $+0.089$ & $-0.047$ & $-0.315$ & $-0.136$ & $+0.141$ & $+0.346$ & $+0.074$ & \textcolor{blue}{$-0.676$} \\
				22 & \textcolor{red}{$+0.554$} & \textcolor{red}{$+0.522$} & $-0.157$ & $-0.299$ & \textcolor{red}{$+0.724$} & $+0.127$ & \textcolor{red}{$+0.763$} & $+0.154$ \\
				23 & $+0.332$ & $+0.181$ & $-0.434$ & $-0.130$ & $+0.365$ & $+0.334$ & $+0.273$ & \textcolor{blue}{$-0.722$} \\
				24 & $+0.428$ & $+0.275$ & $-0.492$ & $-0.147$ & $+0.459$ & $+0.341$ & $+0.368$ & \textcolor{blue}{$-0.737$} \\
				25 & \textcolor{red}{$+0.747$} & \textcolor{red}{$+0.593$} & \textcolor{blue}{$-0.508$} & $-0.151$ & \textcolor{red}{$+0.758$} & $+0.269$ & \textcolor{red}{$+0.666$} & \textcolor{blue}{$-0.610$} \\
				26 & \textcolor{red}{$+0.883$} & \textcolor{red}{$+0.812$} & $-0.330$ & $-0.203$ & \textcolor{red}{$+0.877$} & $+0.118$ & \textcolor{red}{$+0.830$} & $-0.129$ \\
				\bottomrule
			\end{tabular}
		}
	\end{center}
\end{table*}


\begin{algorithm}[t]
	\setstretch{1.25}
	\small
	\caption{\small Chung-Lu Model~\cite{aksoy2017measuring}}
	\label{randomized_hypergraph_alg}
	\SetAlgoLined
	\SetKwInOut{Input}{Input}
	\SetKwInOut{Output}{Output}
	\nonl \hspace{-4mm} \Input{bipartite graph $G'$: $G'=(V\cup E, E')$}
	\nonl \hspace{-4mm} \Output{randomized bipartite graph $\tilde{G}=(V\cup E,\tilde{E})$}
	${\{d_{i}^V\}}_{i=1}^{|V|} \leftarrow$ degrees of nodes in $V$\\
	${\{d_{j}^E\}}_{j=1}^{|E|} \leftarrow$ degrees of nodes in $E$\\
    $\tilde{E} \leftarrow$ an empty multiset \hspace{5mm} \\
	\For{$k=1,...,|E'|$ \label{randomized_hypergraph:loop1}}{  
		$v_i \leftarrow$ sample from $V$ with probability $\propto$ degree \label{random_hyper_alg:select_node}\\
		$e_j \leftarrow$ sample from $E$ with probability $\propto$ degree \label{random_hyper_alg:select_hyperedge}\\
        add $(v_i,e_j)$ to $\tilde{E}$\\
	}
	\textbf{return} $\tilde{G}=(\tilde{V},\tilde{E})$
\end{algorithm}
\begin{algorithm}[t]
	\setstretch{1.25}
	\small
	\caption{\small Transformation from Incidence Graphs to Hypergraphs}
	\label{alg:transform}
	\SetAlgoLined
	\SetKwInOut{Input}{Input}
	\SetKwInOut{Output}{Output}
	\nonl \hspace{-4mm} \Input{incidence graph $\tilde{G}=(V\cup E, \tilde{E})$}
	\nonl \hspace{-4mm} \Output{hypergraph $\hat{G}=(V,\hat{E})$}

        \For{\textbf{each} $\tilde{e}_j\in \tilde{E}$}{  
		$\hat{e}_j\leftarrow  \varnothing$
	}
        \For{\textbf{each} $(v_i,e_j) \in \tilde{E}$}{  
		$\hat{e}_j \leftarrow \hat{e}_j \cup \{v_i\}$\\ \label{alg:transform:set}
	}
        $\hat{E}\leftarrow\varnothing$\\
        \For{\textbf{each} $e_j\in E$}{  
		\If{$\hat{e}_j\neq  \varnothing$}{
              $\hat{E}\leftarrow \hat{E}\cup\{\hat{e}_j\}$
        }
	}
	\textbf{return} $\hat{G}=(V,\hat{E})$
\end{algorithm}

\section{Randomized Hypergraphs} \label{appendix:random}
\blue{To create randomized hypergraphs, we extend the Chung-Lu model~\cite{aksoy2017measuring}, which is a configuration model designed to generate random bipartite graphs, to hypergraphs.
Algorithm~\ref{randomized_hypergraph_alg} provides the pseudocode for the Chung-Lu model, where $V$ and $E$ represent two sets of nodes in both the input and output bipartite graphs.
Note that the output bipartite graph $\tilde{G}$ allows parallel edges (i.e. $\tilde{E}$ is a multiset).
The Chung-Lu model preserves the degrees of nodes in expectation, i.e., the following equalities hold:
\begin{equation}
    \mathbb{E}[\tilde{d}_i^V] = d_i^V, \ \mathbb{E}[\tilde{d}_j^E] = d_j^E, \label{eq:degree_preserve}
\end{equation}
where $d_i^V$ and $d_j^E$  denote the degrees of nodes $v_i\in V$ and $e_j\in E$ in the input graph; and
$\tilde{d}_i^V$ and $\tilde{d}_j^E$ denote the degrees (weighted by edge multiplicity) of nodes $v_i\in V$ and $e_j\in E$ in the output graph $\tilde{G}$, respectively.
Moreover, it can be shown that the obtained random graph is a uniform sample.
The Chung-Lu model is widely used due to efficiency because configuration models that preserve node degrees exactly exist are slower in generating a uniform sample.}

\blue{In order to create a hypergraph randomized from the input hypergraph $G=(V,E)$, we first apply this model to its incidence graph $G'$ to obtain a random bipartite graph $\tilde{G}$, and then we transform $\tilde{G}$ into a random hypergraph $\hat{G}$, as described in Algorithm~\ref{alg:transform}. Note that the multiplicity of edges in $\tilde{E}$ does not affect the output hypergraph since each hyperedge $\hat{e}_j$ in line~\ref{alg:transform:set} is a set, instead of a multiset.
If $G'$ does not have parallel edges (although they are allowed), Eq.~\eqref{eq:degree_preserve} implies the following equalities:
\begin{equation}
    \mathbb{E}[\hat{d}_i] = d_i, \ 
    \mathbb{E}[|\hat{e}_j|] = |e_j|,
\end{equation}
where $d_i$ and $|e_j|$ denote the degree of $v_i\in V$ and the size of $e_j\in E$ in the input hypergraph $G$; and
$\hat{d}_i$ and $|\hat{e}_j|$ denote those in the random hypergraph $\hat{G}$.
That is, both the node degree distribution and the hyperedge size distribution of $G$ are preserved in expectation in $\hat{G}$.
While these properties do not hold theoretically when $G'$ has parallel edges, empirically both distributions are  preserved accurately for real-world hypergraphs, as shown in Figure~\ref{deg_fig}.
We present
additional basic statistics for the generated hypergraphs in Table~\ref{dataset_random_table}.}

\begin{table}[t]
	\begin{center}
		\scriptsize
		\caption{\label{dataset_random_table}Statistics of generated random hypergraphs. Isolated nodes are ignored. The average values of $5$ random hypergraphs are reported.}
		\scalebox{1.1}{
			\begin{tabular}{l|cc}
				\toprule
				\textbf{Dataset} & $|V|$ & $|E|$\\
				\midrule
				coauth-DBLP & 1,483,582 & 2,221,017\\
				coauth-geology& 927,693 & 1,042,749\\
				coauth-history& 685,637 & 634,385\\
				\midrule
				contact-primary  & 242 & 11,416\\
				contact-high & 327 & 6,954\\
				\midrule
				email-Enron & 142 & 1,365\\
				email-EU & 964 & 22,764\\
				\midrule
				tags-ubuntu & 2,974 & 139,170\\
				tags-math & 1,606 & 162,260\\
				\midrule
				threads-ubuntu & 88,225 & 135,076\\
				threads-math & 134,722 & 522,149\\
				\bottomrule
			\end{tabular}
		}
	\end{center}
\end{table}

\begin{table*}[t]
	\begin{center}
		\scriptsize
		\caption{\label{hmotif_cnt_table}The number of \motifs with $k=2$ to $6$ hyperedges. For $k=2$ to $5$-hyperedge \motifs, the number of cases are counted via enumeration. For $k=6$, the number of cases is inferred.}
		\scalebox{1.1}{
			\begin{tabular}{l|ccccc}
				\toprule
				\textbf{Conditions} & $2$ Hyperedges & $3$ Hyperedges & $4$ Hyperedges & $5$ Hyperedges & $6$ Hyperedges\\
				\midrule
				$2^{2^k-1}$ & 8 & 128 & 32,768 & 2,147,483,648 & 9,223,372,036,854,775,808\\
				\midrule
				$2^{2^k-1} - |P_1^{(k)}|$ & 6 & 40 & 1,992 & 18,666,624 & 12,813,206,169,137,152\\
				$2^{2^k-1} - |P_1^{(k)} \cup P_2^{(k)}|$ & 3 & 30 & 1,912 & 18,662,590 & 12,813,206,131,799,685\\
				$2^{2^k-1} - |P_1^{(k)} \cup P_2^{(k)} \cup P_3^{(k)} |$ & 2 & 26 & 1,853 & 18,656,322 & ?\\
				\midrule
				Actual Count & 2 & 26 & 1,853 & 18,656,322 & ?\\
				\bottomrule
			\end{tabular}
		}
	\end{center}
\end{table*}

\section{\black{Extensions beyond Three Hyperedges}} \label{appendix:hmotif:generalized}
\black{The concept of \motifs can be generalized to $k$ hyperedges.
Defining \motifs describing the connectivity pattern of $k$ connected hyperedges requires the following conditions:
\begin{itemize}
	\item \textbf{No symmetric patterns}: Each connectivity pattern should be described by exactly one \motif.
	\item \textbf{No disconnected hyperedges}: All $k$ hyperedges should be connected.
	\item \textbf{No duplicated hyperedges}: \blue{All $k$ hyperedges should be distinct.}
\end{itemize}}

\black{Given a set $\egeneral$ of $k$ connected hyperedges, a node can be included in $1$ (e.g., $e_{s_1} \setminus e_{s_2} \setminus ... \setminus e_{s_k}$) to $k$ ($e_{s_1} \cap e_{s_2} \cap ... \cap e_{s_k}$) hyperedges at the same time. 
Thus, its connectivity pattern can be described by the emptiness of $2^k-1$ sets, which can be expressed as a binary vector of size $2^k-1$. 
Therefore, there can be $2^{2^k-1}$ \motifs, while only a subset of them remains once we exclude those violating any of the three conditions above. 
As a result, the number of remaining \motifs is 
\begin{equation}
2^{2^k-1} - |P_1^{(k)} \cup P_2^{(k)} \cup P_3^{(k)}|, \label{eq:motif_cnt}
\end{equation}
where $P_1^{(k)}$, $P_2^{(k)}$, and $P_3^{(k)}$ represent the sets of \blue{(1) the $k$-hyperedge patterns that are symmetric, (2) those that cannot be obtained from connected hyperedges, and (3) those that cannot be obtained from distinct hyperedges, respectively.}
For $P_1^{(k)}$, among the binary representations that imply the same pattern, all but the lexicographically smallest one are removed. 
That is, $P_1^{(k)}$ represents a set of patterns that are \textit{not} lexicographically smallest among those representing the same pattern.}

\smallsection{\black{\Motifs with Two Hyperedges:}}
\black{Given a pair of two connected hyperedges $\{e_i,e_j\}$, its connectivity can be described by the emptiness of three sets: $e_i - e_j$, $e_j - e_i$, and $e_i \cap e_j$. 
Once we remove the pattern that has duplicated hyperedges ($e_i - e_j = \varnothing$, $e_j - e_i = \varnothing$, and $e_i \cap e_j \neq \varnothing$), 3 patterns remain: (1) $e_i - e_j = \varnothing$, $e_j - e_i \neq \varnothing$, and $e_i \cap e_j \neq \varnothing$, (2) $e_i - e_j \neq \varnothing$, $e_j - e_i = \varnothing$, and $e_i \cap e_j \neq \varnothing$, and (3) $e_i - e_j \neq \varnothing$, $e_j - e_i \neq \varnothing$, and $e_i \cap e_j \neq \varnothing$. 
Since (1) and (2) are symmetric, we remove one of them, then 2 patterns remain.}

\smallsection{\black{\Motifs with Four Hyperedges:}}
\black{Given a set $\{e_i, e_j, e_k, e_l\}$ of four hyperedges, its connectivity pattern can be described by the emptiness of $15$ sets, which makes the total of $32,768$ patterns. 
Once we remove symmetric patterns and leave the unique ones, $1,992$ patterns remain. 
\blue{Among these, $80$ patterns cannot be obtained from connected hyperedges, $73$ patterns cannot be obtained from distinct hyperedges, and there are $14$ patterns that are common between the two sets.} 
Removing these leaves $1,853$ \motifs with four hyperedges.}

\smallsection{\black{\Motifs with $k$ Hyperedges:}}
\black{From~\eqref{eq:motif_cnt}, we can notice that the number of connectivity patterns of $k$-hyperedge \motifs rapidly increases with $k$. 
If $k=5$, there exists a total of $2,147,483,648$ ($=2^{31}$) patterns, while $18,656,322$ remain once we filter those that violate any one of the three conditions. 
The statistics of the number of \motifs for $k=2$ to $6$ are shown in Table~\ref{hmotif_cnt_table}. 
In case of $k=2$ to $5$-hyperedge \motifs, the number of connectivity patterns is counted via enumeration.
For a larger $k$, we use \textit{OEIS}~\cite{neil1996oeis} to obtain the numbers.
The sequence of the numbers of patterns after removing symmetric (redundant) ones ($2^{2^k-1}-|P_1^{(k)}|$) can be found at~\cite{seqA000612}, which is explained as the sequence of \textit{number of P-equivalence classes of switching functions of $k$ or fewer variables, divided by 2}~\footnote{It is also commented that the sequence is equivalent to the number of non-isomorphic fillings of a Venn diagram of $k$ sets.}. 
Since we can infer from Table~\ref{hmotif_cnt_table} that the number of \motif patterns is tightly upper-bounded by the remainder of the first filtering, this sequence gives us a good reference for the actual number of \motifs ($2^{2^k-1}-|P_1^{(k)} \cup P_2^{(k)} \cup P_3^{(k)}|$). 
The sequence increases rapidly, and if $k=10$, there exist $2.48^{301}$ patterns. 
We also note that~\cite{seqA323819} describes the sequence after the second filtering, $2^{2^k-1}-|P_1^{(k)} \cup P_2^{(k)}|$, which is \textit{the number of non-isomorphic connected set-systems covering $k$ vertices}.
Note that there are known formulas for both sequences \cite{seqA000612,seqA323819}.}

\begin{algorithm*}[t]
	\setstretch{1.25}
	\small
	\caption{\methodAWrandom: A Basic On-the-Fly Version of \methodAW without Line Graph Construction}
	\label{alg_memdeg}
	\SetAlgoLined
	\SetKwInOut{Input}{Input}
	\SetKwInOut{Output}{Output}
	\SetKwFunction{proc}{getNeighbors}
	\SetKwProg{myproc}{Subroutine}{}{}
	\nonl \hspace{-4mm} \Input{ \ \ (1) input hypergraph: $G=(V,E)$ \\  (2) number of samples: $r$ \\ {(3) budget size: $b$}}
	\nonl \hspace{-4mm} \Output{estimated count of each \motif $t$'s instances: $\MHT$}

        $\{d_1,\cdots,d_{|E|}\}\leftarrow$ hyperedge degrees in $\GT$ (\# of adjacent hyperedges)\\
        
	$\MH\leftarrow$ map whose default value is $0$ \\
		$Q\leftarrow$ a priority queue whose default value is $\varnothing$\\
		$A\leftarrow$ a map whose default value is $\varnothing$\\  
		$cap\leftarrow$ the remaining capacity of $A$. Its default value is $b$\\
	
	\For{$n\leftarrow 1...r$\label{alg_memdeg:loop1}}{
		$\wedge_{ij}\leftarrow$ a uniformly random \hwedge \label{alg_memdeg:sample} \\ 
			\For{\textbf{each} index $l\in \{i,j\}$}{
				\If{$A[e_l] = \varnothing$ \label{alg_memdeg:if_memoized}}{
					\While{$cap < d_l$ \label{alg_memdeg:while}}{
						$m \leftarrow$ $Q$.pop() \label{alg_memdeg:pop}\\
						\If{$m \notin \{i,j\}$ \label{alg_memdeg:if_ij}}{
							$cap \leftarrow cap + d_m$\label{alg_memdeg:capUp}\\
							$A[e_m] \leftarrow \varnothing$\label{alg_memdeg:remove}
						}
					}
					$A[e_l]\leftarrow $\proc{$e_l$, $G$}\label{alg_memdeg:insert}\\
					$cap \leftarrow cap - d_l$\label{alg_memdeg:capDown}\\
					$Q$.push$(l, d_l)$
				}
			}
			\For{\textbf{each} index $l\in \{i,j\}$ \label{alg_memdeg:loop}}{
				\If{$A[e_l] \neq \varnothing$ \textbf{and} $l \notin Q$}{
					$Q$.push$(l, d_l)$	\label{alg_memdeg:pushBack}
				}
			}
		\For{\textbf{each} hyperedge $e_k \in \left(A[e_i] \cup A[e_j] \setminus \{e_i,e_j\}\right)$\label{alg_memdeg:loop2}}{
			$\MH[\hijk] \mathrel{+}= 1$
			\label{alg_memdeg:count}
		}
		\label{alg_memdeg:loop2:end}
	}
	\For{\textbf{each} \motif $t$\label{alg_memdeg:rescale:start}}{
		\eIf(\hfill $\triangleright$ \texttt{\color{blue}open \motifs}){17 $\leq$ t $\leq$ 22}{ 
			$\MHT \leftarrow \MHT \cdot \tfrac{|\wedge|}{2r}$ 
		}
		(\hfill $\triangleright$ \texttt{\color{blue}closed \motifs}){
			$\MHT \leftarrow \MHT \cdot \tfrac{|\wedge|}{3r}$
		}
	}\label{alg_memdeg:rescale:end}
	\textbf{return} $\MH$ \\ 
	\vspace{2.5mm}
	
		\myproc{\proc{$e$, $G$} \label{alg_memdeg:getNeighbors}}{
			$\hat{N}_e\leftarrow$ a map whose default value is $0$\\
			\For{\textbf{each} node $v\in e$}{
				\For{\textbf{each} hyperedge $e' \in E_v \setminus \{e\}$}{
					$\hat{N}_e[e'] \leftarrow \hat{N}_e[e'] + 1$
				}
			}
			\textbf{return} $\hat{N}_e$ \label{alg_memdeg:getNeighborsFinish}
		}
\end{algorithm*}

\begin{algorithm*}[t]
	\setstretch{1.25}
	\small
	\caption{\blue{\methodAWgreedy: An Advanced On-the-Fly Version of \methodAW without Line Graph Construction}}
	\label{alg_memdeg_adv}
	\SetAlgoLined
	\SetKwInOut{Input}{Input}
	\SetKwInOut{Output}{Output}
	\SetKwFunction{proc}{getNeighbors}
	\SetKwProg{myproc}{Subroutine}{}{}
	\nonl \hspace{-4mm} \Input{ \ \ (1) input hypergraph: $G=(V,E)$ \\  (2) number of samples: $r$ \\ {(3) budget size: $b$}}
	\nonl \hspace{-4mm} \Output{estimated count of each \motif $t$'s instances: $\MHT$}

        $\{d_1,\cdots,d_{|E|}\}\leftarrow$ hyperedge degrees in $\GT$ (\# of adjacent hyperedges)\\
        
	$\MH\leftarrow$ map whose default value is $0$ \\
		$Q\leftarrow$ a priority queue whose default value is $\varnothing$\\
		$A\leftarrow$ a map whose default value is $\varnothing$\\  
		$cap\leftarrow$ the remaining capacity of $A$. Its default value is $b$\\

        $W \leftarrow$ a map whose default value is $\varnothing$ \\
	\For{$n\leftarrow 1...r$\label{alg_memdeg_adv:loop1}}{
		      $\wedge_{ij}\leftarrow$ a uniformly random \hwedge \label{alg_memdeg_adv:sample} \\ 
                \If{\red{$d_i > d_j$ or ($d_i=d_j$ and $i>j$)}\label{alg_memdeg_adv:group:begin}}{
                    $W[e_i] \leftarrow \red{W[e_i] \cup \{\wedge_{ij}\}}$
                }
                \Else{
                    $W[e_j] \leftarrow \red{W[e_j] \cup \{\wedge_{ij}\}}$
                }
	}\label{alg_memdeg_adv:group:end}
        $K \leftarrow$ sort the keys (i.e., hyperedges) of $W$ in descending order of \red{$(d_i,i)$ of each key $e_i$} \label{alg_memdeg_adv:group:sort}\\
        \For{\textbf{each} hyperedge $e_i \in K$\label{alg_memdeg_adv:group:loop}}{
            \For{\textbf{each} hyperwedge $\wedge_{ij}\in W[e_i]$}{
                \For{\textbf{each} index $l\in \{i,j\}$}{
    				\If{$A[e_l] = \varnothing$}{
    					\While{$cap < d_l$}{
    						$m \leftarrow$ $Q$.pop()\\
    						\If{$m \notin \{i,j\}$}{
    							$cap \leftarrow cap + d_m$\\
    							$A[e_m] \leftarrow \varnothing$
    						}
    					}
    					$A[e_l]\leftarrow $\proc{$e_l$, $G$} \hfill $\triangleright$ \texttt{\color{blue} Algorithm~\ref{alg_memdeg}} \label{alg_memdeg_adv:insert}\\
    					$cap \leftarrow cap - d_l$\label{alg_memdeg_adv:capDown}\\
    					$Q$.push$(l, d_l)$
    				}
    		  }
                \For{\textbf{each} hyperedge $e_k \in \left(A[e_i] \cup A[e_j] \setminus \{e_i,e_j\}\right)$\label{alg_memdeg_adv:enum_W:loop}}{
    			$\MH[\hijk] \mathrel{+}= 1$
    			\label{alg_memdeg_adv:enum_W:count}
    		  }
            }
                $Q$.remove$(i)$\label{alg_memdeg_adv:remove:start}\\
                $cap \leftarrow cap + d_{i}$\\
                $A[e_i]\leftarrow \varnothing$ \label{alg_memdeg_adv:remove:end}\\
        }
	\For{\textbf{each} \motif $t$\label{alg_memdeg_adv:rescale:start}}{
		\eIf(\hfill $\triangleright$ \texttt{\color{blue}open \motifs}){17 $\leq$ t $\leq$ 22}{ 
			$\MHT \leftarrow \MHT \cdot \tfrac{|\wedge|}{2r}$ 
		}
		(\hfill $\triangleright$ \texttt{\color{blue}closed \motifs}){
			$\MHT \leftarrow \MHT \cdot \tfrac{|\wedge|}{3r}$
		}
	}\label{alg_memdeg_adv:rescale:end}
	\textbf{return} $\MH$ \\ 
\end{algorithm*}

\section{\black{Number of $k$H-motifs}} \label{appendix:kh-motif:generalized}
\black{
The number of $k$h-motifs for any $k>1$ can be calculated by subdividing each \motif into multiple $k$h-motifs. 
For example, in \motif 1 in Figure~\ref{motif_three_hyperedges}, depending on the cardinality of the the seven considered sets,
there are $(k-1)$ possibilities for the red region,\footnote{A colored region cannot be empty, and thus there are $k-1$ possible states.}
and ${k-1 \choose 2} + k-1={k \choose 2}$ possibilities for the two symmetric green regions.\footnote{There are ${k-1 \choose 2}$  possibilities when they have different states and $k-1$ (symmetric) possibilities when they have the same state.}
Therefore, $(k-1) \times$$k \choose 2$$=\frac{k(k-1)^2}{2}$ $k$h-motifs are subdivided from \motif 1. By applying the same process to each \motif and summing up the results, we obtain the following formula:
\begin{equation}
X(k)=\frac{k(k-1)(k^5+k^4+4k^3+k^2-4k+2)}{6} \label{eq:kh_motif_cnt}
\end{equation}
where $X(k)$ is the number of $k$h-motifs for each $k\geq 2$. Note that the number of \motifs is $X(2)=26$, and the number of \tmotifs is $X(3)=431$.}

\section{Details of On-the-Fly \method}
\label{appendix:on-the-fly}
In this section, we provide detailed explanations of the on-the-fly version of \method, which counts the instances of \motifs without \blue{line-graph construction} in the preprocessing step (see Section~\ref{method:on_the_fly}). 
\blue{We present two versions of on-the-fly algorithms for \methodAW: \methodAWrandom and \methodAWgreedy.
The pseudocodes for these algorithms are provided in Algorithms~\ref{alg_memdeg} and \ref{alg_memdeg_adv}, respectively.}
Unlike the basic version of \methodAW, the \blue{line} graph is not given as an input. 
Instead, the budget $b$ of memoization is given.
To this end, we use three additional data structures:
\begin{itemize}
	\item $Q$: a priority queue that prioritizes high-degree hyperedges. The index of an hyperedge is stored as a key, and its degree is used as the corresponding priority. It is initialized to $\varnothing$.
	\item $A$: a map that memoizes the neighbors of a subset of hyperedges. It is initialized to $\varnothing$. The maximum capacity of $A$ is given as the budget $b$.
	\item $cap$: a variable that records the remaining capacity of $A$. This variable is initialized to $b$.
\end{itemize}

\smallsection{\methodAWrandom (Algorithm~\ref{alg_memdeg}):}
For each hyperedge $e_i$ in each sampled \hwedge $\wedge_{ij} \in \wedge$, Algorithm~\ref{alg_memdeg} first checks whether its neighbors are memoized in the map $A$ (line~\ref{alg_memdeg:if_memoized}). 
If the neighbors are not memoized, it checks the remaining capacity $cap$ and memoizes the neighbors of $e_i$ in $A$ if the remaining capacity allows (i.e., if $d_i \leq cap$).
Otherwise, it removes the neighbors of a hyperedge with the lowest priority from $A$.
Specifically, since we prioritize hyperedges with high degree in $\GT$, a hyperedge with the lowest degree is dequeued from $Q$ (line~\ref{alg_memdeg:pop}), and its neighbors are evicted from $A$ (line~\ref{alg_memdeg:remove}).
Evicting these neighbors increases the remaining capacity $cap$ of $A$ (line~\ref{alg_memdeg:capUp}).
This is repeated until the remaining capacity $cap$ becomes large enough to memoize $e_i$'s neighbors (i.e., $cap \geq d_i$).
Then, the neighbors of $e_i$ are retrieved by calling the \texttt{getNeighbors} function and memoized (line~\ref{alg_memdeg:insert}), followed by the reduction of the remaining capacity $cap$.

\smallsection{\blue{\methodAWgreedy (Algorithm~\ref{alg_memdeg_adv}):}}
\blue{
We additionally develop \methodAWgreedy, an advanced on-the-fly version of \methodAW in Algorithm~\ref{alg_memdeg_adv}.
In \methodAWrandom, when the same hyperedge is processed again after its neighbors are evicted from the line graph, it requires the reconstruction of its neighbors, thereby increasing the computational overhead.
\methodAWgreedy is designed to mitigate this by taking the processing order of hyperwedges into account, aiming to reduce unnecessary reconstruction.
The high-level idea behind  \methodAWgreedy is to consecutively process hyperwedges consisting of the same hyperedges, thereby increasing the chance of utilizing memoized neighbors before they are evicted.
\methodAWgreedy first collects $r$ hyperwedges (lines~\ref{alg_memdeg_adv:loop1}-\ref{alg_memdeg_adv:sample}). 
For each hyperedge, we define its \textit{key hyperedge} as one of the two hyperedges forming the hyperwedge whose degree in $\GT$ \red{(or index in the case of a tie) is greater} than that of the other hyperedge. 
Then, it groups the collected hyperwedges based on their respective key hyperedge (line~\ref{alg_memdeg_adv:group:begin}-\ref{alg_memdeg_adv:group:end}).
After that, the hyperwedges are processed group by group 
(line~\ref{alg_memdeg_adv:group:loop})
in descending order of the degree \red{in $\GT$ (or index in the case of a tie) of key hyperedges} (line~\ref{alg_memdeg_adv:group:sort}), thereby increasing the chance of utilizing memoized neighbors before they are evicted.
The processing method for each hyperwedge remains the same as that of \methodAWrandom.
After processing each group, the neighbors of its key hyperedge are evicted from the line graph (lines~\ref{alg_memdeg_adv:remove:start}-\ref{alg_memdeg_adv:remove:end}) permanently, as all hyperwedges containing it have been processed.
}

 %

\begin{table*}[t]
	\begin{center}
		\caption{\label{classifier_prediction_table}
		\black{Accuracy (ACC) and area under the ROC curve (AUC) for hyperedge prediction using five different classifiers for each dataset.  D1 to D7 refer to the coauth-DBLP, coauth-MAG-Geology, coauth-MAG-History, contact-primary-school, contact-high-school, email-Enron, and email-Eu datasets in order.
        For each dataset,
        the best and second-best results are in \textbf{bold} and \underline{underlined}, respectively.
        }}
        \scalebox{0.62}{
            \subtable[HP26]{
        		\begin{tabular}{c|c|ccccccc|c} 
        				\toprule
        				& \textbf{Classifier} & \textbf{D1} & \textbf{D2} & \textbf{D3} & \textbf{D4} & \textbf{D5} & \textbf{D6} & \textbf{D7} & \textbf{AVG}\\
        				\midrule
        				\multirow{5}{*}{\textbf{ACC$*$}} & \textbf{Logistic Regression} & 0.764 $\pm$ 0.000 & 0.743 $\pm$ 0.000 & 0.654 $\pm$ 0.000 & 0.770 $\pm$ 0.000 & 0.895 $\pm$ 0.000 & \textbf{0.815 $\pm$ 0.000} & 0.904 $\pm$ 0.000 & 0.792 $\pm$ 0.000\\  
                        & \textbf{Decision Tree} & 0.729 $\pm$ 0.000 & 0.705 $\pm$ 0.001 & \underline{0.689 $\pm$ 0.005} & 0.749 $\pm$ 0. 003 & 0.886 $\pm$ 0.002 & 0.769 $\pm$ 0.013 & 0.881 $\pm$ 0.001 & 0.773 $\pm$ 0.004\\
                         & \textbf{Random Forest} & 0.774 $\pm$ 0.002 & 0.740 $\pm$ 0.004 & 0.633 $\pm$ 0.010 & 0.765 $\pm$ 0.002 & 0.879 $\pm$ 0.003 & 0.704 $\pm$ 0.005 & 0.876 $\pm$ 0.003 & 0.767 $\pm$ 0.004\\
                        & \textbf{MLP} & \textbf{0.804 $\pm$ 0.002} & \textbf{0.786 $\pm$ 0.003} & 0.678 $\pm$ 0.007 & \textbf{0.773 $\pm$ 0.009} & \underline{0.906 $\pm$ 0.012} & \underline{0.812 $\pm$ 0.011} & \textbf{0.915 $\pm$ 0.004} & \underline{0.811 $\pm$ 0.007}\\
                        & \textbf{XGBoost} & \underline{0.801 $\pm$ 0.000} & \underline{0.782 $\pm$ 0.000} & \textbf{0.696 $\pm$ 0.000} & \underline{0.772 $\pm$ 0.000} & \textbf{0.907 $\pm$ 0.000} & \textbf{0.815 $\pm$ 0.000} & \underline{0.911 $\pm$ 0.000} & \textbf{0.812 $\pm$ 0.000}\\
                        \midrule
                        \multirow{5}{*}{\textbf{AUC$\dagger$}} & \textbf{Logistic Regression} & 0.843 $\pm$ 0.000 & 0.827 $\pm$ 0.000 & 0.688 $\pm$ 0.000 & \underline{0.892 $\pm$ 0.000} & \underline{0.971 $\pm$ 0.000} & \textbf{0.928 $\pm$ 0.000} & 0.969 $\pm$ 0.000 & 0.874 $\pm$ 0.000\\  
                        & \textbf{Decision Tree} & 0.730 $\pm$ 0.000 & 0.707 $\pm$ 0.001 & 0.698 $\pm$ 0.007 & 0.749 $\pm$ 0.003 & 0.886 $\pm$ 0.002 & 0.769 $\pm$ 0.013 & 0.881 $\pm$ 0.001 & 0.774 $\pm$ 0.003\\
                        & \textbf{Random Forest} & \underline{0.862 $\pm$ 0.002} & 0.826 $\pm$ 0.003 & 0.679 $\pm$ 0.006 & 0.885 $\pm$ 0.001 & 0.966 $\pm$ 0.001 & 0.897 $\pm$ 0.004 & 0.946 $\pm$ 0.001 & 0.866 $\pm$ 0.004\\
                        & \textbf{MLP} & \textbf{0.886 $\pm$ 0.002} & \textbf{0.867 $\pm$ 0.002} & \underline{0.781 $\pm$ 0.008} & \textbf{0.893 $\pm$ 0.000} & \textbf{0.973 $\pm$ 0.001} & 0.917 $\pm$ 0.006 & \textbf{0.973 $\pm$ 0.001} & \underline{0.899 $\pm$ 0.007}\\
                        & \textbf{XGBoost} & \textbf{0.886 $\pm$ 0.000} & \underline{0.865 $\pm$ 0.000} & \textbf{0.811 $\pm$ 0.000} & 0.879 $\pm$ 0.000 & 0.968 $\pm$ 0.000 & \underline{0.922 $\pm$ 0.000} & \underline{0.972 $\pm$ 0.000} & \textbf{0.900 $\pm$ 0.000}\\
                        \bottomrule
        				\multicolumn{5}{l}{$*$ acuracy, $\dagger$ area under the ROC curve.}
        		\end{tabular}
          }
        }
        \scalebox{0.62}{
            \subtable[HP7]{
        		\begin{tabular}{c|c|ccccccc|c} 
        				\toprule
        				& \textbf{Classifier} & \textbf{D1} & \textbf{D2} & \textbf{D3} & \textbf{D4} & \textbf{D5} & \textbf{D6} & \textbf{D7} & \textbf{AVG}\\
        				\midrule
        				\multirow{5}{*}{\textbf{ACC}} & \textbf{Logistic Regression} & 0.713 $\pm$ 0.000 & 0.698 $\pm$ 0.000 & 0.605 $\pm$ 0.000 & 0.764 $\pm$ 0.000 & 0.859 $\pm$ 0.000 & 0.676 $\pm$ 0.000 & 0.858 $\pm$ 0.000 & 0.739 $\pm$ 0.000\\  
                        & \textbf{Decision Tree} & 0.657 $\pm$ 0.000 & 0.629 $\pm$ 0.001 & 0.645 $\pm$ 0.003 & 0.751 $\pm$ 0.002 & 0.847 $\pm$ 0.008 & 0.699 $\pm$ 0.005 & 0.845 $\pm$ 0.002 & 0.725 $\pm$ 0.003\\
                         & \textbf{Random Forest} & 0.736 $\pm$ 0.000 & 0.704 $\pm$ 0.003 & 0.592 $\pm$ 0.006 & \textbf{0.777 $\pm$ 0.002} & 0.826 $\pm$ 0.008 & 0.639 $\pm$ 0.005 & 0.785 $\pm$ 0.004 & 0.723 $\pm$ 0.004\\
                        & \textbf{MLP} & \underline{0.742 $\pm$ 0.005} & \textbf{0.723 $\pm$ 0.005} & \underline{0.654 $\pm$ 0.006} & 0.767 $\pm$ 0.006 & \textbf{0.873 $\pm$ 0.011} & \underline{0.722 $\pm$ 0.008} & \underline{0.876 $\pm$ 0.007} & \underline{0.765 $\pm$ 0.007}\\
                        & \textbf{XGBoost} & \textbf{0.744 $\pm$ 0.000} & \underline{0.722 $\pm$ 0.000} & \textbf{0.683 $\pm$ 0.000} & \underline{0.769 $\pm$ 0.000} & \underline{0.860 $\pm$ 0.000} & \textbf{0.725 $\pm$ 0.000} & \textbf{0.878 $\pm$ 0.000} & \textbf{0.769 $\pm$ 0.000}\\
                        \midrule
                        \multirow{5}{*}{\textbf{AUC}} & \textbf{Logistic Regression} & 0.775 $\pm$ 0.000 & 0.770 $\pm$ 0.000 & 0.612 $\pm$ 0.000 & \underline{0.879 $\pm$ 0.000} & \underline{0.955 $\pm$ 0.000} & 0.804 $\pm$ 0.000 & 0.939 $\pm$ 0.000 & 0.819 $\pm$ 0.000\\  
                        & \textbf{Decision Tree} & 0.657 $\pm$ 0.000 & 0.629 $\pm$ 0.001 & 0.631 $\pm$ 0.006 & 0.743 $\pm$ 0.002 & 0.847 $\pm$ 0.008 & 0.711 $\pm$ 0.006 & 0.830 $\pm$ 0.003 & 0.721 $\pm$ 0.004\\
                        & \textbf{Random Forest} & \underline{0.776 $\pm$ 0.001} & 0.755 $\pm$ 0.003 & 0.630 $\pm$ 0.003 & \textbf{0.880 $\pm$ 0.000} & 0.938 $\pm$ 0.001 & 0.737 $\pm$ 0.003 & 0.882 $\pm$ 0.002 & 0.800 $\pm$ 0.002\\
                        & \textbf{MLP} & \textbf{0.820 $\pm$ 0.002} & \textbf{0.799 $\pm$ 0.002} & \underline{0.731 $\pm$ 0.006} & \textbf{0.880 $\pm$ 0.002} & \textbf{0.960 $\pm$ 0.001} & \textbf{0.831 $\pm$ 0.004} & \textbf{0.955 $\pm$ 0.000} & \textbf{0.854 $\pm$ 0.002}\\
                        & \textbf{XGBoost} & \textbf{0.820 $\pm$ 0.000} & \underline{0.798 $\pm$ 0.000} & \textbf{0.761 $\pm$ 0.000} & 0.868 $\pm$ 0.000 & 0.949 $\pm$ 0.000 & \underline{0.816 $\pm$ 0.000} & \underline{0.954 $\pm$ 0.000} & \underline{0.852 $\pm$ 0.000}\\
                        \bottomrule
        		\end{tabular}
          }
        }
        \scalebox{0.62}{
            \subtable[THP]{
        		\begin{tabular}{c|c|ccccccc|c} 
        				\toprule
        				& \textbf{Classifier} & \textbf{D1} & \textbf{D2} & \textbf{D3} & \textbf{D4} & \textbf{D5} & \textbf{D6} & \textbf{D7} & \textbf{AVG}\\
        				\midrule
        				\multirow{5}{*}{\textbf{ACC}} & \textbf{Logistic Regression} & 0.803 $\pm$ 0.000 & 0.786 $\pm$ 0.000 & 0.693 $\pm$ 0.000 & \underline{0.770 $\pm$ 0.000} & 0.895 $\pm$ 0.000 & \underline{0.815 $\pm$ 0.000} & \underline{0.915 $\pm$ 0.000} & 0.811 $\pm$ 0.000\\  
                        & \textbf{Decision Tree} & 0.772 $\pm$ 0.000 & 0.738 $\pm$ 0.001 & \underline{0.695 $\pm$ 0.005} & 0.755 $\pm$ 0.003 & 0.872 $\pm$ 0.005 & 0.806 $\pm$ 0.011 & 0.887 $\pm$ 0.002 & 0.789 $\pm$ 0.004\\
                         & \textbf{Random Forest} & 0.776 $\pm$ 0.003 & 0.740 $\pm$ 0.003 & 0.597 $\pm$ 0.011 & 0.768 $\pm$ 0.003 & 0.860 $\pm$ 0.007 & 0.712 $\pm$ 0.016 & 0.870 $\pm$ 0.002 & 0.760 $\pm$ 0.006\\
                        & \textbf{MLP} & \underline{0.824 $\pm$ 0.004} & \underline{0.805 $\pm$ 0.003} & \underline{0.695 $\pm$ 0.010} & 0.767 $\pm$ 0.009 & \textbf{0.906 $\pm$ 0.016} & 0.807 $\pm$ 0.006 & 0.911 $\pm$ 0.006 & \underline{0.816 $\pm$ 0.008}\\
                        & \textbf{XGBoost} & \textbf{0.836 $\pm$ 0.000} & \textbf{0.819 $\pm$ 0.000} & \textbf{0.716 $\pm$ 0.000} & \textbf{0.779 $\pm$ 0.000} & \underline{0.904 $\pm$ 0.000} & \textbf{0.827 $\pm$ 0.000} & \textbf{0.920 $\pm$ 0.000} & \textbf{0.829 $\pm$ 0.000}\\
                        \midrule
                        \multirow{5}{*}{\textbf{AUC}} & \textbf{Logistic Regression} & 0.870 $\pm$ 0.000 & 0.853 $\pm$ 0.000 & 0.705 $\pm$ 0.000 & \textbf{0.891 $\pm$ 0.000} & \underline{0.969 $\pm$ 0.000} & \underline{0.907 $\pm$ 0.000} & \underline{0.975 $\pm$ 0.000} & 0.881 $\pm$ 0.000\\  
                        & \textbf{Decision Tree} & 0.772 $\pm$ 0.000 & 0.738 $\pm$ 0.001 & 0.717 $\pm$ 0.006 & 0.755 $\pm$ 0.003 & 0.872 $\pm$ 0.005 & 0.806 $\pm$ 0.011 & 0.887 $\pm$ 0.002 & 0.792 $\pm$ 0.004\\
                        & \textbf{Random Forest} & 0.858 $\pm$ 0.004 & 0.830 $\pm$ 0.003 & 0.632 $\pm$ 0.008 & 0.885 $\pm$ 0.001 & 0.962 $\pm$ 0.001 & 0.899 $\pm$ 0.005 & 0.956 $\pm$ 0.001 & 0.860 $\pm$ 0.003\\
                        & \textbf{MLP} & \underline{0.898 $\pm$ 0.003} & \underline{0.876 $\pm$ 0.006} & \underline{0.769 $\pm$ 0.011} & 0.882 $\pm$ 0.003 & \textbf{0.971 $\pm$ 0.001} & 0.895 $\pm$ 0.002 & 0.971 $\pm$ 0.003 & \underline{0.895 $\pm$ 0.004}\\
                        & \textbf{XGBoost} & \textbf{0.909 $\pm$ 0.000} & \textbf{0.892 $\pm$ 0.000} & \textbf{0.820 $\pm$ 0.000} & \underline{0.886 $\pm$ 0.000} & 0.967 $\pm$ 0.000 & \textbf{0.921 $\pm$ 0.000} & \textbf{0.977 $\pm$ 0.000} & \textbf{0.910 $\pm$ 0.000}\\
                        \bottomrule
        		\end{tabular}
          }
        }
        \scalebox{0.62}{
            \subtable[BASELINE]{
        		\begin{tabular}{c|c|ccccccc|c} 
        				\toprule
        				& \textbf{Classifier} & \textbf{D1} & \textbf{D2} & \textbf{D3} & \textbf{D4} & \textbf{D5} & \textbf{D6} & \textbf{D7} & \textbf{AVG}\\
        				\midrule
        				\multirow{5}{*}{\textbf{ACC}} & \textbf{Logistic Regression} & \underline{0.640 $\pm$ 0.000} & 0.659 $\pm$ 0.000 & \underline{0.670 $\pm$ 0.000} & 0.539 $\pm$ 0.000 & 0.532 $\pm$ 0.000 & 0.517 $\pm$ 0.000 & 0.607 $\pm$ 0.000 & 0.595 $\pm$ 0.000\\
                        & \textbf{Decision Tree} & 0.613 $\pm$ 0.000 & 0.627 $\pm$ 0.001 & 0.612 $\pm$ 0.002 & \underline{0.588 $\pm$ 0.003} & \underline{0.580 $\pm$ 0.005} & 0.566 $\pm$ 0.025 & \underline{0.655 $\pm$ 0.002} & 0.606 $\pm$ 0.005\\
                         & \textbf{Random Forest} & \textbf{0.646 $\pm$ 0.001} & \textbf{0.678 $\pm$ 0.002} & \textbf{0.699 $\pm$ 0.000} & 0.557 $\pm$ 0.002 & 0.527 $\pm$ 0.002 & \underline{0.614 $\pm$ 0.037} & 0.612 $\pm$ 0.003 & \underline{0.619 $\pm$ 0.007}\\
                        & \textbf{MLP} & \textbf{0.646 $\pm$ 0.001} & 0.653 $\pm$ 0.006 & 0.663 $\pm$ 0.005 & 0.537 $\pm$ 0.012 & 0.537 $\pm$ 0.008 & 0.562 $\pm$ 0.007 & 0.574 $\pm$ 0.032 & 0.596 $\pm$ 0.010\\
                        & \textbf{XGBoost} & \textbf{0.646 $\pm$ 0.000} & \underline{0.661 $\pm$ 0.000} & 0.608 $\pm$ 0.000 & \textbf{0.603 $\pm$ 0.000} & \textbf{0.585 $\pm$ 0.000} & \textbf{0.633 $\pm$ 0.000} & \textbf{0.702 $\pm$ 0.000} & \textbf{0.634 $\pm$ 0.000}\\
                        \midrule
                        \multirow{5}{*}{\textbf{AUC}} & \textbf{Logistic Regression} & 0.696 $\pm$ 0.000 & 0.712 $\pm$ 0.000 & \textbf{0.800 $\pm$ 0.000} & 0.555 $\pm$ 0.000 & 0.561 $\pm$ 0.000 & 0.574 $\pm$ 0.000 & 0.650 $\pm$ 0.000 & 0.650 $\pm$ 0.000\\  
                        & \textbf{Decision Tree} & 0.617 $\pm$ 0.000 & 0.628 $\pm$ 0.001 & 0.613 $\pm$ 0.002 & \underline{0.588 $\pm$ 0.003} & \underline{0.580 $\pm$ 0.005} & 0.566 $\pm$ 0.025 & 0.655 $\pm$ 0.002 & 0.607 $\pm$ 0.005\\
                        & \textbf{Random Forest} & 0.699 $\pm$ 0.001 & \textbf{0.742 $\pm$ 0.001} & 0.656 $\pm$ 0.009 & 0.580 $\pm$ 0.001 & 0.563 $\pm$ 0.002 & \textbf{0.706 $\pm$ 0.010} & \underline{0.667 $\pm$ 0.002} & 0.659 $\pm$ 0.004\\
                        & \textbf{MLP} & \underline{0.702 $\pm$ 0.002} & 0.727 $\pm$ 0.003 & \underline{0.770 $\pm$ 0.007} & 0.569 $\pm$ 0.008 & 0.569 $\pm$ 0.002 & 0.635 $\pm$ 0.015 & 0.663 $\pm$ 0.011 & \underline{0.662 $\pm$ 0.007}\\
                        & \textbf{XGBoost} & \textbf{0.707 $\pm$ 0.000} & \underline{0.741 $\pm$ 0.000} & 0.732 $\pm$ 0.000 & \textbf{0.647 $\pm$ 0.000} & \textbf{0.641 $\pm$ 0.000} & \underline{0.701 $\pm$ 0.000} & \textbf{0.781 $\pm$ 0.000} & \textbf{0.707 $\pm$ 0.000}\\
                        \bottomrule
        		\end{tabular}
          }
        }
	\end{center}
\end{table*}



\section{Hyperedge Prediction}
\label{appendix:prediction}
In this section, we provide detailed information about the hyperedge prediction task in Section~\ref{sec:exp:observations:prediction} and additional experimental results.

\smallsection{Settings:}
\label{appendix:prediction:settings}
To generate training and test sets, we first extract hyperedges whose sizes are at least $3$ from each domain. 
In the test set, we remove hyperedges that contain out-of-sample nodes (i.e., nodes that are not contained in the training set). 
In both training and test sets, a fake hyperedge is generated from each real hyperedge to balance the numbers of real and fake hyperedges. 

\smallsection{Negative Hyperedge Sampling:} 
\label{appendix:prediction:negative}
For each positive hyperedge $e$, we first generate a random probability $\alpha$ ($\frac{1}{3} \leq \alpha \leq 1$~\footnote{$\alpha$ is at least $\frac{1}{3}$ since the minimum size of sampled hyperedges is $3$.}). Then, we replace $\alpha \cdot |e|$ among the nodes with randomly selected nodes. Instead of sampling negative nodes from $V$ uniformly at random, we sample based on the noisy distribution, as in~\cite{mikolov2013distributed}, where the probability of each node $v_i$ being selected is:
\begin{equation*}
	P(v_i) = \frac{|E_{v_i}|^{\frac{3}{4}}}{\sum_{j=1}^{|V|}|E_{v_j}|^{\frac{3}{4}}}.
\end{equation*}
For each generated negative hyperedge $e'$, we ensure that the size is the same as the positive one (i.e., no duplicated nodes in the hyperedge satisfying $|e|=|e'|$) and it is negative (i.e., $e'$ does not exist in both training and test sets of positive hyperedges).

\smallsection{\black{Results with More Classifiers:}}
\label{appendix:prediction:more_classifier}
\black{
We evaluate the performance of the four considered feature sets in hyperedge prediction using four additional classifiers: Logistic Regression, Decision Tree, Random Forest, and MLP, in addition to XGBoost.
We use the implementation of all classifiers provided by scikit-learn~\cite{scikit-learn} with default hyperparameters.
In Table~\ref{classifier_prediction_table}, we report the accuracy (ACC) and the area under the ROC curve (AUC) of each classifier on each dataset.
The results indicate that XGBoost, which we use in the main paper, yields higher average ACC and AUC than the other classifiers for all feature sets except for HP7, where XGBoost yields the highest average ACC but the second-highest average AUC.}

\begin{table*}[t]
	\begin{center}
		\caption{\label{classifier_node_classification_table}
		\black{Accuracy (ACC) and average area under the ROC curve (AVG AUC) for node classification using three types of ego-networks and five classifiers. For each feature set, the best and second-best results are in \textbf{bold} and \underline{underlined}, respectively.
  }}
        \scalebox{0.59}{
            \subtable[NP26]{
        		\begin{tabular}{c|c|ccc} 
        				\toprule
        				& \textbf{Classifier} & \textbf{Star} & \textbf{Radial} & \textbf{Contracted}\\
        				\midrule
        				\multirow{5}{*}{\textbf{ACC}} & \textbf{Logistic Regression} & 0.327 $\pm$ 0.000 & 0.200 $\pm$ 0.000 & 0.236 $\pm$ 0.000\\
                        & \textbf{Decision Tree} & 0.475 $\pm$ 0.013 & 0.654 $\pm$ 0.009 & 0.540 $\pm$ 0.009\\
                         & \textbf{Random Forest} & 0.488 $\pm$ 0.010 & 0.570 $\pm$ 0.015 & 0.496 $\pm$ 0.013\\
                        & \textbf{MLP} & 0.383 $\pm$ 0.057 & 0.569 $\pm$ 0.035 & 0.474 $\pm$ 0.086\\
                        & \textbf{XGBoost} & 0.555 $\pm$ 0.000 & \textbf{0.682 $\pm$ 0.000} & \underline{0.618 $\pm$ 0.000}\\
                        \midrule
                        \multirow{5}{*}{\textbf{AVG AUC}} & \textbf{Logistic Regression} & 0.827 $\pm$ 0.000 & 0.765 $\pm$ 0.000 & 0.727 $\pm$ 0.000\\
                        & \textbf{Decision Tree} & 0.726 $\pm$ 0.007 & 0.810 $\pm$ 0.005 & 0.747 $\pm$ 0.005\\
                        & \textbf{Random Forest} & 0.890 $\pm$ 0.001 & 0.912 $\pm$ 0.002 & 0.909 $\pm$ 0.002\\
                        & \textbf{MLP} & 0.768 $\pm$ 0.034 & 0.836 $\pm$ 0.023 & 0.778 $\pm$ 0.056\\
                        & \textbf{XGBoost} & 0.919 $\pm$ 0.000 & \textbf{0.952 $\pm$ 0.000} & \underline{0.942 $\pm$ 0.000}\\
                        \bottomrule
        		\end{tabular}
          }
        }
        \scalebox{0.59}{
            \subtable[NP7]{
        		\begin{tabular}{c|c|ccc} 
        				\toprule
        				& \textbf{Classifier} & \textbf{Star} & \textbf{Radial} & \textbf{Contracted}\\
        				\midrule
        				\multirow{5}{*}{\textbf{ACC}} & \textbf{Logistic Regression} & 0.309 $\pm$ 0.000 & 0.223 $\pm$ 0.000 & 0.241 $\pm$ 0.000\\
                        & \textbf{Decision Tree} & 0.469 $\pm$ 0.011 & 0.490 $\pm$ 0.009 & \underline{0.526 $\pm$ 0.013}\\
                         & \textbf{Random Forest} & 0.400 $\pm$ 0.015 & 0.501 $\pm$ 0.011 & 0.425 $\pm$ 0.008\\
                        & \textbf{MLP} & 0.342 $\pm$ 0.028 & 0.440 $\pm$ 0.017 & 0.368 $\pm$ 0.079\\
                        & \textbf{XGBoost} & 0.523 $\pm$ 0.000 & \textbf{0.545 $\pm$ 0.000} & 0.482 $\pm$ 0.000\\
                        \midrule
                        \multirow{5}{*}{\textbf{AVG AUC}} & \textbf{Logistic Regression} & 0.835 $\pm$ 0.000 & 0.716 $\pm$ 0.000 & 0.731 $\pm$ 0.000\\
                        & \textbf{Decision Tree} & 0.720 $\pm$ 0.006 & 0.744 $\pm$ 0.005 & 0.760 $\pm$ 0.007\\
                        & \textbf{Random Forest} & 0.856 $\pm$ 0.001 & 0.873 $\pm$ 0.003 & 0.880 $\pm$ 0.002\\
                        & \textbf{MLP} & 0.734 $\pm$ 0.024 & 0.749 $\pm$ 0.012 & 0.699 $\pm$ 0.051\\
                        & \textbf{XGBoost} & 0.895 $\pm$ 0.000 & \underline{0.901 $\pm$ 0.000} & \textbf{0.910 $\pm$ 0.000}\\
                        \bottomrule
        		\end{tabular}
          }
        }
        \scalebox{0.59}{
            \subtable[TNP]{
        		\begin{tabular}{c|c|ccc} 
        				\toprule
        				& \textbf{Classifier} & \textbf{Star} & \textbf{Radial} & \textbf{Contracted}\\
        				\midrule
        				\multirow{5}{*}{\textbf{ACC}} & \textbf{Logistic Regression} & 0.432 $\pm$ 0.000 & 0.318 $\pm$ 0.000 & 0.268 $\pm$ 0.000\\
                        & \textbf{Decision Tree} & 0.585 $\pm$ 0.016 & 0.647 $\pm$ 0.011 & 0.593 $\pm$ 0.013\\
                         & \textbf{Random Forest} & 0.515 $\pm$ 0.014 & 0.598 $\pm$ 0.014 & 0.542 $\pm$ 0.028\\
                        & \textbf{MLP} & 0.502 $\pm$ 0.026 & 0.627 $\pm$ 0.038 & 0.480 $\pm$ 0.040\\
                        & \textbf{XGBoost} & 0.650 $\pm$ 0.000 & \textbf{0.723 $\pm$ 0.000} & \underline{0.673 $\pm$ 0.000}\\
                        \midrule
                        \multirow{5}{*}{\textbf{AVG AUC}} & \textbf{Logistic Regression} & 0.876 $\pm$ 0.000 & 0.838 $\pm$ 0.000 & 0.784 $\pm$ 0.000\\
                        & \textbf{Decision Tree} & 0.772 $\pm$ 0.009 & 0.806 $\pm$ 0.006 & 0.776 $\pm$ 0.007\\
                        & \textbf{Random Forest} & 0.915 $\pm$ 0.003 & 0.936 $\pm$ 0.002 & 0.933 $\pm$ 0.001\\
                        & \textbf{MLP} & 0.842 $\pm$ 0.012 & 0.875 $\pm$ 0.021 & 0.787 $\pm$ 0.025\\
                        & \textbf{XGBoost} & 0.946 $\pm$ 0.000 & \textbf{0.967 $\pm$ 0.000} & \underline{0.956 $\pm$ 0.000}\\
                        \bottomrule
        		\end{tabular}
          }
        }
        \scalebox{0.59}{
            \subtable[BASELINE]{
        		\begin{tabular}{c|c|ccc} 
        				\toprule
        				& \textbf{Classifier} & \textbf{Star} & \textbf{Radial} & \textbf{Contracted}\\
        				\midrule
        				\multirow{5}{*}{\textbf{ACC}} & \textbf{Logistic Regression} & 0.309 $\pm$ 0.000 & 0.505 $\pm$ 0.000 & 0.418 $\pm$ 0.000\\
                        & \textbf{Decision Tree} & 0.549 $\pm$ 0.013 & 0.649 $\pm$ 0.013 & 0.590 $\pm$ 0.015\\
                         & \textbf{Random Forest} & 0.597 $\pm$ 0.016 & 0.581 $\pm$ 0.012 & 0.538 $\pm$ 0.011\\
                        & \textbf{MLP} & 0.530 $\pm$ 0.034 & 0.542 $\pm$ 0.045 & 0.476 $\pm$ 0.061\\
                        & \textbf{XGBoost} & 0.550 $\pm$ 0.000 & \underline{0.659 $\pm$ 0.000} & \textbf{0.668 $\pm$ 0.000}\\
                        \midrule
                        \multirow{5}{*}{\textbf{AVG AUC}} & \textbf{Logistic Regression} & 0.848 $\pm$ 0.000 & 0.904 $\pm$ 0.000 & 0.871 $\pm$ 0.000\\
                        & \textbf{Decision Tree} & 0.762 $\pm$ 0.007 & 0.807 $\pm$ 0.007 & 0.774 $\pm$ 0.008\\
                        & \textbf{Random Forest} & 0.920 $\pm$ 0.002 & 0.922 $\pm$ 0.001 & 0.930 $\pm$ 0.001\\
                        & \textbf{MLP} & 0.910 $\pm$ 0.007 & 0.922 $\pm$ 0.016 & 0.898 $\pm$ 0.026\\
                        & \textbf{XGBoost} & 0.922 $\pm$ 0.000 & \textbf{0.951 $\pm$ 0.000} & \underline{0.950 $\pm$ 0.000}\\
                        \bottomrule
        		\end{tabular}
          }
        }
	\end{center}
\end{table*}

\begin{table*}[t]
	\begin{center}
	\caption{\label{variant_3h_motif_table}
		\black{Comparison of \tmotifs and its variants. For MR($p$), HR($p$, mean), HR($p$, max), and HR($p$, min), the min, average, and max performances over $p=\{0.1,0.2,\cdots, 0.9\}$ are reported.  For each task, the best and second-best results are in \textbf{bold} and \underline{underlined}, respectively. Surprisingly, despite the simplicity of \tmotifs, using it consistently outperforms or achieves similar results compared to other variants across all tasks, providing evidence of its effectiveness.}}
    \scalebox{0.70}{
        	\begin{tabular}{c|ccc|ccc|ccc|ccc|ccc}
                    \toprule
                    \multirow{3}{*}{} & \multicolumn{3}{c|}{\textbf{Hypergraph Clustering}} & \multicolumn{6}{c|}{\textbf{Hyperedge Prediction} (coauth-DBLP)} & \multicolumn{6}{c}{\textbf{Node Classification}}\\
                    \cmidrule{2-16}
                    & \multicolumn{3}{c|}{\textbf{NMI score}} & \multicolumn{3}{c|}{\textbf{ACC}} & \multicolumn{3}{c|}{\textbf{AUC}} & \multicolumn{3}{c|}{\textbf{ACC}} & \multicolumn{3}{c}{\textbf{AVG AUC}}\\
                    \cmidrule{2-16}
                    & \textbf{min} & \textbf{mean} & \textbf{max} & \textbf{min} & \textbf{mean} & \textbf{max} & \textbf{min} & \textbf{mean} & \textbf{max} & \textbf{min} & \textbf{mean} & \textbf{max} & \textbf{min} & \textbf{mean} & \textbf{max}\\
        			\midrule
        			\multirow{1}{*}{\textbf{\Tmotifs (i.e., Abs($1$))}} &  & \textbf{1.000} &  &  & \textbf{0.836} &  &  & \underline{0.909} &  &  & \textbf{0.723} &  &  & \textbf{0.967} &\\
                    \midrule
                    \multirow{1}{*}{\textbf{Abs($2$)}} &  & \textbf{1.000} &  &  & \underline{0.830} &  &  & 0.908 &  &  & 0.673 &  &  & 0.955 & \\
                    \midrule
                    \multirow{1}{*}{\textbf{MR($p$)}} & 0.824 & 0.932 & \textbf{1.000} & 0.802 & 0.814 & 0.829 & 0.886 & 0.895 & 0.902 & 0.677 & 0.690 & 0.700 & 0.951 & 0.956 & 0.961\\
                    \midrule
                    \multirow{1}{*}{\textbf{HR($p$, mean)}} & 0.877 & \underline{0.965} & \textbf{1.000} & 0.804 & 0.819 & 0.829 & 0.887 & 0.899 & \textbf{0.910} & 0.664 & 0.681 & 0.695 & 0.952 & 0.955 & 0.960\\
                     \midrule
                    \multirow{1}{*}{\textbf{HR($p$, max)}} & 0.877 & 0.934 & \textbf{1.000} & 0.805 & 0.820 & \underline{0.830} & 0.890 & 0.900 & \textbf{0.910} & 0.682 & 0.695 & \underline{0.709} & 0.951 & 0.957 & \underline{0.963} \\
                    \midrule
                    \multirow{1}{*}{\textbf{HR($p$, min)}} & 0.877 & 0.934 & \textbf{1.000} & 0.805 & 0.818 & 0.826  & 0.890 & 0.898 & 0.907 & 0.673 & 0.694 & \textbf{0.723}  & 0.954 & 0.957 & 0.962\\
                    \bottomrule
        	\end{tabular}
        }
	\end{center}
\end{table*}

\section{Node Classification}
\label{appendix:node_classification}

\black{In this section, we provide detailed information about the node classification task in Section~\ref{sec:exp:observations:prediction} and additional experimental results.}

\smallsection{\black{Settings:}}
\black{We randomly sample $100$ nodes from each of the $11$ real-world hypergraphs, ensuring that the corresponding radial ego-networks contain between $10$ and $13,500$ hyperedges. Note that the tags-math-sx dataset has the largest radial ego-networks, with a mean of around $13,500$ hyperedges.
We split the selected 1,100 nodes into a training set (80\%) and a test set (20\%) while ensuring that each label has an equal number of nodes in both sets.
}


\smallsection{\black{Additional Experimental Results:}}
\black{
We evaluate the performance of the four considered feature sets in node classification using all combinations of three ego-network types (spec., star, radial, and contracted) and five classifiers (spec., Logistic Regression, Decision Tree, Random Forest, MLP, and XGBoost).
We use the implementation of all classifiers provided by scikit-learn~\cite{scikit-learn} with default hyperparameters.
In Table~\ref{classifier_node_classification_table}, we report the accuracy (ACC) and the average area under the ROC curve (AVG AUC) for each combination. 
It should be noticed that the combination of radial ego-networks and XGBoost, which we use in the main paper, performs well for all feature sets.
Specifically, the combination yields higher ACC and AVG AUC than the other combinations for NP26 and TNP. 
Additionally, for NP7, the combination yields the highest ACC but the second highest AVG AUC. For BASELINE, the combination yields the highest AVG AUC but the second highest ACC.
}


\section{\black{Details and Variants of \Tmotifs}}
\label{appendix:variants}

\black{In this section, we explore potential variations of \tmotifs and assess their performance through experiments. We aim to show that the proposed \tmotifs strike a balance between simplicity and strong characterization power, making them an effective tool for analyzing hypergraph structures. Visual representations of all 431 \tmotifs are provided in Figure~\ref{all_3H-motifs_fig}.
We consider the following variants in which the two states of each non-empty set are defined differently:
\begin{itemize}
    \item {\bf Absolute - Abs($\theta$):} The states of each non-empty set are defined based on its cardinality $c$ as (1) $0< c\leq \theta $ and (2) $c>\theta$, where $\theta=1$. Note that \tmotifs are equivalent to Abs($1$).
    \\
    \item {\bf Motif Ratio - MR($p$):} 
    The states of each non-empty set are defined based on its cardinality $c$ and the total number of unique nodes $n$ in the considered instance as (1) $0< \frac{c}{n} \leq p$ and (2) $\frac{c}{n}>p$, where $0<p<1$.   
    \\
    \item {\bf Hyperedge Ratio - HR($p,\sigma$):}
    The states of each non-empty set are defined based on its cardinality $c$ and the hyperedges $S$ containing the set as (1) $0< \sigma((\frac{c}{|e|})_{e\in S}) \leq p$ and (2) $\sigma((\frac{c}{|e|})_{e\in S})>p$, where $0<p<1$ and $\sigma(\cdot)$ is the \textit{mean}, \textit{max}, or \textit{min} function.   
    \\
\end{itemize}
We conduct experiments to compare the effectiveness of \Tmotifs (i.e., Abs($1$)), Abs($2$),  MR($p$), HR($p$, mean), HR($p$, max), HR($p$, min) for $p=\{0.1,0.2,\cdots, 0.9\}$. We evaluate the variants through hypergraph clustering, hyperedge prediction, and node classification using the same experimental settings described in Sections~\ref{sec:exp:characterization_power} and \ref{sec:exp:observations:prediction}.
For hyperedge prediction task is performed on the coauth-DBLP dataset.
}

\black{Table~\ref{variant_3h_motif_table} presents the results of comparing the performances. For MR($p$), HR($p$, mean), HR($p$, max), and HR($p$, min), the table reports the min, average, and max performances over $p=\{0.1,0.2,\cdots, 0.9\}$.
Remarkably, even though the definition of \tmotifs is simplest, the use of \tmotifs leads to the best performance in hypergraph clustering and node classification tasks. In hyperedge prediction, using \tmotifs achieves the highest ACC and the second-highest AUC. These results provide support for the effectiveness of \tmotifs.}

\blue{We further investigate a wider range of $\theta$ value (spec., from $1$ to $15$) in Abs($\theta$)  to assess its impact on performance in the considered tasks. As shown in Figure~\ref{fig:performance_thres}, our empirical findings reveal that employing a threshold of $1$ (i.e., utilizing \tmotifs) consistently results in the highest performance across the tasks.
Especially, as shown in Figure~\ref{fig:performance_thres}(d), with $\theta$ fixed at $1$, the hyperedge prediction performance remains very close to that achieved with the optimal $\theta$ value for each dataset.
In Appendix~\ref{appendix:further_analysis}, we provide additional analysis regarding the effectiveness of employing $\theta=1$.
}

\begin{figure*}[t]
	\centering     
	\subfigure[\label{hypergraph_clustering_thres}\blue{Hypergraph clustering performance (NMI) w.r.t. $\theta$}]{\includegraphics[width=0.22\textwidth]{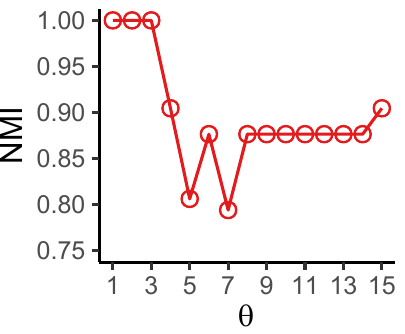}}
        \hspace{10pt}
        \subfigure[\label{hyperedge_prediction_thres}\blue{Average hyperedge prediction accuracy w.r.t. $\theta$}]{\includegraphics[width=0.22\textwidth]{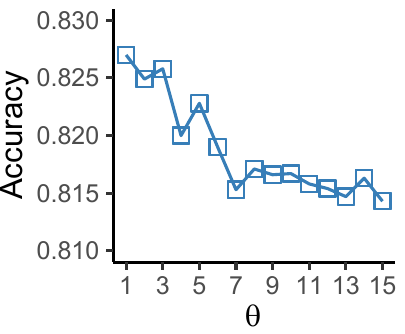}}
        \hspace{10pt}
        \subfigure[\label{node_classification_thres} \blue{Node classification accuracy w.r.t. $\theta$}]{\includegraphics[width=0.22\textwidth]{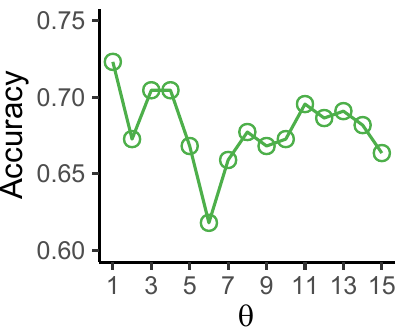}}
        \hspace{10pt}
        \subfigure[\label{1_vs_best_theta} \blue{Average hyperedge prediction accuracy (optimal $\theta$ vs fixed $\theta$)}]
        {\includegraphics[width=0.22\textwidth]{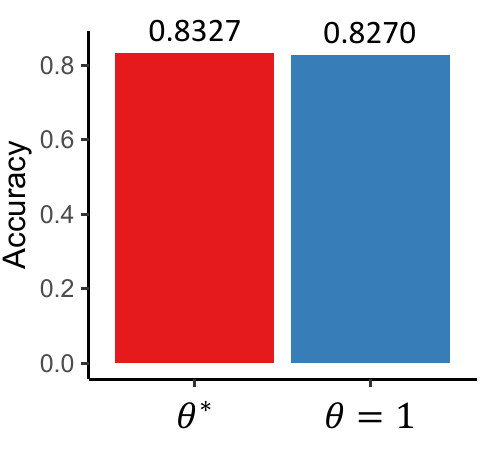}} \\
        \vspace{-3mm}
	\caption{\label{fig:performance_thres} 
		\blue{\Tmotifs (i.e., Abs($1$)) consistently demonstrates superior performance in all three applications compared to Abs($\theta$) with other $\theta$ values.
        In (b), we report the hyperedge prediction accuracy averaged over all datasets.
        In (d), $\theta^*$ (red bar) represents the average hyperedge prediction accuracy achieved with the optimal $\theta$ value for each dataset, while $\theta=1$ (blue bar) represents the average hyperedge prediction accuracy with $\theta$ fixed at $1$. Remarkably, even with $\theta$ fixed at $1$, the performance remains very close to the highest accuracy attainable.}
  }
	
\end{figure*}

\begin{figure*}[t]
	\centering     
	\includegraphics[width=0.23\textwidth]{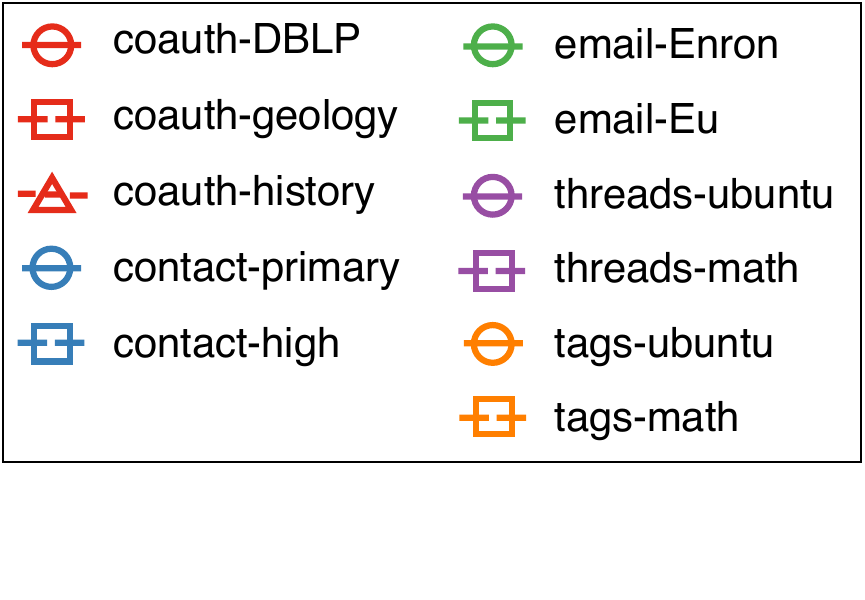}
        \hspace{20pt}
        \subfigure[\label{fig:cardinality_dist} \blue{Subset Cardinality Distribution}]{\includegraphics[width=0.235\textwidth]{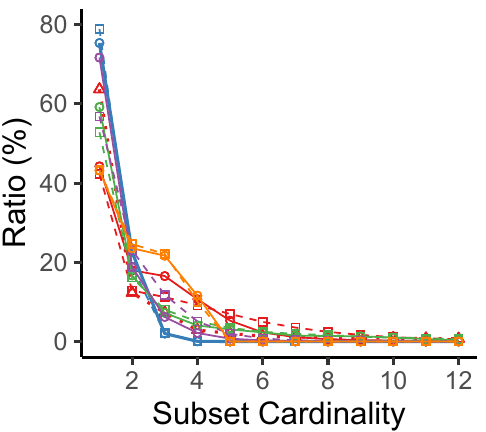}}
        \hspace{10pt}
	\subfigure[\label{fig:IG} \blue{Amount of Extra Information w.r.t. $\theta$}]{\includegraphics[width=0.235\textwidth]{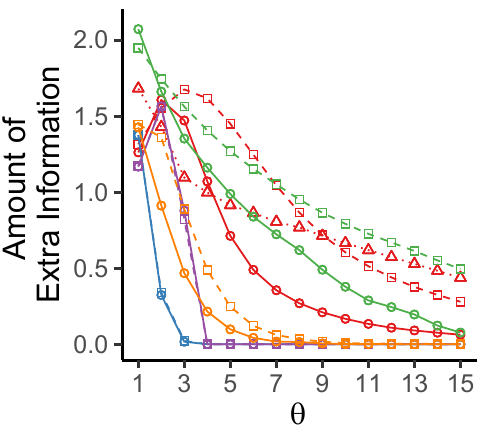}}
	\caption{\label{dist_IG} \blue{(a) The cardinalities of subsets, upon which both \motifs and \tmotifs are defined, exhibit a strong bias toward the value $1$ in all datasets. (b) The amount of extra information gained with \tmotifs over \motifs is maximized when utilizing the cardinality threshold $\theta=1$, as employed by \tmotifs, in 7 out of 11 hypergraphs.}}
\end{figure*}

\section{Data Analysis Regarding the Effectiveness of \Tmotifs}
\label{appendix:further_analysis}

\blue{In this section, we present data analyses for gaining insights into the effectiveness of \tmotifs, which is demonstrated by comprehensive experiments in Appendix~\ref{appendix:variants}.
Recall that \tmotifs employs a threshold value of $\theta=1$ to categorize the states of the seven subsets into three distinct groups based on their cardinalities.
Our analyses below demonstrate why choosing $\theta=1$ is a suitable decision.
}


\smallsection{Subset Cardinality Distribution:}
\blue{We first investigate the cardinalities of the seven subsets, based on which \motifs and \tmotifs are defined, in the 11 real-world hypergraphs.
As shown in Figure~\ref{fig:cardinality_dist}, in all considered real-world hypergraphs, the count of subsets with each cardinality diminishes considerably as the cardinality increases.
This observation suggests that when employing a threshold value $\theta$ greater than $1$, only a very small fraction of subsets falls into the state $>\theta$, which is a new state introduced in \tmotifs but not present in \motifs.
Intuitively, this limited utilization of this new state leads to a reduction in the amount of additional information provided by \tmotifs in comparison to \motifs.}

\smallsection{Amount of Extra Information:}
\blue{To numerically validate our intuition above, we measure the amount of extra information in \tmotifs over \motifs using the \textit{conditional entropy} as follows:
\begin{equation}
H(\text{\tmotif}|\text{\motif})=-\sum_{i=1}^{26}\frac{N_i}{N_{tot}}\sum_{j\in K_i}\frac{n_j}{N_i}\ln{\frac{n_j}{N_i}}, \label{eq:ig}
\end{equation}
where $N_{tot}$ is the count of all \motif instances, $N_i$ is the count of the instances of \motif $i$, $K_i$ is the set of the indices of \tmotifs that can share instances with \motif $i$ (e.g. $K_1=\{\text{1,\dots,6}\}$, as shown in Figure~\ref{fig:tmotif_example}), and $n_{j}$ is the count of the instances of \tmotif $j$.
If we treat the \motif and \tmotif corresponding to each instance as random variables, the conditional entropy measures the amount of information needed to describe the corresponding \tmotif given the corresponding \motif.
If the instances of each \motif are uniformly distributed among the corresponding \tmotifs, we can achieve the maximum amount of additional information, which is $\sum_{i=1}^{26}\frac{N_i}{N_{\text{tot}}}\ln{|K_i|}$. Conversely, if the instance counts of \motif $i$ are highly skewed toward one of the corresponding \tmotifs, the extra information would approach zero, indicating a minimum.
In Figure~\ref{fig:IG}, we report the amount of extra information in each real-world hypergraph while varying $\theta$ from $1$ to $15$. In all datasets, the most significant additional information is obtained with a small $\theta$, with 7 out of 11 hypergraphs yielding the highest amount at $\theta=1$. This potentially explains the superior performance observed in machine tasks when $\theta=1$ (see Figure~\ref{fig:performance_thres}).}

\begin{figure*}[t]
 	\centering
 	\subfigure[\small Node: coauth-DBLP]{\includegraphics[width=0.24\textwidth]{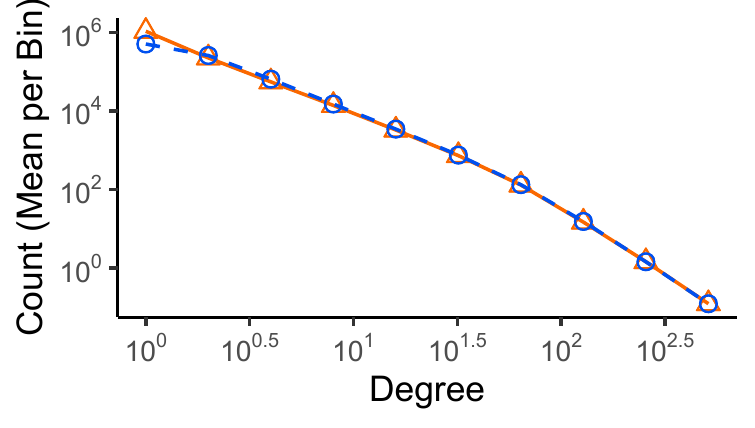}}
 	\subfigure[\small Hyperedge: coauth-DBLP]{\includegraphics[width=0.24\textwidth]{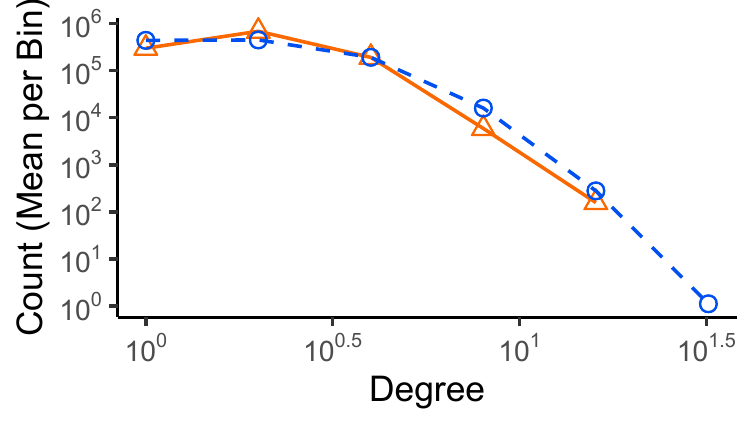}}
 	\subfigure[\small Node: coauth-geology]{\includegraphics[width=0.24\textwidth]{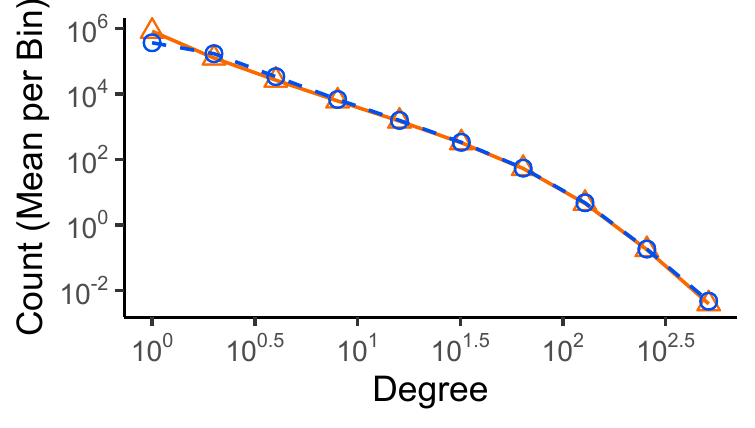}}
 	\subfigure[\small Hyperedge: coauth-geology]{\includegraphics[width=0.24\textwidth]{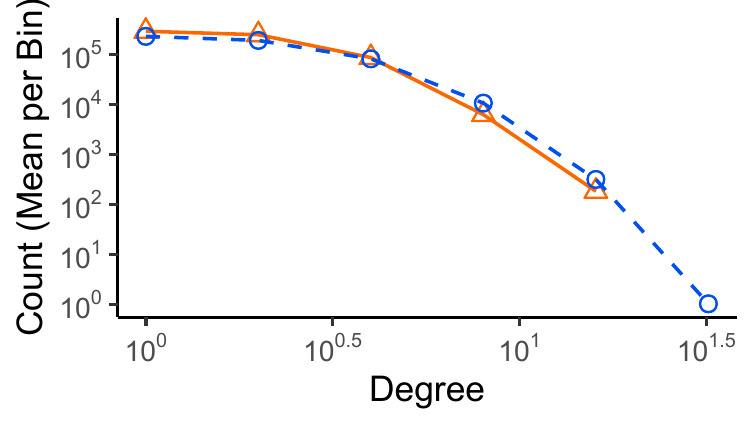}}\\
 	\subfigure[\small Node: coauth-history]{\includegraphics[width=0.24\textwidth]{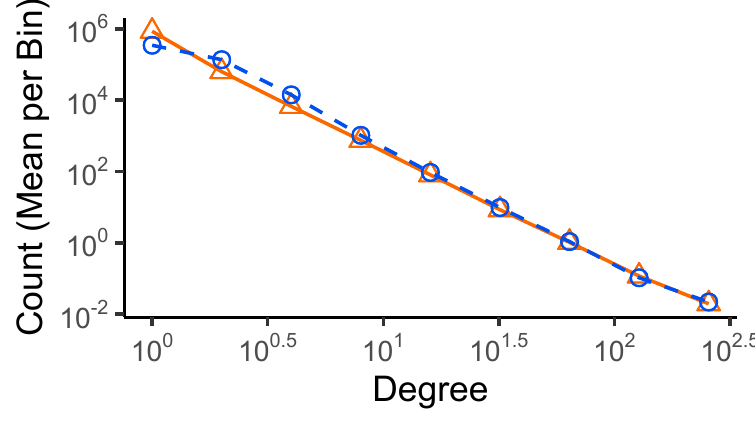}}
 	\subfigure[\small Hyperedge: coauth-history]{\includegraphics[width=0.24\textwidth]{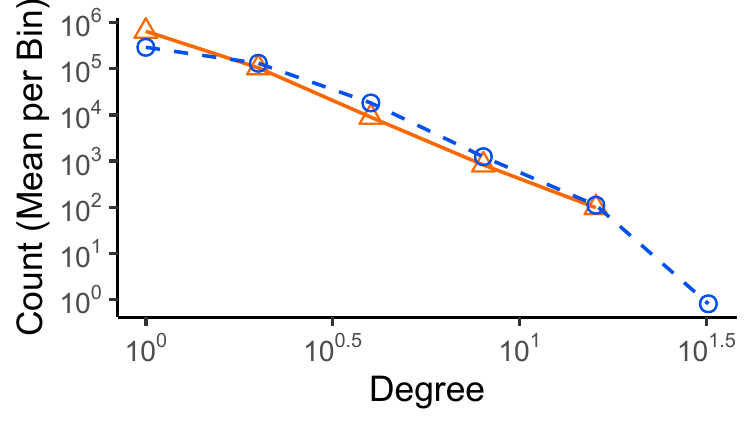}}
 	\subfigure[\small Node: contact-primary]{\includegraphics[width=0.24\textwidth]{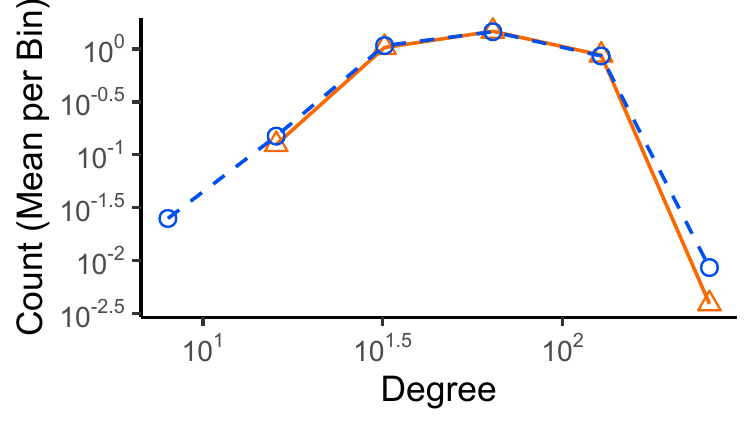}}
 	\subfigure[\small Hyperedge: contact-primary]{\includegraphics[width=0.24\textwidth]{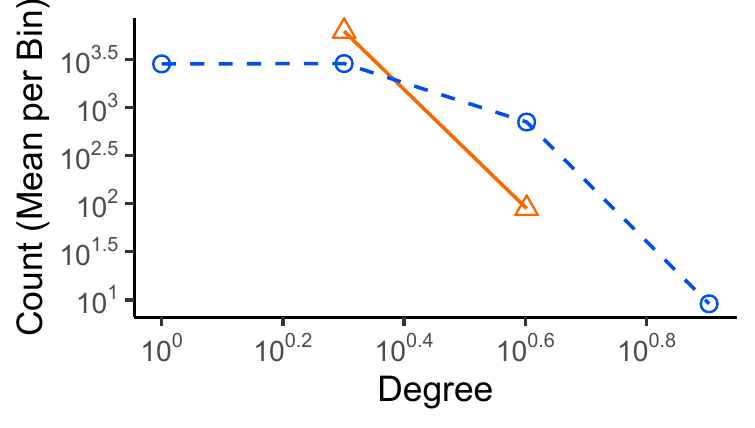}}\\
 	\subfigure[\small Node: contact-high]{\includegraphics[width=0.24\textwidth]{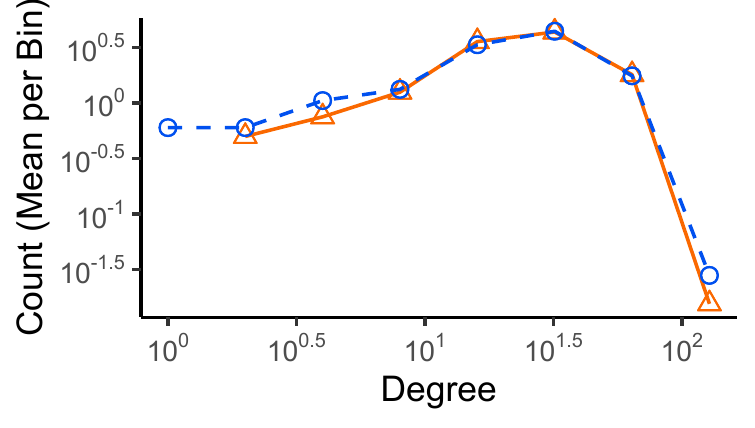}}
 	\subfigure[\small Hyperedge: contact-high]{\includegraphics[width=0.24\textwidth]{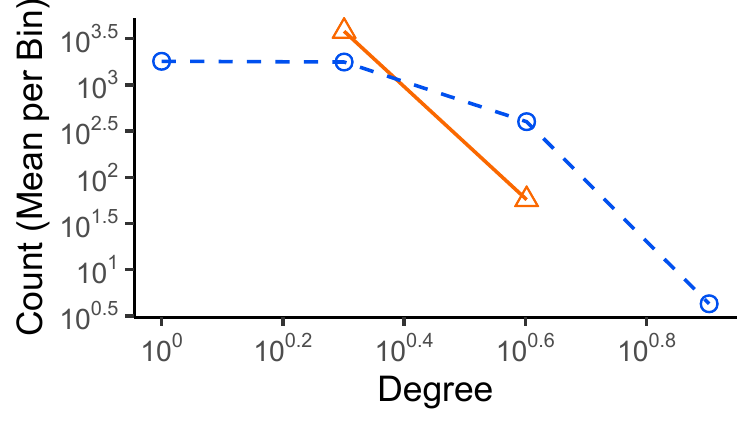}}
 	\subfigure[\small Node: email-Enron]{\includegraphics[width=0.24\textwidth]{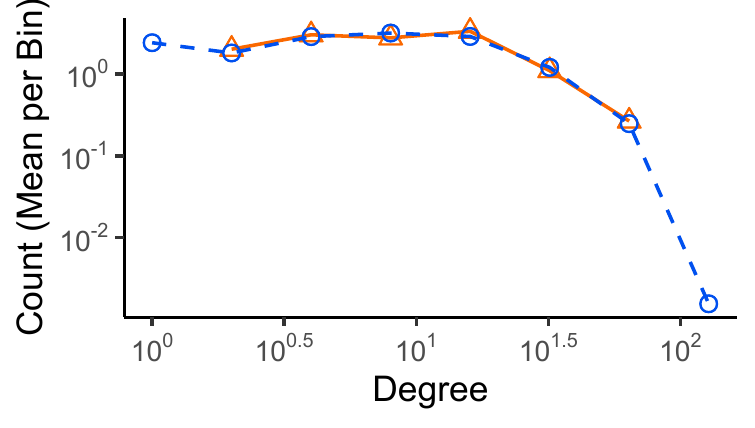}}
 	\subfigure[\small Hyperedge: email-Enron]{\includegraphics[width=0.24\textwidth]{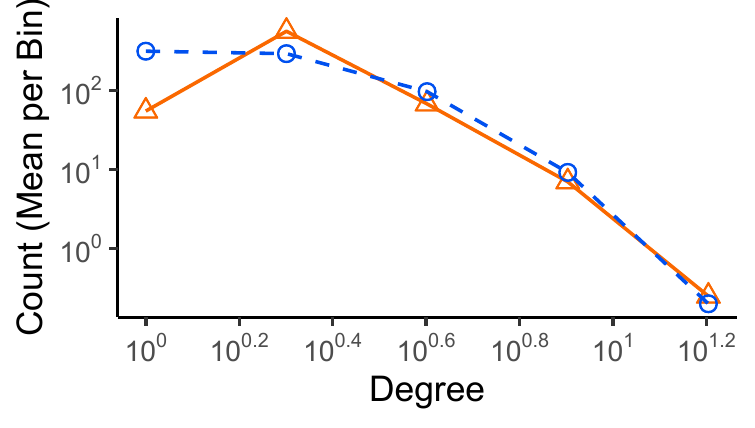}}\\
 	\subfigure[\small Node: email-Eu]{\includegraphics[width=0.24\textwidth]{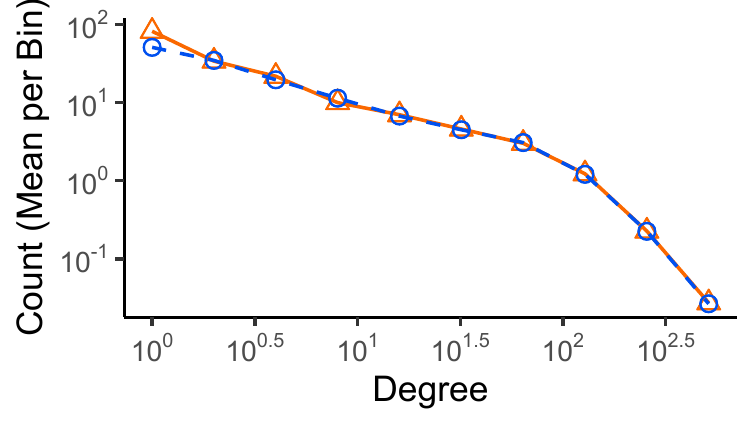}}
 	\subfigure[\small Hyperedge: email-Eu]{\includegraphics[width=0.24\textwidth]{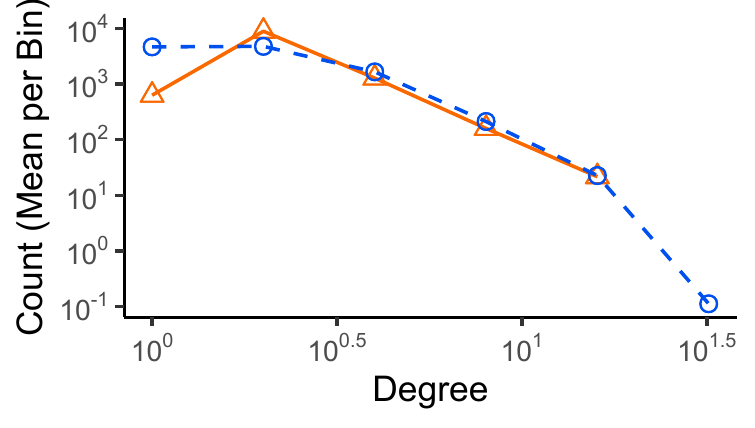}}
 	\subfigure[\small Node: tags-ubuntu]{\includegraphics[width=0.24\textwidth]{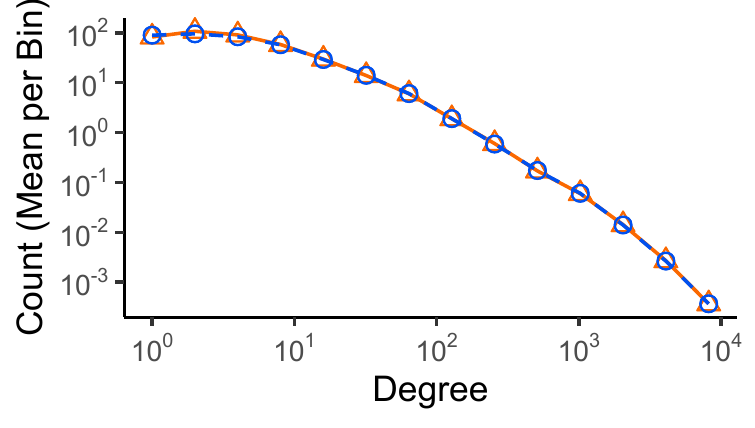}}
 	\subfigure[\small Hyperedge: tags-ubuntu]{\includegraphics[width=0.24\textwidth]{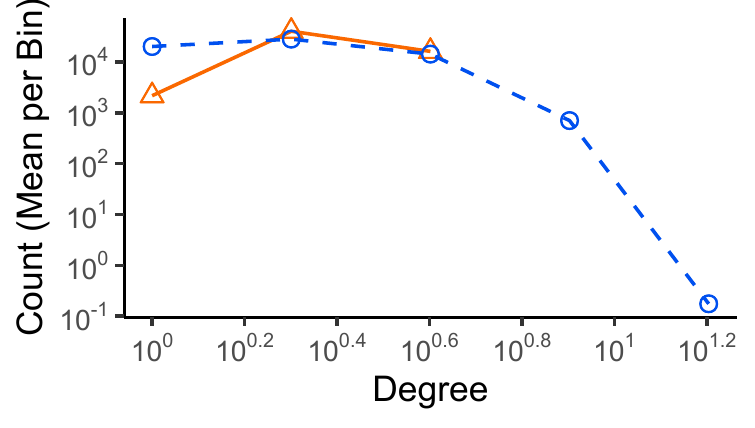}}\\
 	\subfigure[\small Node: tags-math]{\includegraphics[width=0.24\textwidth]{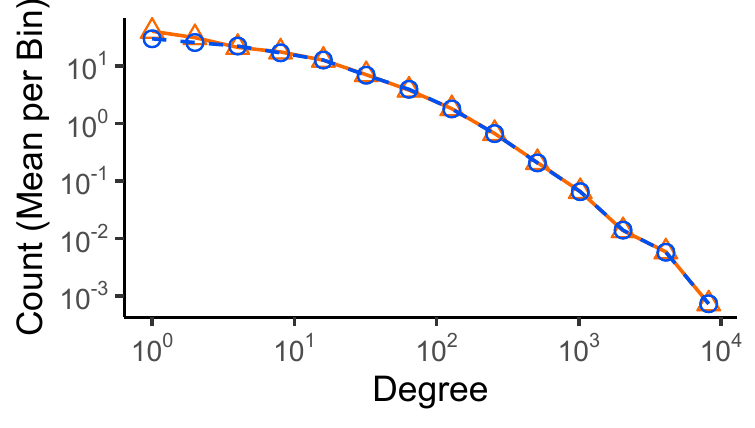}}
 	\subfigure[\small Hyperedge: tags-math]{\includegraphics[width=0.24\textwidth]{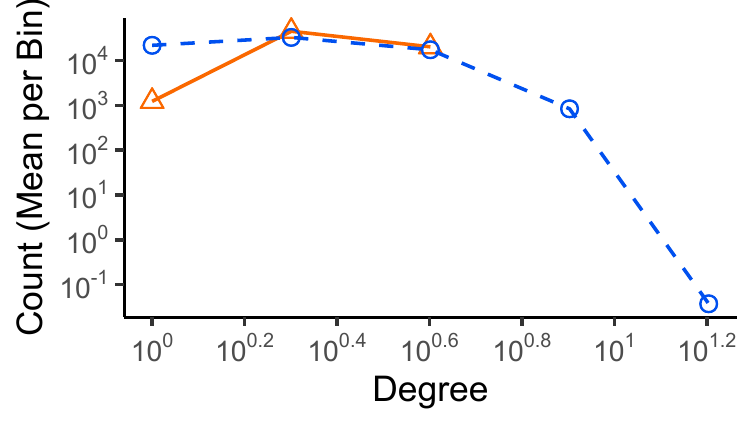}}
 	\subfigure[\small Node: threads-ubuntu]{\includegraphics[width=0.24\textwidth]{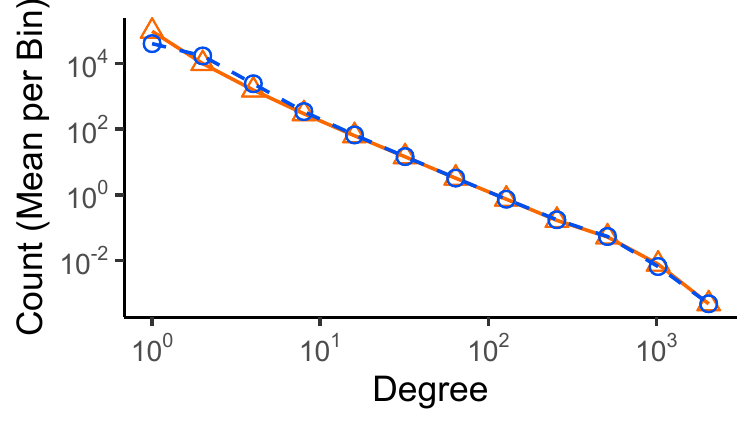}}
 	\subfigure[\small Hyperedge: threads-ubuntu]{\includegraphics[width=0.24\textwidth]{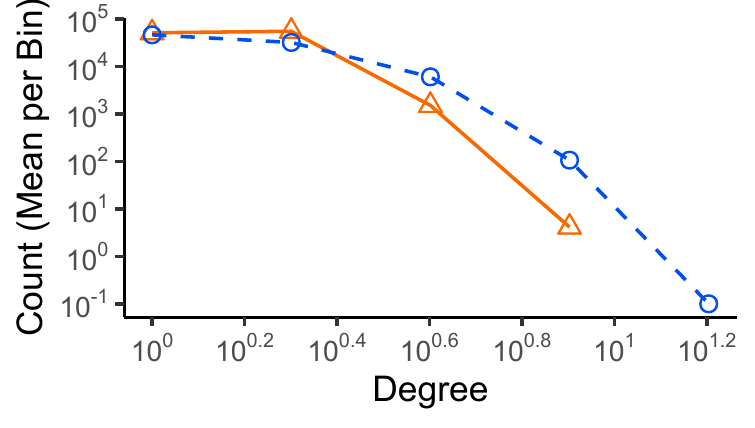}}\\
 	\subfigure[\small Node: threads-math]{\includegraphics[width=0.24\textwidth]{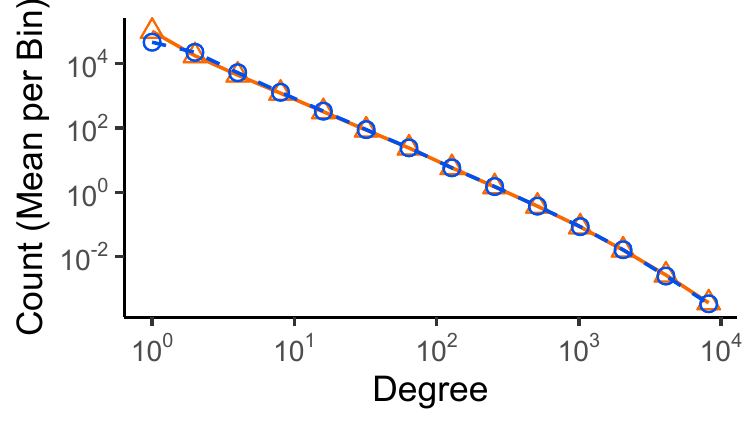}}
 	\subfigure[\small Hyperedge: threads-math]{\includegraphics[width=0.24\textwidth]{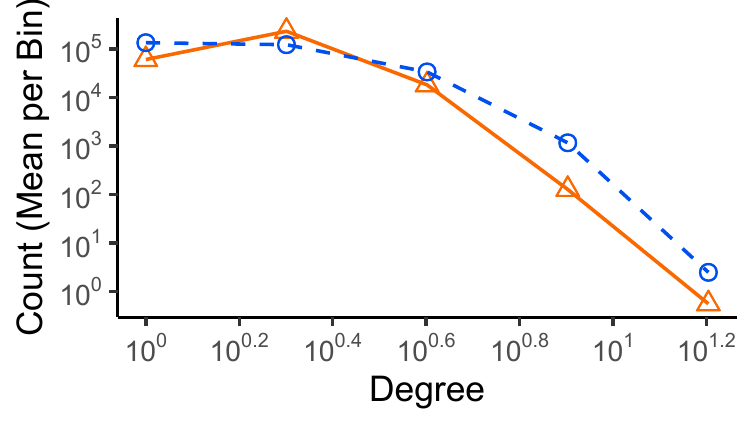}}\\
 	\includegraphics[width=0.4\textwidth]{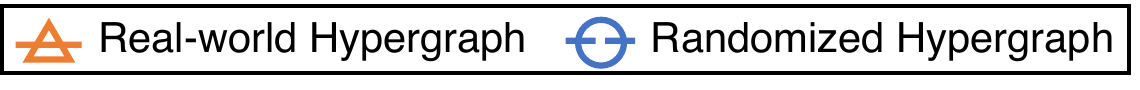}
 	\caption{Degree distributions of nodes and size distributions of hyperedges in real-world hypergraphs and the corresponding random hypergraphs.}\label{deg_fig}
 \end{figure*}

\begin{figure*}[t]
	\centering
	\includegraphics[width=\linewidth]{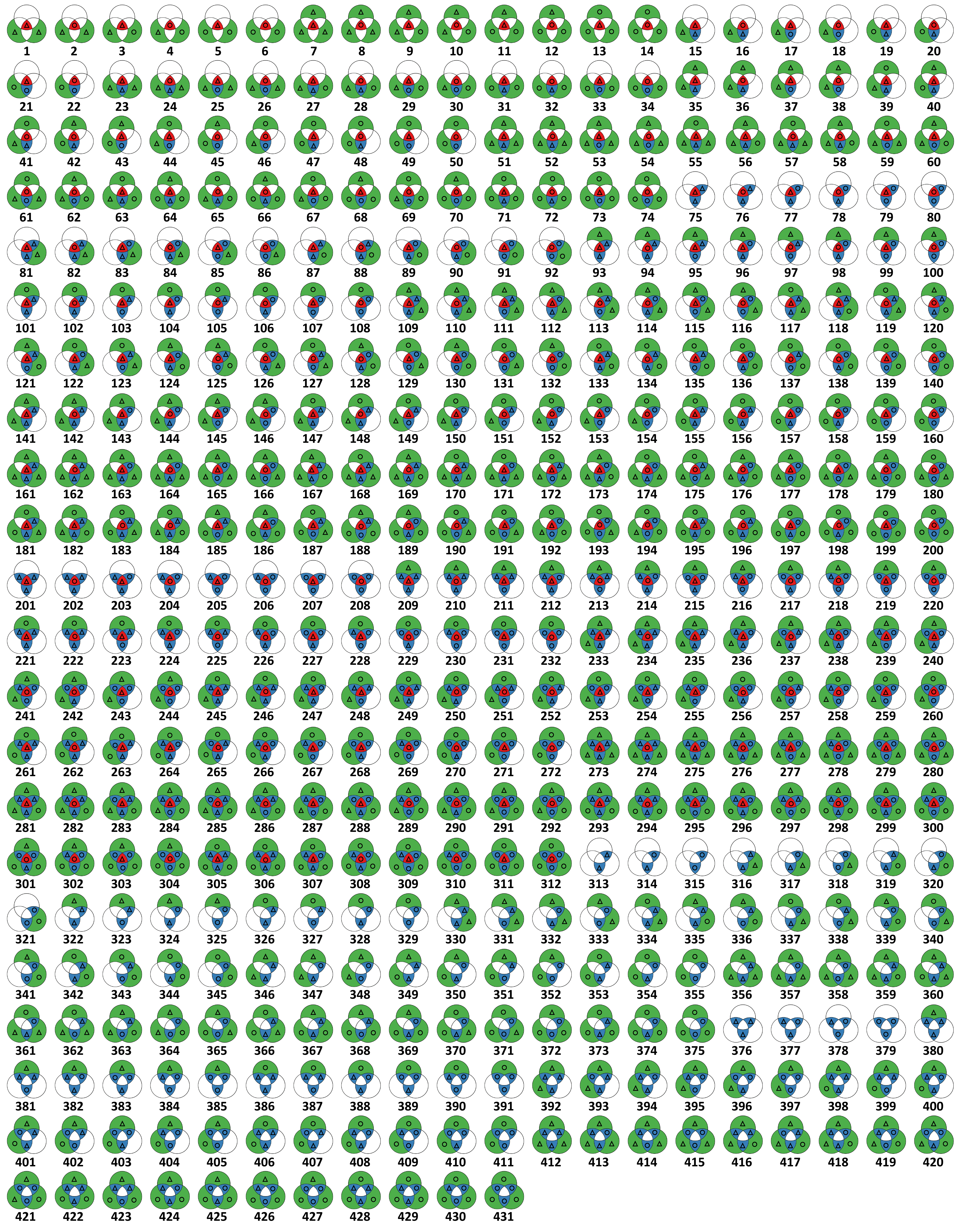} \\
        \vspace{-2mm}
	\caption{\black{\label{all_3H-motifs_fig} The 431 3H-motifs studied in this work. In each Venn diagram, uncolored regions are empty without containing any nodes. Colored regions with a triangle contain more than $0$ and at most $\theta$ nodes, while colored regions with a circle contain more than $\theta$ nodes.}}
\end{figure*}
\end{document}